\newcommand{\hi}{\mathcal{H}}

\newcommand{\hio}{\mathcal{H}_{1}}
\newcommand{\hit}{\mathcal{H}_{2}}
\newcommand{\hith}{\mathcal{H}_{3}}
\newcommand{\his}{\mathcal{H}_{\mathcal{S}}}
\newcommand{\hir}{\mathcal{H}_{\mathcal{R}}}

\newcommand{\Y}{\yen}

\newcommand{\Eff}{\mathcal{E}}
\newcommand{\ip}[2]{\left\langle\,#1\,|\,#2\,\right\rangle} %inner product
\newcommand{\ket}[1]{|#1\rangle} %ket

\newcommand{\bra}[1]{\langle#1|} %bra
\newcommand{\tr}[1]{\textrm{tr}\left[#1\right]} %trace
\newcommand{\R}{\mathcal{R}}
\newcommand*\colvec[3][]{
    \begin{pmatrix}\ifx\relax#1\relax\else#1\\\fi#2\\#3\end{pmatrix}
}
\renewcommand{\S}{\mathcal{S}}

\newcommand{\B}{\mathcal{B}}%generic observable
\newcommand{\E}{\mathsf{E}}%generic observable
\newcommand{\F}{\mathsf{F}}%generic observable
\newcommand{\G}{\mathsf{G}}%generic joint observable
\renewcommand{\P}{\mathsf{P}}%sharp observable
\renewcommand{\R}{\mathcal{R}}%quantum reference frame

\documentclass[a4paper, 12pt, openany]{book} %chose the paper size and font size. Openany ensures that all all chapters and similar may begin at any page, not only odd pages. For the introductory pages and appendices we want openany, but for chapter pages in the main content we want chapters to begin only on odd pages (right hand side). The book class ensures that the margins are automatically adjusted such that left hand pages are slightly moved to the left and vice versa at the right, which makes the thesis very readable and good looking when printed in bound book format.
\usepackage{physics}
\usepackage{soul}
\usepackage{relsize}
\usepackage{lmodern}
\usepackage{slantsc}
\usepackage{bbm}
\usepackage{graphicx}
\usepackage{amsmath}
\usepackage{bbold}
\usepackage{amsfonts}
\usepackage{amssymb}
\usepackage{amsthm}
\usepackage{amsbsy}
\usepackage{mathrsfs}
\usepackage{varioref}
\usepackage{dsfont}
\usepackage{bm}
\usepackage{color}
\usepackage[nopar]{lipsum}
\usepackage{lipsum}

\usepackage{ upgreek }

\usepackage{booktabs} % for much better looking tables
\usepackage{array} % for better arrays (eg matrices) in maths
\usepackage{paralist} % very flexible & customisable lists (eg. enumerate/itemize, etc.)
\usepackage{verbatim} % adds environment for commenting out blocks of text & for better verbatim
\edef\restoreparindent{\parindent=\the\parindent\relax}
\usepackage[parfill]{parskip}
\usepackage{tikz-cd}
\usepackage{color}
\usepackage[T1]{fontenc}
\newtheorem{theorem}{Theorem}[section] % 1st argument is your name for it
\newtheorem{proposition}[theorem]{Proposition}
\newtheorem{definition}[theorem]{Definition}
\newtheorem{example}[theorem]{Example}

\newcommand{\T}{\mathcal{T}}

\newcommand{\hirs}{\hir \otimes \his}
\usepackage{amsmath}
\usepackage[makeroom]{cancel}

\usepackage{enumitem}

\usepackage[utf8]{inputenc} %to manage special characters
\usepackage{fancyhdr} %to customize the headers
\usepackage[lmargin=1.5in, rmargin=1in, tmargin=1in, bmargin=1in]{geometry} %sets the margins for the pages
\usepackage{url} %to include urls
\usepackage{listings} %include this if you want to include code in the thesis
\usepackage{amsmath,amssymb} %mathematical package
\usepackage{array, booktabs} %to make better tables
\usepackage{graphicx} %to include graphics
\usepackage{float} %to include floats
\usepackage[export]{adjustbox} %to adjust floats
\usepackage{color} %edit e.g. text colors

\usepackage[backend = biber,
            style = numeric,
            date = long,     % Long: 24th Mar. 1997 | Short: 24/03/1997
            sorting = none,
            maxcitenames = 3,   % max names to include before et. al.
            ]{biblatex} %customize the look of your citations and bibliography
\usepackage{comment} %to be able to comment out sections in the .tex files
\usepackage{afterpage} %to customize page commands such as below
\newcommand\myemptypage{
    \null
    \thispagestyle{empty}
    \addtocounter{page}{-1}
    \newpage
    } %sets new page command to insert an empty page without adding to the page counter or having a page number

\usepackage{adjustbox}

\usepackage{hyperref}
\usepackage{cleveref}

\usepackage{pifont}

\begin{document}

%%%%%%%%%%%%%%%%%%%%%%%%%%%%%%%%%%%%%%%%%%%%%%%%%%%%%%%%
%\begin{comment}
% The title page:
% For NTNU students this page will be generated automatically when submitting your paper, and should not be included in the final file from Latex. Delete or comment out the title page setup. The final report should then start with the first page being the abstract. I have included a title page here so it is possible to see how it may look like, and for those who does not get an automatically generated title page. Of course you will need to change the names and titles etc. to your case.

%the title page should be an odd page (right hand side)

\begin{titlepage}
\newgeometry{left=1.6in, right=2in}
\vspace*{1.5cm}
\begin{center}
\noindent  \textcolor{gray}{\large Jan G{\l}owacki} \\
\vspace{1.5cm}

\noindent \textbf{\Large Operational Quantum Frames}
\vspace{0.5cm}

\noindent {\large An operational approach to quantum reference frames} \\

\vspace{2.5cm}
\noindent PhD Thesis \\
Supervisor: Marek Ku{\' s} \\
April 2023 \\

\vspace{0.7cm}
\noindent Polish Academy of Sciences \\
Center for Theoretical Physics \\
%Department of Physics \\

\vspace{4.5cm}
\begin{figure}[h]\centering
    \includegraphics[width=0.3\textwidth]{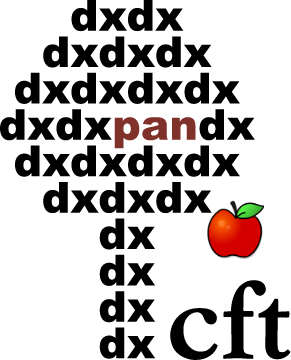}
\end{figure}
\end{center}

\end{titlepage}
\restoregeometry
\myemptypage %empty page such that the abstract starts at the first right hand side after the title page
%\end{comment}
%%%%%%%%%%%%%%%%%%%%%%%%%%%%%%%%%%%%%%%%%%%%%%%%%%%%%%%%

% The pre-chapters
%\frontmatter
\chapter*{Abstract} %pre-chapters should not be numbered, hence the "*"
\pagenumbering{roman} %introductory pages should be roman
\setcounter{page}{1}
%\addcontentsline{toc}{chapter}{\protect\numberline{}Abstract} %add the chapter to the table of contents, this is not automatically added when creating unnumbered chapters (*). Add it in a chapter style, and keep all chapters on the same numberline indent regardless of number or not on the chapter
The goal of this thesis is to provide a conceptual basis and general formal machinery for a \emph{relational} and \emph{operational} approach to the foundations of quantum theory. The framework 
builds on the ideas of so-called \emph{quantum reference frames} (QRF) and previous work on
\emph{quantum measurement theory} in the presence of symmetries. In a nutshell, the QRF program is based on the idea that reference frames should be treated as physical systems,
combined with an assumption that as such they should be modeled within quantum
mechanics. This perspective is aligned with insisting on relationality in physics, which
is understood as justifying the fact that observations are always made with respect to
some other, reference, system. Broadly speaking, physics should then be primarily concerned with relations between physical systems. In this work, we combine these insights
with the emphasis on operationality, understood as refraining from introducing into the
framework objects not directly related to in principle verifiable probabilities of measurement outcomes, and identifying the setups indistinguishable as such. Combining these insights
with intuitions from special relativity and gauge theory, we introduce an operational
notion of a quantum reference frame---which is defined as a \emph{quantum system equipped
with a covariant positive operator-valued measure} (POVM)---and build a framework
based on the concept of \emph{operational equivalence} that allows us to enforce operationality by quotienting the quantum state spaces with equivalence relation of indistinguishability by the available effects, assumed to be \emph{invariant under gauge
transformations}, and \emph{framed} in the sense of respecting the choice of the frame’s
POVM. Such effects are accessed via the \emph{yen} ($\Y^\R$) construction introduced in previous work, which maps effects on the system to those on the composite system, satisfying
gauge invariance and framing. Such effects are called \emph{relative}, and the classes of states
indistinguishable by them are referred to as relative states. These can be identified with
states on the system since they correspond to the image of the predual map $\Y^\R_*$.
We show that when the frame is \emph{localizable}, meaning that it allows for states that give
rise to a highly localized probability distribution of the frame’s observable, by restricting
the relative description upon such localized frame preparation we recover the usual,
non-relational formalism of quantum mechanics. We also provide a consistent way of
translating between different relative descriptions by means of frame-change maps,
and compare these with the corresponding notions in other approaches to QRFs, establishing
\emph{operational agreement} in the domain of common applicability.

\chapter*{Streszczenie}
 
Niniejsza rozprawa omawia podstawy koncepcyjne i wprowadza aparat matematyczny dla \emph{relacyjnego} i \emph{operacyjnego} podejścia do teorii kwantowej. Formalizm jest oparty na pojęciu \emph{kwantowego układu odniesienia} (Quantum Reference Frames -- QRF) i wcześniejszych rezultatach \emph{teorii pomiaru kwantowego} w obecności symetrii. W skrócie, program QRF party jest na założeniu, że układy odniesienia powinny być traktowane jako systemy fizyczne, w połączeniu z przekonaniem, że jako takie powinny być one modelowane w ramach mechaniki kwantowej. Perspektywa ta jest zgodna z postulatem relacyjności w fizyce, który mówi, że obserwacje zawsze są dokonywane w odniesieniu do jakiegoś ukłądu odniesienia. Szerzej rzecz ujmując, fizyka powinna zatem przede wszystkim zajmować się \emph{relacjami} między systemami fizycznymi. W tej pracy łączymy te spostrzeżenia z naciskiem na operacyjność, rozumianą jako powstrzymywanie się od wprowadzania do formalizmu obiektów niezwiązanych bezpośrednio z prawdopodobieństwami wyników co do zasady weryfikowalnych pomiarów, i co za tym idzie identyfikację opisów nierozróżnialnych jako takie. Dodając intuicje z Teorii Względności i Teorii Cechowania proponujemy operacyjną definicję kwantowego układu odniesienia - jako \emph{systemu kwantowego wyposażonego w kowariantną pozytywną miarę operatorową} (Positive Operator-Valued Measure -- POVM). Formalizm jest oparty na koncepcji \emph{ekwiwalencji operacyjnej}, która pozwala na wprowadzenie operacyjności w naszym rozumieniu poprzez zastosowanie relacji równoważności na przestrzeni stanów kwantowych ze względu na nierozróżnialność przez dostępne obserwable. Zakładamy, że są one \emph{niezmiennicze wobec transformacji cechowania}, rozumianych jako działanie elementu odpowiedniej grupy na układzie złożonym z systemu i układu referancyjnego, i \emph{referencyjne} w sensie poszanowania wyboru POVM układu odniesienia. Takie obserwable są dostępne za pośrednictwem konstrukcji \emph{yen} ($\Y^\R$), wprowadzonej w poprzednich pracach, która przekształca obserwable systemu na niezmiennicze względem transformacji cechowania i referencyjne obserwable układu złożonego. Takie efekty nazywane są \emph{względnymi}, a klasy stanów nierozróżnialnych przez nie \emph{stanami względnymi}. Mogą być one rozumiane jako stany systemu, ponieważ odpowiadają obrazowi mapy sprzężonej $\Y^\R_*$. Udowadniamy, że gdy układ odniesienia jest \emph{lokalizowalny}, co oznacza, że pozwala na stany prowadzące do silnie zlokalizowanych rozkładów prawdopodobieństwa obserwabli układu odniesienia, opis względem zlokalizowanego układu odniesienia odpowiada zwykłemu, nierelacyjnemu opisowi kwantowemu. Zapewniamy również spójny sposób tłumaczenia między różnymi opisami względnymi za pomocą \emph{map zmiany układów} i porównujemy je z odpowiadającymi pojęciami w innych podejściach do QRF - okazuje się, że procedury są \emph{operacyjnie nierozróżnialne} w obszarze wspólnej stosowalności.
 %insert the chapter text from the files

\chapter*{Preface}
%\addcontentsline{toc}{chapter}{\protect\numberline{}Preface} 
This thesis would not have been written had I never met Dr. Leon Loveridge. Our intense collaboration spanned over the last year, and was made possible by the ultimate freedom to follow my intuitions and engage in the research that I find worthy, granted by my Ph.D. supervisor Prof. Marek Ku{\'s}. It is not a summary of my research as a Ph.D candidate, that concerned areas as remote as categorical foundations for smooth geometry, the origin of the constraint's bracket of General Relativity, or reconstructions of quantum mechanics. This is because I feel the framework Leon was developing for the last decade, which I only discovered after meeting him at the Sejny Summer Institute workshop in 2021, approaches a complete form that deserves a fresh and perhaps more accessible presentation. I have to say that the progress we have made in understanding these ideas and developing the foundations and tools of the framework felt like a rapid development. Due to this intensity, part of me that is really glad that the process of writing the paper on operational quantum reference frame transformations, and this thesis, is almost behind me, while the other part can not wait to resolve the questions that remain open and complete the framework, which now seems within reach. Discovering Leon's universe of functional analysis employed in capturing relational paradigm for quantum physics was a revelation to me. It felt like I finally found myself in the realm of concepts and tools very natural to my taste and also powerful enough to allow turning my long-held intuitions into theorems that can be proved. This research influenced my thinking not only about operationality, relationality, and gauge principles but also about spatiotemporality, agency, space-time, and the emergence of gravity.

The work presented in this thesis stems from a project that seemed to be just a little note on frame changes based on the relativization construction when it was started, together with Titouan Carette, but then requested more and more care to fill in the gaps and forced a deeper understanding of the formalism as a whole. On the way, we were led to develop the concept of operational equivalence which is now of crucial importance and discover the class of framed effects that shed light on the relativization procedure on which the whole formalism was based and allowed for an operational understanding of the frame-change maps. As it stands now, besides filling the gaps along the lines already sketched, the algebraic and convex-theoretic generalizations of this formalism seem within reach, a whole world of possible applications waits to be analyzed, and novel perspectives on space-time emergence and understanding agency in quantum theory call to be explored, all within the broad operationally-relational paradigm that begins to stand on a firm formal foundation.

I need to confess, that entering the academic world as a job market made me treat the writing process a bit like leaving a testament. Hopefully, this is not the case -- I would miss $\Y$ a lot! It's been a demanding but beautiful journey.

\newpage
\chapter*{Acknowledgements}

Besides {\bf Dr. Leon Loveridge}, {\bf Prof. Marek Ku{\'s}} and {\bf Titouan Carette} that have already been mentioned, I wish to express here my gratitude towards the following human beings and collectives thereof:
\begin{itemize}
\vspace{5pt}
    \item[\ding{163}] my parents, {\bf Agnieszka} and {\bf Pawe\l}, who made me believe that I can follow my dreams and encouraged independence thought,\vspace{3pt}
    \item[\ding{163}] my grandparents, {\bf Augustyna} and {\bf Andrzej}, who supported me in various ways along my non-trivial educational path,\vspace{3pt}
    \item[\ding{163}] my first true supervisor, {\bf Prof. Klaas Landsman}, who encouraged me to study mathematics and stay bold in choosing my research problems,\vspace{3pt}
    \item[\ding{163}] my employer, {\bf Center for Theoretical Physics} that proved to be a place where an independent young researcher can develop,\vspace{3pt}
    \item[\ding{163}] my scientific community, {\bf Basic Research Community for Physics}, in which I found not only like-minded people and most important research collaborators but also true friends,\vspace{3pt}
    \item[\ding{163}] my friends in {\bf Centro Groupo Jobel}, artistic residence, and community, who provided me with ideal working conditions making the writing of this thesis as efficient and pain-free as possible,\vspace{3pt}
    \item[\ding{163}] entirety of my {\bf Friends}, which I can not enumerate, and who continue supporting me in my intellectual and other expeditions.
\end{itemize}

\tableofcontents
%\addcontentsline{toc}{chapter}{\protect\numberline{}Contents}

%add to table of contents list of figures and tables, and insert list of figures and tables
%\addcontentsline{toc}{chapter}{\protect\numberline{}\listfigurename}
%\listoffigures
%\addcontentsline{toc}{chapter}{\protect\numberline{}\listtablename}
%\listoftables

%\chapter*{Abbreviations}
%\addcontentsline{toc}{chapter}{\protect\numberline{}Abbreviations}
%\input{Chapters/02Abbreviations}
%\newpage
%\myemptypage
%add an empty non-counted page by the command below in order to get the first chapter on the left hand side, if needed (check your page number so that the first chapter is on an odd page)

%%%%%%%%%%%%%%%%%%%%%%%%%%%%%%%%%%%%%%%%%%%%%%%%%%%%%%%%
%Customize the layout of the main content of your thesis

\pagestyle{fancy} %set customized page style for header
\fancyhf{} %clear header and footer fields
\renewcommand{\headrulewidth}{0pt} %set to no rule
\fancyhead[LE, RO]{\thepage} %set the page number at left for even, right for odd pages
\fancyhead[RE, LO]{\leftmark} %set the chapter name at right for even, left for odd pages
%is is possible to design the header with the chapter as you wish, e.q. only the chapter or only the name, all lowercase instead etc.
%you could also design the footer if you wish, for example:
%\fancyfoot[LE, RO]{\thepage}
\setlength{\headheight}{14.49998pt} %set the header height

%%%%%%%%%%%%%%%%%%%%%%%%%%%%%%%%%%%%%%%%%%%%%%%%%%%%%%%%
%main content 

\mainmatter
\pagenumbering{arabic}
\chapter{Introduction}
The purpose of theoretical physics, as we understand it, is, or maybe should be, to provide coherent narratives for `how the physical world works' or, more modestly, for the ways we interact with it, accompanied by mathematically sound formalisms powerful enough to phrase and analyze measurement scenarios and derive predictions non-conflicting the experimental results. One such narrative, lacking to date both a firm mathematical foundation and verifiable experimental predictions not to be achieved otherwise, is that of the quantum reference frames (QRF) program. In short, it amounts to intersecting Einstein's view of reference frames being actual physical systems with the universality of quantum mechanics, thus forcing the reference frames to be modeled by quantum mechanics. Perhaps the main hopes driving the development of this research idea are that it may let us achieve the following:
\begin{itemize}
    \item resolve the problem of the arbitrariness of the Heisenberg's cut,
    \item make justice to the fact that observable quantities are relational,
    \item set the stage for the resolution of the problem of quantum gravity.
    \end{itemize}
Besides that, we believe that a properly developed QRF formalism will also point~to:
\begin{itemize}
    \item alternative to the troubled quantum field theoretic paradigm for foundations of relativistic quantum physics and beyond,
    \item deeper understanding of the interplay of quantum mechanics and relativity,
    \item novel reconstructions of quantum formalism from plausible physical principles,
    \item feasible approaches to space-time emergence.
\end{itemize}

We briefly comment on these hopes now, addressing them more formally either in the course of the presentation or in the discussion chapter \ref{ch:discuss}.

One may argue that the textbook quantum mechanics is a theory of quantum systems interacting with classical measuring instruments, thus making the relation between the classical and quantum theories somewhat obscure -- we would like to think about classical physics as being approximated by the more fundamental quantum formalism, while we still require the classical world for phrasing quantum mechanics. This unsatisfactory state of affairs was famously argued to be unavoidable by Bohr on the grounds of operationality understood as the need for classical communication of measurement results to assure the intersubjectivity of the theory. The QRF programs emphasize the need for treating all the systems taking part in deriving experimental results as quantum, contesting Heisenberg's cut and thus clarifying the desired ordering in the space of theories with quantum formalism being strictly more fundamental than the classical one. The strong form of operationality requested by Bohr is then lost, but we argue that its weaker, in our view more reasonable, form can be reconciled with the QRF ideas -- this is the subject of this thesis.

In QRF frameworks, the observable quantities are always understood as relational -- relative to a (quantum) reference frame, thus dealing with the relationality of observation explicitly, at least on the declarative/interpretational level.

The hope that such a way of thinking may help us understand gravity in the quantum realm stems from the realization that, when viewed relationally, space-time should rather be thought of as a collection of relative distances and time differences, and not as a background manifold with varying geometry. Thus taking the frames of reference that give rise to relational spatiotemporal quantities into the quantum realm may allow for the fully quantum picture of these notions, which we use to capture gravitational interactions in our best classical theories. Space-time as a, possibly dynamic, \emph{background} is neither operational nor relational.

Due to the lack of its rigorous non-perturbative implementation, the framework of quantum field theory (QFT), besides its tremendous success in providing correct predictions for the outcomes of experiments carried out in the quantum and relativistic regimes, is not fully satisfactory as a fundamental formalism for physics. When pushed further to the realm of theories of gravitation the situation gets only worse, and since this state of affairs persists over many decades now, it is perfectly reasonable to search for alternative paradigms and mathematical tools for relativistic quantum physics.

Quantum theory of systems embedded in space-time, as we normally view them, seems not to be independent of the features of this space-time. This manifests itself e.g. in the interplay of the allowed communication tasks and correlations with the group of symmetries of space (see e.g. \cite{Hhn2016,2022arXiv221014811J}).

Understanding such interplay may be useful for reconstructing quantum mechanics -- perhaps we should try reconstructing a relational framework, requiring compatibility with the spatiotemporal relations as we know them. Such a restatement of the problem may provide the missing ingredients of a satisfactory reconstruction.

There seems to be growing agreement in the foundations community that space-time should perhaps not be treated as a fundamental notion, and thus should somehow emerge from more fundamental ingredients of a framework, possibly alongside its non-trivial geometry capturing the presence of gravity. A framework concerned solely with quantum systems and their relations is a natural starting point for such considerations, aligned with the so-called `quantum first' approach to the matter.

The ideas underlying the QRF program are not new; they can be traced back to Arthur Eddington \cite{eddington1946fundamental}, with the contributions funding the subject due to Yakir Aharonov et.al \cite{aharonov1,Aharonov_2}. Besides the one presented here, the modern approaches, each aimed to some extent at addressing the high hopes summarised above, can be organized, with respect to their main sources of inspiration, into the following three non-independent sub-programs. The currently most popular approach is motivated by gauge theory and Dirac quantization of constrained systems, known as the \emph{perspective-neutral} approach \cite{vanrietvelde2018switching,krumm2021quantum,de2021perspective}. Another one is founded on a direct description of frame-change maps \cite{giacomini2019relativistic,de2020quantum}, and in some sense can be seen as semi-embedded in the perspective-neutral approach. The older developments, establishing the field, are based on information-theoretic ideas \cite{brs,castro2021relative}. The framework presented here, however very different in the motivations and concepts, is formally closest to this last theoretical development. It is worth emphasizing, that the different mentioned formalisms, besides sharing the core intuitions, differ significantly in terms of other physical principles that they are based on and the mathematical apparatus that is being used to implement them. Thus the perspective-neutral approach seems best suited to deal with gauge constrained systems, the approaches build around the frame changes seem appropriate for building informal\footnote{Considerable difficulties arise when e.g. notions like pure quantum states corresponding to a space-time metric are to be made rigorous.} models for quantum coordinate transformations, the information-theoretic considerations being best approached from the latter perspective. What distinguished the presented approach is its emphasis on operationality (defined below), which makes it useful for considerations not grounded in the concept of quantization of classical models. We propose a radically quantum approach decoupled from classical ways of thinking by beginning with general considerations of operationality and relationality in the presence of symmetries, with the framework taking a definite shape upon requiring the universality of quantum mechanics. It is hoped that this approach can be extended to the realm of Generalized Probabilistic Theories.

Before presenting the approach of this thesis, which we would like to call \emph{operational}, we would like to briefly point out the shortcomings and difficulties encountered by the existing approaches. We only mention our personal points of dissatisfaction, a full analysis of the landscape of QRF frameworks being far from the scope of the present work which is primarily aimed to communicate the operational framework in its fullest form.

Besides compromising operationality, all the formalisms developed so far lack solid mathematical foundations in the context of non-compact topological groups. Depending on the framework and acceptable level of rigor, without highly non-trivial extra assumptions, they are well-formulated for countable or compact groups. The perspective-neutral approach is believed to admit rigorous implementation in the realm of general Lie groups at the price of leaving the usual Hilbert space based framework for quantum mechanics and entering the realm of rigged Hilbert spaces and refined algebraic quantization. However, the rigging map that is being employed in the construction of the physical Hilbert space is only known to be defined under extra conditions, and even then it may lead to trivial physical inner-product giving rise to an empty physical Hilbert space \cite{giuliniUniquenessTheoremConstraint1999}. The only situation when the rigging map is known to be defined without extra assumptions is in the context of a compact group, when the construction is not needed. On top of these issues, the quantum reference frames considered in \cite{de2021perspective} are given by possibly distributional coherent state systems, contrary to our usage of the term where only coherent state systems understood as collections of vectors in the Hilbert space are allowed \ref{ex:css}. How this generality affects the prospects of the rigorous formulation of the rigging map and its implementation as the coherent averaging projector in the calculations presented in \cite{de2021perspective}, is never discussed. We believe these issues should be addressed before the perspective-neutral framework can be rigorously formulated in the setup of non-compact groups, while we agree with the authors that the passage from unimodular to non-unimodular non-compact groups does not seem problematic.

Let us now present the bird-eye view on the operational approach to quantum reference frames that is the focus of this thesis. It is a continuation of the theoretical efforts of \cite{loveridge2012quantum,lov6,lov1,lov2,lov4,loveridge2019relative,loveridge2020relational} that, as we will argue throughout this work, are beginning to culminate into a closed, full-fledged formalism for operational relational quantum kinematics. The resulting framework as we see it now can be understood as an implementation of the following principles.

\begin{enumerate}[label=\Roman*.]
    \item {\bf Operationality.} Since a physical theory is meant to model the aspects of the world that can (in principle) be falsified by means of experiment, it should be concerned with assigning probabilities to propositions, verifiable in measurement scenarios. We call a framework \emph{operational} if its subject is \emph{limited} to such probability assignments. In other words, an operational framework should be primarily concerned with probability distributions of (in principle) observable quantities and refrain from introducing notions non-aligned with this purpose.
    \item {\bf Relativity of measurement.}
    The only meaningful observables should be \emph{relative}, i.e. contingent upon the choice of the \emph{measuring instrument}. This is a generalization of the underlying idea of relativity that observable quantities are only meaningful after the \emph{frame of reference} has been \emph{specified}.
    \item {\bf Gauge-invariance and frame-covariance.} In the presence of an underlying symmetry structure of the theory, the relative observables should be invariant, while the measuring instruments should transform covariantly.
    \item {\bf Universality of quantum mechanics.} We understand it to state that Hilbert space based quantum mechanical formalism of normal states and positive operator-valued measures is rich enough to describe the accesible physical world, which should then be modeled as composed of quantum systems~alone.
\end{enumerate}
The operationality principle (I) is very well implemented by restricting to the realm of Generalized Probabilistic Theories, while the universality of quantum mechanics principle (IV) constrains the framework considerably, also allowing for concrete realizations of the previous two principles. To satisfy operationality as stated above, the formalism based on quantum mechanics should be primarily concerned with positive operator-valued measures (POVMs).\footnote{A reader not familiar with the operator-algebraic setup for quantum mechanics is advised to first consult the concise Appendix \ref{basics} for the necessary background.} Indeed, if a physical system is modeled on a Hilbert space $\hi$, with the states $\omega \in \S(\hi)$ given by the density operators\footnote{In particular, no non-normal or distributional states are permitted in the approach, contrary to \cite{de2021perspective}}, and probability distributions are understood as non-negative countably additive normalized measures\footnote{In the Discussion \ref{ch:discuss} we present conceptual arguments in favor of slightly \emph{restricting} the class of considered probability measures and allowed POCMs.} $p: \mathcal{F}(\Sigma) \to [0,1]$ on a measurable \emph{sample space} $(\Sigma,\mathcal{F})$, the most general continuous (in $\omega$) procedure of assigning the latter to the former is an assignment
\begin{equation}
     \S(\hi) \times \mathcal{F}(\Sigma) \ni (\omega, X) \mapsto p^{\E}_\omega(X) := \tr[\omega{\E}(X)]
\end{equation}
satisfying $\mathbb{0} \leq {\E}(X)$ for every measurable $X \in \mathcal{F}(\Sigma)$, ${\E}(\Sigma)=\mathbb{1}_\S$ and ${\E}(X\cup Y) = {\E}(X)+{\E}(Y)$ for disjoint $X,Y \in \mathcal{F}(\Sigma)$. It then follows that each ${\E}(X)$ is an effect,~i.e.
\begin{equation}
{\E}(X) \in \Eff(\hi) = \{F \in B(\hi) | \mathbb{0} \leq F \leq \mathbb{1}_\S\}
\end{equation}
for all $X \in \mathcal{F}(\Sigma)$, and thus such assignment of probability distributions is uniquely described by the \emph{positive operator-valued measure} (POVM)
\begin{equation}
{\E}: \mathcal{F}(\Sigma) \to \Eff(\hi)
\end{equation}
on $\Sigma$ with the collection of effects $\E(X)$ satisfying the mentioned properties. The measurable subsets $X \in \mathcal{F}(\Sigma)$ represent \emph{propositions} about the system $\S$, and the numbers $p^{\E}_\omega(X) \in [0,1]$ probabilities of these propositions being true given that the system $\S$ has been prepared in the state $\omega \in \S(\hi)$.\footnote{Such propositions indeed form a Boolean algebra, a subalgebra of the powerset~$\mathcal{P}(\Sigma)$. We avoid the usage of phrases such as ``\emph{the} outcome of the measurement of an observable ${\E}$ will be contained in the subset $X$'' as they only apply to the case when the sample space is discrete and hence we can speak of \emph{the outcome} -- otherwise, due to the always constrained precision of our observations, the verifiable claims will always concern subsets of the sample space. We like to think about the Boolean algebras of actually verifiable claims as more fundamental than the underlying sample spaces. See chapter \ref{ch:discuss} for a further discussion of this issue and the corresponding research direction.}

The gauge-invariance principle (III) is understood here similarly to the perspective-neutral approach in the sense of referring to some sort of global picture that should be invariant with respect to the action of the underlying group. However, as we will see, its implementation is very different -- in the presented framework it concerns only \emph{observables} and not the \emph{vectors} in $\hi$. However, the corresponding operational state spaces (see below fro the definition) turn out to consist precisely of invariant normal states in the realm of compact groups \ref{prop:invstspcomp} -- in this sense, we get close to the perspective-neutral philosophy of invariant states and Dirac observables, restricting however to normal, non-distributional states to cherish the fundamental in our view and aligned with the GPT spirit duality between states and effects. We will then have a (diagonal) action of the underlying symmetry group $G$ on the composite systems $\R\otimes\S$ (where $\R$ denotes the quantum reference system) and assume `physical' observables to be \emph{invariant}, i.e. they will be POVMs with the image in the algebra $B(\hirs)^G$ of invariant operators.

The frame-covariance principle (III) can be understood as a generalization of the situation encountered in the Special Theory of Relativity. The frames of reference considered there are mutually inertial coordinate systems, which are all \emph{equivalent} -- this is the content of the relativity principle. The term `equivalent' here can be understood as the fact that the descriptions of physical laws can be freely translated between different such frames, without altering their symbolic form. If a reference frame is modeled as an orthogonal coordinate system on the affine Minkowski space, such translations are given by elements of the Poincar{\'e} group acting transitively on the underlying space. To provide an operational, i.e. satisfying principle (I), analog of this situation, while also respecting the principle (IV) of quantum universality, which amounts to describing frames as quantum systems, we equip them with \emph{observables of orientation} on which the symmetry group $G$ underlying the formalism, e.g. the Poincar{\'e} group if the framework is to respect Special Relativity, \emph{acts}. Since in this context changing the reference frame is the same as acting upon the one we have with a group element, the action should now translate \emph{propositions about the frame orientation}, so it should be \emph{transitive}, and given on the sample spaces $\Sigma_\R$ of the frame-orientation observables. The group is then understood as composed of `operations' that can be performed on the frames, the covariance from principle (III) meaning that such operations can be modeled on the level of the frame's algebra. More precisely, given a group $G$ and its transitive action on $\Sigma_\R$, that extended to subsets of $\Sigma$ will be written as $X \mapsto g.X$ for $g \in G$, we have an action on the frames' algebra such that\footnote{The action of $G$ on $B(\hir)$ will always be assumed to be given via strongly continuous unitary representation of $G$ on $\hir$.}
\begin{equation}
    \E_\R(g.X) = U_\R(g)\E_\R(X) U_\R(g)^* \text{  for all } X \in \mathcal{F}(\Sigma_\R),  g \in G.
\end{equation}
A quantum reference frame $\R$ is then defined as a quantum system equipped with a POVM that is \emph{covariant}, i.e. satisfies the above property, referred to as the \emph{frame-orientation observable}.\footnote{Transitivity of the $G$-action on the sample spaces is usually a part of the definition of a covariant POVM.} The coherent state systems can be understood as equipped with covariant POVMs of the form
\begin{equation}
    \E^\eta(X) = \frac{1}{\lambda} \int_X \dyad{\eta_g}d\mu(g),
\end{equation}
where $\{\eta_g\}_{g \in G}$ is the coherent state system\footnote{Note that we do \emph{not} consider distributional coherent state systems, see \ref{ex:css} for the definition.} and $\mu$ denotes the Haar measure~on~$G$.

Fixing the reference $\R$ is understood as a choice made by an observer about \emph{which} quantum system will be used as a reference for observations, and also \emph{how} is it going to be used, i.e. which observable on $\R$ will be used to describe frames' orientation, with respect to which other quantities become meaningful as relational ones, as described below. After such a choice has been made, we only consider the observables respecting this choice whenever the frame is considered as a part of a composite system -- these are called \emph{framed} and are understood as those that can be accessed via the specified reference frame. In the light of the gauge-invariance principle (III), we then propose to model the relative observables as \emph{invariant} POVMs on the composite system $\R \otimes \S$ that respect the choice of the frame-orientation observable. This is achieved by the $\Y$ construction \cite{lov1,loveridge2017relativity} that, given a reference frame $\R$, provides an observable satisfying both requirements given \emph{any} observable on $\S$. It is understood as a map, contingent upon the choice of the frame-orientation observable, that \emph{relativizes} the standard, non-relative description of the system $\S$, making the choice of reference explicit. For a fixed frame $\R$ the $\Y^\R$ map is then given by the following integral\footnote{Originally, the $\Y^\R$ map, referred to as the \emph{$\R$-relativization map}, was only defined for frame observables being POVMs on the group $G$ itself. Recently it has been extended to homogeneous spaces on \emph{finite} groups \cite{homogyen2023}. This is the main deficiency of the presented framework and will be addressed shortly, see chapter \ref{ch:discuss} for some perspectives on this research problem.}~\cite{loveridge2020relational}
\begin{equation}
\Y^\R: B(\his) \ni A_\S \mapsto \int_G d\E_\R(g) \otimes U_\S(g)A_\S U_\S(g)^* \in B(\hirs)^G
\end{equation}

Indeed, a simple change of variables gives $h.\Y^\R(A_\S)=\Y^\R(A_\S)$ for all $h \in G$. Since this map, besides being linear and bounded, is also normal and (completely) positive, it can also be (equivalently) stated as a map between the effect spaces or, perhaps most naturally under the current interpretation, as a prescription of invariant POVMs to arbitrary ones, given by\footnote{The notation has been chosen to reflect the observation that when both Hilbert spaces are taken to be complex numbers, i.e. $\his \cong \hir \cong \mathbb{C}$, the notion of convolution of measures is recovered.}
\begin{equation}
\E_\S \mapsto \E_S * \E_\R := \Y^{\E_\R} \circ \E_\S: \mathcal{F}(\Sigma) \to \Eff(\hirs)^G.
\end{equation}
The relative observables are thus defined as the relativized ones. In this work, the $\Y$ construction is for the first time justified on operational grounds, in the context of a finite group and localizable (see below) frame, as generating \emph{all} and \emph{only} the invariant operators that respect the choice of the frame observable. Extending this result to a more general setting is a work in progress, which we reflect on in chapter \ref{ch:discuss}.

The operationality principle (I) entails that, after the frame has been chosen, the states of the composite systems $\R \otimes \S$ should be distinguished only as far as they give rise to different probability distributions upon an evaluation of the \emph{relativized effects}. This is achieved by the \emph{operational equivalence} procedure that quotients out the state spaces with respect to the equivalence relation defined by the allowed/accessible observables. This can be done since for \emph{any} subset $\mathcal{O} \subseteq \Eff(\hi)$ the relation
\begin{equation}
\Omega \sim_{\mathcal{O}} \Omega' \Leftrightarrow \tr[\Omega F]=\tr[\Omega' F] \hspace{3pt} \forall F \in \mathcal{O}.
\end{equation}
is an equivalence relation on $\T(\hi)$, and on $\S(\hi)$, thus allowing to define the \emph{operational state spaces} as the quotient spaces $\S(\hi)/\hspace{-3pt}\sim_{\mathcal{O}}$. They turn out to be total convex subsets of the real Banach spaces $\T(\hi)^{\rm{sa}}/\hspace{-3pt}\sim_{\mathcal{O}}$. Thus using the operational equivalence methodology, we provide an operational definition of \emph{relative states}: given a reference frame $\R$, the relative states are defined as the \emph{operational equivalence classes} of states on the composite system $\R\otimes\S$ that can be distinguished by the \emph{relative effects}, which are given by the (ultra-weak closure of) the image of the relativization map $\Y^\R$. We then show that the relative states can be equivalently characterized as states in $\S(\his)$ that lie in the image of the \emph{predual map} $\Y^\R_*$, and we write $\Omega^\R$ for $\Y^\R_*(\Omega) \in \S(\his)^\R$ and $\Omega \in \S(\hirs)$. Their definition is thus dual to that of the relative operators, providing a generalization of the standard trace-class/bounded operators duality. Indeed, writing $B(\his)^\R = \Y^\R\left(B(\his)\right)^{cl}$ (ultraweak closure), $\T(\his)^\R = \Y^\R_*(\T(\hirs))$ and $[\_]^\bigstar$ for the Banach space duality, we have
\begin{equation}
    [\T(\his)^\R]^\bigstar \cong B(\his)^\R.
\end{equation}

In fact, this is a special case of a Banach duality present in any operational setup defined by the set of available effects $\mathcal{O} \subseteq \Eff(\hi)$. Indeed, for arbitrary $\mathcal{O}$ we have
\begin{equation}
    [\mathcal{T}(\hi)/\hspace{-3pt}\sim_{\mathcal{O}}]^\bigstar \cong \rm{span}(\mathcal{O})^{cl}.
\end{equation}
Thus the algebraic structure of the dual of the Banach space in which the operational states live is in general lost, perhaps pointing to the more general, convex-theoretic setup to be appropriate for operational quantum physics. We reflect on potentially extending the presented framework into the realm of the Generalized Probability Theories in chapter~\ref{ch:discuss}.

It is very convenient to distinguish a class of quantum reference frames that exhibit some classical-like features. A frame $\R$ is called \emph{localizable} if for any $\epsilon > 0$ and $X \in \mathcal{F}(\Sigma)$ such that $\E_\R(X)\neq 0$ there exists a pure state $\ket{\xi_\epsilon} \in \hir$ for which $\ip{\xi_\epsilon}{\E_\R(X)\xi_\epsilon}>1-\epsilon$. A sufficient\footnote{In the context of POVMs on metrizable spaces, so, in particular, the covariant ones, this is also a necessary condition -- see \eqref{prop:locseqgen}.} condition a POVM may satisfy to admit such states is the norm-1 property \cite{heinonen2003norm} which states that the norm of any effect $||\E(X)||$ is either one or zero. Thus projection-valued measures are always localizable, but POVMs coming from coherent state systems will generally not be localizable. The operational interpretation of the states $\ket{\eta_g}$ as corresponding to the frame `being oriented' at $g \in G$ can, strictly speaking, retain its meaning only when they are `perfectly distinguishable', i.e. we have $\langle \eta_g,\eta_g' \rangle = \delta(g,g')$. Since without invoking distributional states this can only be achieved in the case of $G$ being at most countable, we avoid this kind of conditioning on frame orientations (and indeed the use of coherent state systems in general) but instead use the \emph{$\omega$-restriction maps}, which are given~by extending by linearity and continuity the following simple mapping
\begin{equation}
    \Gamma_\omega: B(\hirs) \ni A_\R \otimes A_\S \mapsto \tr[\omega A_\R]A_\S \in B(\his),
\end{equation}
where $\omega \in \S(\hir)$ is any state of the reference. Whenever the sample space is metrizable (which will always be the case for the frame-orientation observables), the localizability condition can equivalently be stated as the existence, for any point $x \in \Sigma$, of a sequence of pure states $\omega_n(x)$ such that the corresponding sequence of probability measures $\mu^{\E_\S}_{\omega_n(x)}$ converges (weakly) to the Dirac measure $\delta_x$. Taking $\omega = \omega_n(g)$ allows for sharp conditioning of the frame on a specified group element as a limiting procedure, while in general, the restriction map provides what can be thought of as \emph{probabilistic gauge fixing}, since generically $\omega \in \S(\hir)$ provides only a probability distribution of the frame orientation, namely the one $X \mapsto p_\omega^{\E_\R}(X)$, that approximates $\delta_x$ only in the case of a localizable frame $\R$ and $\omega = \omega_n(g)$. When the $\omega$-restriction map is applied to the relative operators, it gives rise to descriptions contingent not only upon the choice of the frame and its frame-orientation observable but also upon the frames' state. The corresponding operational states lie in the image of the \emph{$\omega$-conditioned $\R$-relativization map} defined as
\begin{equation}
    \Y^\R_\omega := \Gamma_\omega \circ \Y^\R: B(\his) \to B(\his) 
\end{equation}
The corresponding operational states, i.e. elements of $\S(\his)/\hspace{-3pt}\sim_\mathcal{O}$ with
\begin{equation}
\mathcal{O}=\Eff(\his)^\R_\omega := \Y^\R_\omega(\Eff(\his)),
\end{equation}
are of the form
\begin{equation}
    \rho ^{(\omega)} :=
    \Y^\R_*(\omega \otimes \rho ) = \int_G d\mu_{\omega}^{\E_\R}(g) U_\S(g)^*\rho U_\S(g),
\end{equation}
where $\rho \in \S(\his)$, and thus can be seen as states of $\S$ that have been `smeared' with respect to the probabilistic gauge fixing of $\R$ given by $\omega \in \S(\his)$. Since they arise by applying $\Y^\R_*$ to product states of the form $\omega \otimes \rho$, they are called \emph{$\omega$-product relative states}. In some sense, the procedure of attaching an external frame by $\Y^\R$ and then conditioning the description upon a chosen state with $\Gamma_\omega$ provides a generalization of the $G$-twirling procedure which is recovered in the compact $G$ case when $\omega$ is taken to be invariant. This way an operational justification is added to the claims made in the context of the information-theoretic approach that $G$-twirl should be understood as a lack of knowledge on the frame orientation, as explained in \cite{lov1, loveridge2017relativity}.

Via localizable frames, we can make direct contact with the standard formulation of quantum mechanics, in which the measuring apparatuses are thought of as classical systems. Indeed, in the case of localizable frames, $B(\his)^\R$ is a von Neumann algebra  with $\Y$ providing an isometric $*$-isomorphism $B(\his) \cong B(\his)^\R$ \eqref{prop:locyen}. In such a case the relative states are dense in the states of $\S$. More precisely, given a sequence of states $\omega_n$ localized at $e \in G$, we can approximate any state $\rho \in \S(\his)$ to arbitrary precision by the conditioned relative states due to the following \cite{lov1}
\begin{equation}
\lim_{n \to \infty}\Y^\R_*(\omega_n \otimes \rho) = \rho,
\end{equation}
where $\omega_n = \omega_n(e)$ is a localizing sequence centered at $e \in G$ and the limit is understood in terms of point-wise convergence of the expectation values.\footnote{The only topology that we use on $\T(\hi)$ and $\S(\hi)$ is the dual one to the ultraweak one on $B(\his)$, referred to as \emph{operational}.}

After presenting the operational setup for quantum reference frames as briefly described above, we also provide a \emph{frame-change map} aligned with our principles, where the localizable frames also play a prominent role. To this end, following other approaches, we adopt the \emph{internal} perspective on the problem. We then need a description of the system modeled on a \emph{total} Hilbert space $\hi_\T$ from which the frames will be extracted. We then notice that whatever the choice of an internal frame, i.e. regardless of the details of the decomposition $\hi_\T \cong \hir \otimes \his$ and the choice of the frame-orientation observable $\E_\R$, the $\R$-relative description will be invariant in the sense that $B(\his)^\R \subseteq B(\hirs)^G$. Thus the algebra $B(\hi_\T)^G$ and the corresponding invariant state space $\S(\hi_\T)_G$ provide the arena for the \emph{global} considerations that we need. A pair of subsystems, that will serve as frames, is then specified by decomposing the total Hilbert space into $\hi_\T \cong \hio \otimes \hit \otimes \his$. We assume this is done in a way compatible with the global $G$-action so that we have $U_\T = U_1 \otimes U_2 \otimes U_\S$. To complete the setup, we fix a pair of covariant POVMs $\E_i: \mathcal{B}(G) \to B(\hi_{\R_i})$, so that $\R_1$ and $\R_2$ become quantum reference frames.\footnote{We consider generalizing this notion of an internal frame in chapter \ref{ch:discuss}.} We then define the \emph{$\omega$-lifting map} $\mathcal{L}^\R_{\omega}: \S(\his)^\R \to \S(\hirs)_G$, with $\omega \in \S(\hir)$ an arbitrary state of the reference. It allows to \emph{lift} states relative to $\R$, i.e. those states in $\S(\his)$ that lie in the image of $\Y^\R_*$, to the global states in $\S(\hirs)_G$. It is given as a predual of the map
\begin{equation}
    \Gamma^\R_\omega := \Y^\R \circ \Gamma_\omega: B(\hirs)^G \to B(\his)^\R,
\end{equation}

called the \emph{relativized $\omega$-restriction map}, which takes invariant operators in the global description, restricts them upon the chosen references' state, and then relativizes.
The lifting map acts simply as
\begin{equation}
\mathcal{L}^\R_{\omega} = (\Gamma^\R_\omega)_*: \S(\his)^\R \ni \Omega^\R \mapsto [\omega \otimes \Omega^\R]_G \in \S(\hirs)_G,
\end{equation}
where $[\_]_G$ denotes the operational equivalence with respect to the invariant effects $\Eff(\hirs)^G$, and can thus be seen as an operational analogous to the \emph{inverse reduction map} introduced in \cite{de2021perspective}.

Upon the assumption of localizability of the frame $\R_1$, the lifting map formalizes the idea of `attaching' $\ket{e}_1$ state to a relative state present in the approaches based on QRF transformations (see e.g. \cite{de2020quantum}), where it is only rigorously defined in the countable group setting. The frame-change map that we propose is then given as a localized limit of the composition maps
\begin{equation}
\Y^{\R_2}_* \circ \mathcal{L}^{\R_1}_{\omega_n}: \mathcal{S}(\hit \otimes \his)^{R_1}  \to \mathcal{S}(\hio \otimes \his)^{R_2} ,
\end{equation}
where $\omega_n$ is the localizing sequence of states of the first frame centered at $e \in G$ as before. We thus pass through our \emph{global} description of the total system as can be seen on the following commuting diagram \vspace{5pt}
\begin{center}
\begin{tikzcd}
& \mathcal{S}(\hio \otimes \hit \otimes \his)_G  \arrow[rd, "\Y^{\R_2}_*"] &   \\
\mathcal{S}(\hit \otimes \his)^{R_1} \arrow[ru, "\mathcal{L}^{\R_1}_{\omega_n}"] \arrow[rr,] 
&          &  \mathcal{S}(\hio \otimes \his)^{R_2}.
\end{tikzcd}
\end{center}

When the relevant operational equivalences are taken into account, i.e. the choices of $\E_i$ are respected throughout, such frame-change map translates\footnote{The word `translation' here, and in what follows, is to be understood as a metaphor for providing an analogous description, much like restating a fact in a different language. The word `description' should be treated similarly so that our `translating descriptions' means providing a corresponding account of the same physical situation from a different perspective. This is very different than the usage of `translation' e.g. in \cite{giacomini2019quantum}, where it refers to shift operators of the translation group acting on the Euclidean space.} consistently between the states relative to different frames. This means that given a global state $\Omega~\in~\S(\hi_\T)_G$, the frame-change map takes the relative states with respect to $\R_1$ to those relative to $\R_2$, which can be stated as commutativity of the diagram \vspace{5pt}
\begin{center}
\begin{tikzcd}
            & \mathcal{S}(\hio \otimes \hit \otimes \his)_G \arrow[ld, "\pi_{\E_2} \circ \Y^{\R_1}_*"'] \arrow[rd, "\pi_{\E_1} \circ \Y^{\R_2}_*"] &   \\
\mathcal{S}(\hit \otimes \his)^{R_1}_{\E_2}  \arrow[rr,"\Phi_{1 \to 2}^{loc}"] &                                                                             &  \mathcal{S}(\hio \otimes \his)^{R_2}_{\E_1},
\end{tikzcd}
\end{center}
where $\pi_{\E_i}$ denote projections onto equivalence classes of relative states that take into account the choice of the other frames' orientation observable, for instance  
\begin{equation}
\pi_{\E_1}: \mathcal{S}(\hit \otimes \his)^{R_1} \to \mathcal{S}(\hit \otimes \his)^{R_1}_{\E_2} = \mathcal{S}(\hit \otimes \his)^{R_1}/\hspace{-3pt}\sim_{\E_2},
\end{equation}
where $\sim_{\E_2}$ denotes the operational equivalence with respect to the effects of the form $\E_2(X) \otimes F_S$ with $X \in \mathcal{B}(G)$ and $F_\S$ arbitrary. This way we make sure that the choice of both frame-orientation observables is respected throughout. In fact, without taking this into account, the frame-change maps would not provide translations coherent in the sense specified above. If we further assume localizability of $\R_2$ the reverse frame-change map provides an inverse, i.e. we have 
\begin{equation}
\Phi^{loc}_{2 \to 1} \circ \Phi^{loc}_{1 \to 2} = \text{Id}_{\mathcal{S}(\hit \otimes \his)^{R_1}_{\E_2}}
\end{equation}
This setup is easily extended to three (or more) frames, in which case we show that assuming localizability of e.g. $\R_1$ and $\R_2$, the frame-change maps are also composable, i.e. we have
\begin{equation}
\pi_{\E_2} \circ \Phi^{loc}_{1 \to 3} = \Phi^{loc}_{2 \to 3} \circ \Phi^{loc}_{1 \to 2},
\end{equation}
where the $\pi_{\E_2}$ projection takes into account the choice of the second frame's orientation-observables, which is absent from the $\Phi^{loc}_{1 \to 3}$ construction. 

We compare our frame-change maps with those defined in \cite{de2020quantum} in the setup of ideal frames and finite group $G$. We find that they agree precisely on separable states composed from basis states $\ket{g}$, thus recovering the classical intuition on which \cite{de2020quantum} is based. For general states, the agreement holds \emph{up to operational equivalence}. The result extends to the case when the second frame is a non-ideal quantum reference frame based on a coherent state system when the operational map is compared with the ``relational Schr{\"o}dinger picture'' frame-change map of the perspective-neutral approach~\cite{de2021perspective}.

Lastly, we propose a procedure of translating the relative descriptions given with respect to different frames treated \emph{externally}. Given a relative state $\rho^{\R_1} \in \S(\his)^{\R_1}$ and on in $\Omega \in \S(\hi_1 \otimes \hi_2)_G$, the latter allows to compute the probability distribution of relative orientation which allows constructing the corresponding relative state with respect to the second frame via
\begin{equation}
    \rho^{\R_2} := \int_G d\mu_{\Omega}^{\E_2 * \E_1}(g) g.\rho^{\R_1}.
\end{equation}
We live further analysis of these ideas for the future.

As a final general comment let us emphasize the fully relational approach to spatiotemporal notions being forced on the framework by the operationality and relativity of measurement principles (I and II). Indeed, as space or time are never measured directly, there is no place for the notion of space-time as primitive or fundamental in an operational framework. In the context of the presented formalism, the spatiotemporal \emph{relations} between quantum systems are described by the relevant quantum observables, e.g. relative-orientation observables. We thus view space-time as a convenient fiction, useful in the context of localized systems. The continuity of the relative orientations of frames upon localizability, and thus of the fictitious space-time they may be thought of as inhabiting, stems from the continuity of the underlying group. From this perspective, the group symmetry structure is considered more fundamental than the space on which it may be realized as the group of local isometries. We reflect further on these issues in discussion \ref{ch:discuss}.

The operational approach to quantum reference frames, as presented in this work and briefly summarized above, builds on the invariance of observables approach to imposing symmetry in quantum theory, the definition of the $\Y$ construction, the restriction maps $\Gamma_\omega$, and their interplay in the context of localizable reference systems, that were developed in the previous works \cite{loveridge2012quantum,lov6,lov1,lov2,lov4,loveridge2019relative,loveridge2020relational}. The new input presented here includes the following advances, many of which can also to be found in a preprint~\cite{opqrftrans}:
\begin{enumerate}
    \item[1)] Principle-based formulation of the framework.
    \item[2)] Definition and properties of \emph{operational equivalence} \eqref{def:opeq} and \emph{operational state spaces} (\ref{prop:statespace},\ref{def:opstsp},\ref{generalst}), \emph{operational Banach dualities} \eqref{prop:iml}.
    \item[3)] Characterization of localizable (norm-1) POVMs on metrizable spaces in terms of (weak) approximations of the Dirac measures \eqref{prop:locseqgen}.
    \item[4)] Definitions of different types of quantum reference frames \eqref{def:qrfs}.
    \item[5)] The idea of \emph{framing} \eqref{sec:frmdes} and the definition of \emph{relational descriptions} \eqref{sec:reldes}.
    \item[6)] An \emph{operational justification of the $\Y$ construction} \eqref{sec:yenmap}.
    \item[7)] Definition of relative-orientation observables for pairs of frames \eqref{def:relorobs}.
    \item[8)] Distinction between the \emph{relative} and \emph{invariant} descriptions, deeper understanding of the latter leading to incorporating the \emph{internal} perspective on frames \eqref{sec:dblinv}.
    \item[9)] Strengthening the results concerning localized relative descriptions \eqref{prop:locyen}, \eqref{th:con1}.
    \item[10)] Introduction of the Relational Reproducibility Property \eqref{RRP} and showing that it is satisfied by the relative-orientation observables \eqref{measuringY}.
    \item[11)] The concept of \emph{restricted relativized description} \eqref{sec:relresdes} and the \emph{lifting maps}~\eqref{def:liftmap}.
    \item[12)] The operational approach to changing internal frames of reference \eqref{ch:FrChM}, culminating in the definition of the \emph{frame-change map} \eqref{def:frchm} and the proof of invertibility, composability, and consistency of the provided construction \eqref{thm:frmchm}.
    \item[13)] A procedure of changing operational reference frames treated \emph{externally} \eqref{sec:extfcm}.
\end{enumerate}
A comprehensive list of newly defined spaces and maps is provided in Appendix~\ref{app:notation}.
\myemptypage

\chapter{Preliminaries}
We begin with some mathematical preliminaries. After some general remarks concerning state spaces, group actions, and topologies that we will use, we define the notion of \emph{operational equivalence} which constitutes our main tool for imposing the operationality principle (I) in the operator-algebraic setup. We then recall the definition of covariance of POVMs realizing our covariance principle (III) and provide a series of physically motivated examples illustrating its validity, including those POVMs that arise as coherent state systems and are extensively used in the perspective-neutral approach \cite{de2021perspective}. Finally, we discuss the localizability properties of POVMs, which will later be used to distinguish the localizable frames. They will be crucial for recovering the standard, i.e. non-relational setup in the sense of localizing the reference systems, and also for the definition of our frame-change maps. Equipped with these notions we proceed to discuss the basics of our operational setup for quantum reference frames in the next chapter.

Honoring the \emph{universality of quantum mechanics} principle (IV) the framework is based on the standard operator algebraic setting. A reader not used to this language is welcome to consult our Appendix \ref{basics} for a friendly exposition of the necessary background. When it comes to mathematics, our starting point can be summarized as follows.
\begin{enumerate}
    \item Quantum systems are modeled on separable Hilbert spaces, denoted by $\hi$.
    \item State spaces $\S(\hi) \subseteq \mathcal{T}(\hi)^{\rm{sa}}$ are total convex (see below) subsets of the real Banach space of trace-class operators $\mathcal{T}(\hi)$ given by the \emph{density operators}.
    \item Observables are modeled as positive operator-valued measures.
\end{enumerate}

However, while imposing operational equivalence we will encounter more general state spaces. The definition of a state space suitable for our purpose comes from the setting of Generalized Probability Theories (GPT) \cite{Beneduci2022-yo}.

\begin{definition}[affine functional, effect]\label{def:gpteff}
    An \emph{affine functional} on a convex subset $\S \subset V$ of a real vector space $V$ is a map $f: \S \to \mathbb{R}$ such that for any set of positive numbers $\{\lambda_i\}$ with $\sum_i \lambda_i = 1$ and any set of elements $\omega_i \in \S$ we have
    \[
        f \left(\sum_i \lambda_i \omega_i\right) = \sum_i \lambda_i f(\omega_i).
    \]
    Such $f$ is called \emph{bounded} if $f(\omega) < \infty$, and an \emph{effect} if $0 \leq f(\omega) \leq 1$, for all $\omega \in \S$. The set of bounded affine functionals on $\S$ will be denoted $B(\S)$, the set of effects by $\Eff(\S)$.
\end{definition}

\begin{definition}[state space]\label{def:stsp}
    A convex subset $\S \subset V$ of a real vector space $V$ is \emph{total convex} if the effects separate the elements of $\S$, i.e. when from $f(\omega_1)=f(\omega_2)$ for all $f \in \Eff(\S)$ we can conclude that $\omega_1 = \omega_2$. Such subsets will be referred to as \emph{state spaces}.
\end{definition}

A natural notion of a map between state spaces is the following.\footnote{Surprisingly, it does not seem to be present in the literature. We expect the resulting category to be well-behaved and plan to investigate its properties in future work.}

\begin{definition}[state space map]
    A \emph{state space map} is an affine map $\phi: \S \to \S'$, i.e. such that for any set of positive numbers $\{\lambda_i\}$ with $\sum_i \lambda_i = 1$ and any set of elements $\omega_i \in \S$ we have
    \[
        \phi \left(\sum_i \lambda_i \omega_i\right) = \sum_i \lambda_i \phi(\omega_i).
    \]
    Two state spaces are \emph{isomorphic} if there is an invertible state space map between~them.
\end{definition}

It can be shown \cite{Beltrametti1997-gu}, that whenever $\S \subseteq V$ is total convex, $B(\S)$ is an order-unit Banach space, dual to the base-norm space $\rm{span}(\S)$ in which $\S$ is a base for the positive cone $\rm{span}(\S)^+$. Much of the structure of $\S(\hi) \subset \T(\hi)^{\rm{sa}}$ is then recovered. The state space maps $\S \to \S'$ then have dual maps $B(\S') \to B(\S)$ that are unital and order-preserving. In the Hilbert space setting the state space maps are precisely the preduals of normal, unital, positive maps $B(\hi') \to B(\hi)$. \emph{Complete} positivity, however, is not granted, which is why we refrain from calling the state space maps \emph{channels}. We will not be using this language, except for making sure all our operational state spaces (see below) are total convex, and the maps between them affine. In particular, the frame-change map we propose is an affine functional, so a state space map. Generalizing the presented framework to the GPT setting, and possibly finding its reconstruction as embedded in such a general landscape, is one of our long-term goals. We briefly reflect on this research direction in chapter \ref{ch:discuss}. 

When it comes to group representations, we will only consider strongly continuous unitary representations of locally compact second countable topological groups $G$, written $U:G \to B(H)$. We will often write $g.A$ to stand for $ U(g) A U(g)^* $ with $A\in B(\hi)$ and $g.\rho$ for $U(g)^* \rho U(g)$ with~$\rho~\in~\mathcal{S}(\hi)$. The sample spaces of POVMs, if equipped with a $G$-action, will be topological spaces, and the action will be assumed continuous. Extended to the $\sigma$-algebra $\mathcal{B}(\Sigma)$ of Borel subsets it will be written as $X \mapsto h.X$ for $h \in G$.

Due to the emphasis on operationality, the preferred topology on $B(\hi)$ will be the \emph{topology of pointwise convergence of expectation values}, i.e., $A_n \to A \in B(\hi)$ exactly when $\tr[\Omega A_n] \to \tr[\Omega A]$ for all $\Omega \in \T(\hi)$. On the predual $\T(\hi)$, and by the restriction on the state space $\S(\hi)$, we use the corresponding, or dual, topology, i.e. $\Omega_n \to \Omega$ exactly when $\tr[\Omega_n A] \to \tr[\Omega A]$ for all $A \in B(\hi)$. We will refer to this topology as the \emph{operational topology}. The superscript $[\_]^{cl}$ will always refer to the \emph{ultraweak} closure of the subsets in operator algebras, and the \emph{operational} closure of the subsets in trace-class operators or state spaces.

In the remainder of this chapter, we first introduce the notion of operational equivalence that allows us to capture the operationality principle (I), then recall the definition of covariance of POVMs that realize the covariance principle (III), and finally, we discuss localizability of POVMs that allows for recovering standard quantum mechanics and for the usage of classical intuitions to define frame-change maps.

\newpage
\section{Operational equivalence}\label{sec:opeq}

Crucial for our implementation of the operationality principle (I) is the notion of operational equivalence which we now introduce. It is meant to be applied in situations when the set of available observables is restricted, for reasons either practical -- when the available measuring devices allow for accessing only a subset of effects, e.g. when the usage of a particular observable on a given system has been fixed -- or conceptual -- realizing some principles that disregard certain effects that would be available as mathematical objects, e.g. when the setup is assumed to respect certain symmetry conditions. We will encounter both of these instances below. In such situations, even though the system in question is modeled on a Hilbert space $\hi$, the state space $\S(\hi)$ is not operationally justified since no longer can all density operators be \emph{distinguished} by the considered effects. Indeed, according to the operationality principle (I), the states that cannot be distinguished by means of experiment \emph{should be identified}. Thus the following.

\begin{definition}[operational equivalence]\label{def:opeq}
    For any subset $\mathcal{O}\subseteq \Eff(\hi)$ we define the operational $\mathcal{O}$-equivalence relation on $\T(\mathcal{H})$ as follows
\[\Omega \sim_{\mathcal{O}} \Omega' \Leftrightarrow \tr[\Omega F]=\tr[\Omega' F] \hspace{3pt} \forall F \in \mathcal{O}.\]
\end{definition}

One easily verifies that this is always an equivalence relation and hence we can quotient the space of trace-class operators by the operational $\mathcal{O}$-equivalence. It turns out that the Banach space structure of $\T(\hi)$ is preserved under the quotient, with the ultraweak closure of $\rm{span}(\mathcal{O})$ realized as the dual space.

\begin{proposition}\label{prop:iml}
The space $\T(\mathcal{H})/\hspace{-3pt}\sim_{\mathcal{O}}$ is a Banach space and there is an isometric isomorphism between its dual and the ultraweak closure of the span of $\mathcal{O}$, i.e. we have
\[
\left[\T(\mathcal{H})/\hspace{-3pt}\sim_{\mathcal{O}} \right]^\bigstar \cong \rm{span}(\mathcal{O})^{cl}.
\]

\end{proposition}

\begin{proof}
For $F \in \mathcal{O}$, we write $\phi_F$ for the continuous linear functional $\rho \mapsto \tr[\rho F]$ and identify $\rm{span}(\mathcal{O})$ with the corresponding subspace in the dual space. It amounts to remark that
\[
\rho \sim_{\mathcal{O}} \rho' \Leftrightarrow \forall \phi_F \in \mathcal{O} \hspace{3pt} \phi_F(\rho)=\phi_F(\rho') \Leftrightarrow \forall \phi_F \in \mathcal{O} \hspace{3pt} \phi_F(\rho-\rho')=0 \Leftrightarrow \rho-\rho' \in {}^{\perp}\mathcal{O},
\]
where ${}^{\perp}\mathcal{O}$ is the pre-annihilator of $\mathcal{O}$ defined as ${}^{\perp}\mathcal{O}:=\bigcap_{\phi_F\in \mathcal{O}}\ker(\phi_F)$, which is always closed in $\T(\hi)$ as an intersection of closed sets. Moreover, the pre-annihilator is always a subspace since ${}^{\perp}\mathcal{O} = {}^{\perp}\rm{span}(\mathcal{O})$. The quotient space $\T(\hi) /\hspace{-3pt}\sim_{\mathcal{O}} = \T(\hi) /{}^{\perp}\mathcal{O}$ is then a Banach space with the quotient norm defined as
\[||\rho + {}^{\perp}\mathcal{O} ||=\inf_{\mu \in {}^{\perp}\mathcal{O}} ||\rho+\mu||.\]
Finally, Theorems 4.9 and 4.7 in \cite{rudin1974functional} give
\[
\left(\T(\hi)/\hspace{-3pt}\sim_{\mathcal{O}} \right)^* = \left(\T(\hi)/{}^{\perp}\mathcal{O}\right)^* \simeq {}^{\perp}\mathcal{O}^{\perp} = \left({}^{\perp}\rm{span}(\mathcal{O})\right)^{\perp} = \rm{span}(\mathcal{O})^{cl}.
\]
\end{proof}

The proposition above can be seen as restricting the $\mathcal{O}=\Eff(\hi)$ setting, in which we have $\rm{span}(\Eff(\hi))^{cl} = B(\hi)$ and the usual $\T(\hi)^* \cong B(\hi)$, to the situation where the set of available effects is smaller. Thus the operators in $\rm{span}(\mathcal{O})^{cl}$ can distinguish $\mathcal{O}$-equivalence classes. In fact, we have the following.

\begin{proposition}\label{prop:statespace}
    The set $\S(\mathcal{H})/\hspace{-3pt}\sim_{\mathcal{O}}$ is a total convex subset of $\T(\hi)^{\rm{sa}}/\hspace{-3pt}\sim_{\mathcal{O}}$, and is thus a state space and, moreover,  is closed in the quotient operational topology.
\end{proposition}

\begin{proof}
Since the real linear structure of $\T(\hi)^{\rm{sa}}/\hspace{-3pt}\sim_{\mathcal{O}}$ comes from $\T(\hi)^{\rm{sa}}$, convexity is preserved under the quotient. In particular, writing $[\_]_\mathcal{O}$ for the $\mathcal{O}$-equivalence classes, for any $\rho,\rho' \in \S(\hi)$ and $0 \leq \lambda \leq 1$ we have
\[
\lambda[\rho]_\mathcal{O} + (1-\lambda)[\rho']_\mathcal{O} = [\lambda \rho + (1-\lambda)\rho']_\mathcal{O} \in \S(\hi)/\hspace{-3pt}\sim_{\mathcal{O}}.
\]

The bounded affine functionals on $\S(\mathcal{H})$ are given by $\rho \mapsto \tr[\rho A]$ with $A \in B(\hi)$, with the effects given by the subset $\mathcal{E}(\hi) = \{F \in B(\hi) | \mathbb{0} \leq F \leq \mathbb{1}\}$. The effects on $\S(\mathcal{H})/\hspace{-3pt}\sim_{\mathcal{O}}$ are then those that are well-defined on classes $\S(\mathcal{H})/\hspace{-3pt}\sim_{\mathcal{O}}$, and hence are given by the operators in $\mathcal{E}(\hi) \cap \rm{span}(\mathcal{O})^{\rm cl}$. Indeed, $F \in \mathcal{E}(\hi)$ is well-defined on the $\mathcal{O}$-equivalence classes of states if whenever $\rho \sim_\mathcal{O} \rho'$ we have $\tr[\rho F] = \tr[\rho'F]$, which means that $F \in \rm{span}(\mathcal{O})^{\rm cl}$. The effects then separate the elements of $\S(\mathcal{H})/\hspace{-3pt}\sim_{\mathcal{O}}$ by construction, providing total convexity.

The state space $\S(\hi)$ is operationally closed in $\T(\hi)$ since for any sequence of states $(\rho_n) \subset \S(\hi)$ such that $\lim_{n \to \infty}\tr[\rho_n A] = \tr[T A]$ for all $A \in B(\hi)$ and some $T\in \T(\hi)$, we can conclude that $T \in \S(\hi)$. Indeed, the continuity of the trace gives positivity and normalization of $T$. The operational topology on $\S(\hi)/\hspace{-3pt}\sim_{\mathcal{O}}$ is the quotient topology of the one on $\T(\hi)$ so we have
\[
\lim_{n \to \infty} [\rho_n]_\mathcal{O} = [T]_\mathcal{O} \in \S(\hi)/\hspace{-3pt}\sim_{\mathcal{O}}.
\]
\end{proof}

\begin{definition}[Operational state space]\label{def:opstsp}
    The total convex subset $\S(\mathcal{H})/\hspace{-3pt}\sim_{\mathcal{O}}$ of $\T(\hi)^{\rm{sa}}/\hspace{-3pt}\sim_{\mathcal{O}}$ will be referred to as an \emph{$\mathcal{O}$-operational state space}.
\end{definition}

The set $\mathcal{O}$ will often be the image of a normal, positive, unital map. In such a case, the corresponding operational state space admits an alternative useful characterization. 

\begin{proposition}\label{generalst}
    Given an normal, positive, unital map $F:B(\mathcal{K})\to B(\mathcal{H})$ there is a state space isomorphism\footnote{Notice that in general this correspondence doesn't hold at the level of the ambient Banach spaces. In fact, $\Im F_* $ might not be closed and hence not a Banach space. Considering its norm-closure gives a bijective bounded linear map $\T(\mathcal{H})/\hspace{-3pt}\sim_{\Im F} \to F_* (\T(\mathcal{K}))$ but generally not an isometry.}
    \[
    \S(\mathcal{H})/\hspace{-3pt}\sim_{\Im F} \cong F_* (\S(\mathcal{K}))
    \]
\end{proposition}

\begin{proof}
    Since $F$ is normal, we can write ${}^\perp \Im F = \ker F_* $, and thus $\T (\mathcal{H})/\hspace{-3pt}\sim_{\Im F} = \T (\mathcal{H})/\ker F_*$. Then $F_* $ restricts to an invertible bounded linear map $\T (\mathcal{H})/\ker F_* \to \Im F_* $. As $F$ is linear, unital and positive, the $F_*$ map restricts further to an affine bijection $\S (\mathcal{H})/\ker F_* \to F_* (\S(\mathcal{K}))$, providing the expected state space isomorphism.
\end{proof}

As a simple instance of the proposition above, we note that for a von Neumann algebra $F: \mathcal{N} \hookrightarrow B(\hi)$ the predual is the quotient map $ F_*: \T(\hi) \to \T(\hi)/\hspace{-3pt}\sim_\mathcal{N}$, and the normal state space is given by $F_*(\S(\hi)) \cong \S(\hi)/\hspace{-3pt}\sim_\mathcal{N}$.

\newpage
\section{Covariance}

Due to the covariance principle (III), frame-orientation observables are covariant POVMs. Here is a precise definition.

\begin{definition}[Covariant POVMs]\label{def:imp}
Consider a strongly continuous projective unitary representation $U:G \to B(\hi)$ of a \emph{locally compact second countable Hausdorff topological} group~$G$, and a \emph{topological space} $\Sigma$, considered as a \emph{measurable space} with the $\sigma$-algebra $\mathcal{B}(\Sigma)$ of Borel subsets, and equipped with a continuous transitive $G$-action. Then a POVM $\E:\mathcal{B}(\Sigma) \to \Eff(\hi)$ is called \emph{covariant} if for any $g \in G$ and $X \in \mathcal{B}(\Sigma)$ we have
\begin{equation}\label{eq:covp}
    \E (g.X)= U(g) \E (X) U(g)^*,
 \end{equation}
 where $g.\_$ denotes the extension of the action of $G$ on $\Sigma$ to the Borel subsets.
\end{definition}

Covariant POVMs provide a general notion of quantum observables in the presence of symmetry. They generalize coherent state systems that are considered frame-orientation observables in the perspective-neutral approach \cite{de2021perspective}, as the following example shows.

\begin{example}[Systems of coherent states]\label{ex:css}
\normalfont
Consider a vector $\ket{\phi} \in \hi$ that is \emph{cyclic} for a given representation $U$, i.e. the such that $\rm{span}(\{U(g) \ket{\phi}, g \in G\})$ is dense in $\hi$. The orbit $\{\ket{\phi(g)} := U(g) \ket{\phi}\}$ is then called a system of (Perelomov-Gilmore) coherent states \cite{Perelomov,ali2000coherent}. Under an extra square integrability condition, they resolve identity, meaning that
\begin{equation}\label{eq:pgco}
    \int_G \dyad{\phi(g)} d \mu (g) = \lambda \mathbb{1},
\end{equation}
where $\lambda$ is positive and $\mu$ is the Haar measure as before. Then
\begin{equation}
\E^\phi(X):= \frac{1}{\lambda}\int_X \dyad{\phi(g)}{\phi(g)} d \mu (g)
\end{equation}
is a covariant POVM on $G$. We will refer to such POVM a \emph{coherent state POVM}.\footnote{One should be careful with applying the results of the theory of coherent states as presented e.g. in \cite{Perelomov} as it was developed under the assumption the representations are \emph{irreducible}, which we certainly do not want to assume here. We will avoid the use of coherent state systems as frame-orientating observables, except when it will be useful to make contact with the perspective-neutral approach.}
\end{example}

Notice here, that since the action of $G$ on $\Sigma$ is assumed to be \emph{transitive}, $\Sigma$ has to be homeomorphic to the quotient topological space $G/H$ for some closed subgroup $H \subseteq G$. The identification is given by fixing a point $x \in \Sigma$, taking $H=H_x$ to be the stabilizer subgroup of $x$, which gives a bijection $\Sigma \cong G/H_x$, and noticing that since all such subgroups are conjugate to one another, the resulting quotient spaces $G/H_x$ are all homeomorphic. The subtle difference between $\Sigma$ and $G/H$ is that the latter admits a distinguished coset containing the identity $eH$, while the former does not. The following is a useful definition of \emph{equivalent} covariant POVMs.

\begin{definition}[Equivalence of covariant POVMs]
    Given a pair of covariant POVMs $\E_i: \mathcal{B}(\Sigma_i) \to \Eff(\hi_i)$ with $i=1,2$, we call them \emph{unitarily equivalent} if there is a unitary map $U:\hi_1 \to \hi_2$ and a homeomorphism $f:\Sigma_1 \to \Sigma_2$ such that
    \[
        \E_2(X) = U\E_1(f^{-1}(X))U^*.
    \]
\end{definition}

Up to the unitary equivalence, all the covariant POVMs are then given on the quotient spaces $\Sigma = G/H$. Covariant PVMs, so projection-valued measures, also referred to as sharp POVMs \eqref{def:povms}, often called \emph{systems of imprimitivity} in mathematics literature, have been fully characterized by Mackey's Imprimitivity Theorem which provides a bijective correspondence between systems of imprimitivity based on $G/H$ and unitary representations of the \emph{subgroup} $H \subseteq G$ (see e.g. \cite{landsman2006between}). It can be used to provide generic examples of covariant PVMs. For instance, taking $H=\{e\}$ (with the trivial representation), the theorem provides unitary equivalence between all the covariant PVMs on $\Sigma=G$.% mapping to the effects in an irreducible representation of $G$. 

\begin{example}\label{ex:covsharpfin}
\normalfont
Consider (left) regular representation of a countable group $G$, given on $\hi = L^2(G)$ by $U(g)\ket{g'} = \ket{gg'}$, where $\ket{g}$ denotes elements of the orthonormal basis given of indicator functions. The corresponding unique (up to unitary equivalence) covariant PVM is given by
\begin{equation}
    P: g \mapsto \dyad{g} \in B(L^2(G)).
\end{equation}
\end{example}

Upon fixing the representation of $H$ to be trivial in the Imprimitivity Theorem correspondence, the simple example above naturally extends to the following.

\begin{example}\label{ex:covsharp}
\normalfont
Consider a $\sigma$-finite measure space $(G/H, \mu)$ equipped with a transitive (left) $G$-action such that $\mu$ is \emph{invariant}. Then
\begin{equation}
(U(g)f)(x)=f(g^{-1}.x)
\end{equation}
is a representation of $G$, with the unique (up to equivalence) covariant PVM given by
\begin{equation}
  P(Y): f \mapsto \chi_Y f,
\end{equation}
where $Y \subseteq G/H$. Taking $H=\{e\}$ yields the regular representation as before, equipped with the canonical PVM based on $G$ (with the Haar measure). For instance, take $G=\mathbb{R}$, understood as a position sample space of a particle. The action of $\mathbb{R}$ on the wave functions in $L^2(\mathbb{R})$ is given by the shift operator
\begin{equation}
(U(y)\psi)(x)=\psi(x-y),
\end{equation}
with the sharp position observable recovered as the unique shift-covariant PVM (see~\ref{pozobs})
\begin{equation}
P(Y)\psi(x) = \chi_Y \psi(x).
\end{equation}
\end{example}

with $Y \subseteq \mathbb{R}$. Many more, also unsharp, examples of covariant POVMs used to model quantum observables under the presence of symmetries can be given, see e.g. \cite{busch1997operational,loveridge2017relativity}.

Notice here that since $G$ is assumed second countable and Hausdorff and $H$ need to be closed, the quotient spaces $G/H$ are also Hausdorff, and since $H$ is a subgroup the quotient map is open and thus preserves second-countability. Urysohn's metrization theorem then assures metrizability of all the homogeneous spaces that we consider.

\newpage
\section{Localizability}

The notion of localizability of a POVM that we are now going to introduce is needed for recovering the standard kinematics of quantum physics from the relational one summarized here (Theorem \ref{th:con1}), and defining the frame change maps in chapter~\ref{ch:FrChM}.

Consider a quantum observable modeled by a POVM $\E: \mathcal{F}(\Sigma) \to \Eff(\hi)$. It maps the states of the systems $\omega \in \S(\hi)$ to the probability distributions over $\Sigma$ via the \emph{Born formula}\footnote{When the probability distribution $p_{\omega}^{\E}$ will be considered as a (positive normalized) \emph{measure} for integration, it will often be denoted by $\mu_{\omega}^{\E}$.}
\begin{equation}\label{eq:born}
    \mu_{\omega}^{\E}(X) = \tr[\omega \E (X)].
\end{equation}

POVMs may have different characteristics in terms of the probability distributions they give rise to upon evaluation on quantum states. For instance, we will be interested in whether a given observable can provide highly localized probability distributions or not. If arbitrary localization of $\E$ is possible, there are states on $\S$ that give \emph{definite} truth values of \emph{almost} any experimentally verifiable statement regarding the quantity corresponding to $\E$. This can be stated in terms of the $\epsilon$-decidability property~\cite{heinonen2003norm}.

\begin{definition}[Localizability]\label{def:loc}
A POVM $\E : \mathcal{B}(\Sigma) \to \Eff(\hi)$ POVMs is called \emph{localizable} iff it satisfies the the \emph{$\epsilon$-decidability property}: for every $\E(X)\neq 0$ and for any $\epsilon > 0$ there exists a pure state $\ket{\xi_\epsilon} \in \mathcal{S}(\hi)$ for which $\ip{\xi_\epsilon}{\E(X)\xi_\epsilon}>1-\epsilon$.
\end{definition}

 Thus if $\E$ is localizable, for \emph{any} measurable subset $X \in \mathcal{F}(\Sigma)$ we can find a pure state that gives the probability of the $X$-proposition being true arbitrarily close to the identity. The states corresponding to the subsets $X \subseteq \Sigma$ much smaller than the available experimental resolution will then assign almost definite truth values to almost all experimentally verifiable statement concerning the observable modeled by $\E$. In other words, the \emph{localizable} POVMs, when evaluated on specific pure states, lose their inherently probabilistic nature encoded in the operational setup as endorsed in this work, and instead exhibit \emph{classical-like} behavior. As shown in \cite{heinonen2003norm}, $\E$ is localizable iff it satisfies the \emph{norm-$1$ property}, i.e. 
\begin{equation}\label{norm1}
||\E(X)|| = 1 \vee 0 \text{ for all } X \in \mathcal{F}(\Sigma).
\end{equation}
Sharp POVMs (i.e. those for which $\E(X)$ are all projections) are then always localizable. Indeed, given $X\in \mathcal{F}(\Sigma)$ any vector in the image of the corresponding projection $\xi \in \Im \E(X)$ gives $\ip{\xi}{\E(X)\xi}=1$, so in this case no limiting procedure is needed. If the sample space $\Sigma$ is \emph{metrizable}, which is true for covariant POVMs, we have a very useful characterization of the localizable ones in terms of the probability measures: ${\E}$ is localizable iff the Dirac delta measure centered at any $x \in \Sigma$ can be approximated to arbitrary precision with measures of the form $\mu^{\E}_{\omega_n(x)}(X) = tr[{\E}(X)\omega_n(x)]$, where $\omega_n(x)$ are pure states. We will call $\omega_n(x)$ a \emph{localizing sequence centered~at~$x$}.

\begin{proposition}\label{prop:locseqgen}
    Consider a POVM $\E : \mathcal{B}(\Sigma) \to B(\hi)$ and assume $\Sigma$ metrizable. Then the following are equivalent:
    \begin{enumerate}
        \item $\E$ is localizable.
        \item For any $x \in \Sigma$ there exists a sequence of pure states $\omega_n(x) \in \S(\hi)$ such that
        \begin{equation}
        \lim_{n \to \infty}\mu^{\E}_{\omega_n(x)} = \delta_x
        \end{equation}
       in the sense of weak convergence of measures.
    \end{enumerate}
\end{proposition}
\begin{proof} 
 For \emph{1.} $\Rightarrow$ \emph{2.} denote $B_n$ the open ball centred at $x$ of radius $1/n$. Since $\E$ satisfies the $\epsilon$-decidability property we can choose unit vectors $\ket{\xi_n}$ such that $\braket{\xi_n }{\E(B_n) \xi_n }> 1 - 1/n $. Denoting by $\omega_n(x)$ the associated pure state as before we get $\mu^{\E}_{\omega_n(x)}(B_n ) > 1 - 1/n$.

To show weak convergence we will use the porte-manteau theorem \cite{billingsley2013convergence}. To this end, we need to show that for each measurable set $X$ with negligible border, that is such that $\delta_x (\partial X) = 0 $ for all $x$, we have: $\lim\limits_{n\to \infty} \mu^{\E}_{\omega_n(x)}(X) = \delta_x (X) $. We then calculate
\[\mu^{\E}_{\omega_n(x)}(X) = \mu^{\E}_{\omega_n(x)}(X\setminus B_n ) + \mu^{\E}_{\omega_n(x)}(X\cap B_n ) \]

For the first term, we have
\[
\mu^{\E}_{\omega_n(x)}(X\setminus B_n ) \leq \mu^{\E}_{\omega_n(x)}(\Sigma \setminus B_n ) = 1-\mu^{\E}_{\omega_n(x)}(\B_n ) \leq \frac{1}{n},
\]
so it vanishes as $n$ goes to infinity. For the second term, assume $x \notin \partial X $ and distinguish two cases.

\begin{itemize}
    \item If $ x\in X $, then $x \in \mathring{X}$ (the interior of $X$). As $\mathring{X}$ is open, for large enough $n$ we have $B_n \subseteq \mathring{X} \subseteq X$, so $X\cap B_n = B_n$, and then $\mu^{\E}_{\omega_n(x)}(X\cap B_n )= \mu^{\E}_{\omega_n(x)}( B_n )> 1 - 1/n $. Thus, the second term goes to $1$ as $n$ goes to infinity.
    \item If $ x\notin X $, then $x\in \Sigma \setminus \overline{X}$ (the complementary of the adherence of $X$), which is an open set. So for $n$ large enough we always have $B_n \subseteq \Sigma \setminus \overline{X} \subseteq \Sigma \setminus X $, hence $X\cap B_n = \emptyset $, leading to $ \mu^{\E}_{\omega_n(x)}(X\cap B_n )= 0 $. Thus, the second term goes to $0$ as $n$ goes to infinity.
\end{itemize}

We then have $\lim\limits_{n\to \infty} \mu^{\E}_{\omega_n(x)}(X) = \delta_x (X)$, from which by porte-manteau theorem we can conclude that the sequence $\{\mu^{\E}_{\omega_n(x)}\}_n$ converges weakly to $\delta_x$ in the space of measures on $\Sigma$.

 For \emph{2.} $\Rightarrow$ \emph{1.} fix $x\in X \in\mathcal{F}(\Sigma)$ arbitrary and pick a sequence of vectors $\xi_n$ such that $\dyad{\xi_n}=\omega_n(x)$. The weak convergence then gives that for any $\epsilon >0$ we can find $n$ large enough so that $|1 - \ip{\xi_n}{\E(X)\xi_n}| < \epsilon$. Since $\E \in \Eff(\hi)$ we can lift the absolute value, which gives the claim.
\end{proof}
\myemptypage

\chapter{Relational Quantum Kinematics}\label{ch:RelKin}
In this chapter, we provide the core ingredients of the framework, discussing its relation to other implementations of similar ideas and concepts present in the literature as we go. Before introducing quantum reference frames and related notions, we begin by applying the operational equivalence methodology to impose invariance on composite systems $\R \otimes \S$, which embodies our invariance principle (III). We use this simple setting to contrast our approach to imposing gauge-invariance with others present in the literature, revealing similarities, but also important differences. We refer to this setup as \emph{invariant descriptions}. These considerations, besides making direct contact with other approaches present in the literature, will be used in the contexts of the \emph{perspective-independent} and \emph{global} descriptions introduced towards the end of this chapter, the latter one supporting the internal perspective on quantum reference frames which we utilize for describing frame-change maps in chapter \ref{ch:FrChM}.\\
Next, we proceed to our general definition of quantum reference frames as covariant positive operator-valued measures, thus embracing the principle (III). We also distinguish important classes of frames and discuss the conceptual and technical differences between the perspective presented here and in other approaches. We introduce an important from the operational perspective class of \emph{localizable} frames.\\
Having specified the frame, the operationality principle (I) demands appropriate constraints on the set of available effects in any effect space of the form $\Eff(\hirs)$. Fulfilling this requirement, again by the means of applying the operational equivalence, results in what we call \emph{framed descriptions}. They will be used as a stepping stone for the definitions in the next section, and in the context of our frame-change maps.\\
We then combine the principles of operationality (I), relativity of measurement (II) \emph{and} covariance \& invariance (III) to define the \emph{relational description} as given up to the operational equivalence with respect to the \emph{invariant} and \emph{framed} effects.\\
Finally, we use the \emph{relativization map $\Y^\R$} \cite{lov1}, which provides \emph{direct access} to the \emph{invariant} and \emph{framed} effects on the composite system $\R\otimes\S$ upon the choice of the \emph{covariant} frame-orientation observable $\E_\R$ given on $G$ itself (principal frame), to define the \emph{relative descriptions} by again invoking \emph{operational equivalence}, this time with respect to such \emph{relativized} effects. We conjecture that in the cases when the $\Y$ map is defined, the \emph{relational} and \emph{relative} descriptions are \emph{equivalent}. We provide proof in the finite group setting.\\
Lastly, we reflect on the twofold meaning of invariance when stipulated on a single system and define the \emph{perspective-independent} and \emph{global} descriptions, vital for our understanding of the internal view on quantum reference frames.

\newpage
\section{Invariant descriptions}\label{sec:reldes}

The main purpose of this section is to point to the similarities and differences between the presented framework and others when it comes to stipulating \emph{gauge-invariance}. Our approach is based on operational equivalence. Given a composite system $\R\otimes\S$, we then restrict the set of available effects to consist of the ones \emph{invariant} with respect to the \emph{diagonal} (strongly continuous) unitary representation $U_{\R\S} = U_\R \otimes U_\S$ of a (locally compact) group $G$ on $B(\hirs)$ and define the corresponding operational state space.

\begin{definition}
    Given a pair of systems $\R$ and $\S$, we will refer to the operationally $\Eff(\hirs)^G$-equivalent trace-class operators in~$\T(\hirs)$, where
    \[
    \Eff(\hirs)^G = \{F_{\R\S} \in \Eff(\hirs)~|~g.F_{\R\S} = F_{\R\S}\}
    \]
    as \emph{$G$-equivalent}. The classes of $G$-equivalent states, denoted $[\Omega]_G$ with $\Omega \in \S(\hirs)$, will be referred to as \emph{global states}. 
\end{definition}

The classes of $G$-equivalent states then consist precisely of those states that cannot be distinguished by the invariant effects on the composite system $\R\otimes\S$. Since the effects span the full algebra, we have $\rm{span}\{\Eff(\hirs)^G\}^{cl} = B(\hirs)^G$ and the Propositions \ref{prop:iml} and \ref{prop:statespace} immediately give the following.

\begin{proposition}
    We have the following Banach space isomorphism
    \[
    \left[\T(\hirs)/\hspace{-3pt}\sim_G\right]^\bigstar \cong B(\his)^G.
    \]
    Moreover, the set of $G$-equivalent states
    \[
    \S(\hirs)_G := \S(\hirs)/\hspace{-3pt}\sim_G
    \]
    is a state space in $\T(\hirs)^{\rm{sa}}/\hspace{-3pt}\sim_G$.
\end{proposition}

In standard quantum mechanics, we have a Banach space isomorphism $\T(\hi)^\bigstar \cong B(\hi)$, which says roughly that quantum observables arise as dual to states. The von Neumann subalgebra of invariant operators can be thought of as an analog of the full system algebra in the standard approach to quantum mechanics upon specifying the reference system and after invoking the invariance principle (III). We refer to this setup as the \emph{invariant description} of the system $\S$ with respect to the frame $\R$.

Notice that in general $\S(\hirs)_G$ cannot be identified with the set $\S(\hirs)^G $ of invariant (normal) states, i.e. those satisfying $g.\Omega =\Omega$, this last set being empty if $G$ is not compact. However, in the compact case, the $G$-equivalent states coincide with the invariant states used in the information-theoretic approach \cite{brs,castro2021relative}. To see this, assume $G$ compact and recall the \emph{$G$-twirl} (or \emph{incoherent group average}) map $\mathcal{G}:B(\hi)\to B(\hi)$ given by
\begin{equation}
    \mathcal{G}(A)=\int_G d\mu(g) U(g)A U(g)^*,
\end{equation}
where $\mu$ is the Haar measure. The $G$-twirl is an normal, positive, unital map, with the pre-dual $\mathcal{G}_* : \T(\hi) \to \T(\hi)$ taking a similar form
\begin{equation}
\mathcal{G}_* (\rho)=\int_G d\mu(g) U(g)^* \rho U(g).
\end{equation}
Both $\mathcal{G}$ and $\mathcal{G}_* $ are surjective, respectively on the set $B(\hi)^G $ of invariant operators and on the set $\T(\hi)^G $ of invariant trace class operators. If $G$ is not compact the integral in general does not converge, either for states or operators and we, therefore, avoid the use of the $G$-twirl in this framework. However, in the case of compact $G$, both $G$-twirling maps are well-defined and the Proposition \ref{generalst} gives the isomorphism of state spaces
\begin{equation}\label{invstspcomp}
S(\mathcal{\hi})^G = \mathcal{G}_*(\S(\hi)) \cong \S (\mathcal{\hi})_G,
\end{equation}
ensuring that, in this case, the states respecting our operationality (I) and invariance~(III) principles coincide with the invariant ones used in other approaches. In fact, in this specific case, the correspondence lifts to the ambient Banach spaces.

\begin{proposition}\label{prop:invstspcomp}
   For compact $G$ we have a Banach space isomorphism
   \[
        \T(\hi)^G \cong \T(\hi)_G.
   \]
\end{proposition}

\begin{proof}
We have $ \T(\hi)_G = \T(\hi)/\ker(\mathcal{G}_* ) $ and $\T(\hi)^G = \Im \mathcal{G}_* $. The predual map $\mathcal{G}_*$ factorizes through a bijective map $\tilde{\mathcal{G}_*}:\T(\hi) /\ker(\mathcal{G}_* ) \to \Im \mathcal{G}_* $. We will show that $\tilde{\mathcal{G}_*}$ is an isometry. First, $\mathcal{G}_* $ is a contraction
\[|| \mathcal{G}_* (\rho) || = ||\int_G g\cdot \rho d\mu(g)||\leq \int_G || g\cdot \rho || d\mu(g) \leq \int_G ||\rho|| d\mu(g)= ||\rho || \]
it follows that for all $\mu \in \ker(\mathcal{G}_* )$ :
\[|| \tilde{\mathcal{G}_* }(\rho + \ker(\mathcal{G}_* )) || = || \mathcal{G}_* (\rho) ||= || \mathcal{G}_* (\rho + \mu) || \leq || \rho + \mu || \]
So $|| \tilde{\mathcal{G}_* }(\rho+ \ker(\mathcal{G}_* )) ||\leq ||\rho + \ker(\mathcal{G}_* )||= \inf_{\mu \in \ker(\mathcal{G}_*)} ||\rho + \mu ||$. Then, since $\mathcal{G}_* $ is idempotent we also have
\[||\rho + \ker(\mathcal{G}_* )||= \inf_{\mu \in \ker(\mathcal{G}_*)} ||\rho + \mu || \leq ||\rho + (\mathcal{G}_* (\rho) - \rho ) || = ||\mathcal{G}_* (\rho) ||=||\tilde{\mathcal{G}_* }(\rho+ \ker(\mathcal{G}_* )) ||, \]
and thus it provides an isometry between $\T(\hi)_G $ and $\T(\hi)^G$.
\end{proof}

In the case of the general (locally compact Hausdorff) group $G$ and its (strongly continuous) representation on $\hi$, on the contrary to $\S(\hi)^G$, the set $\S(\hi)_G$ of $G$-equivalent states always provides a non-trivial state space as some invariant operators are always there, regardless of the compactness of $G$ or the specifics of the representation. This further justifies our starting point of putting the invariance requirement on the observables, rather than the states. \cite{lov1,Loveridge2020a,lov4,miyadera2020quantum,loveridge2017relativity}.

The $G$-equivalent state space $\S(\hirs)_G$ can be seen as the operational analog of the (unit vectors in the) physical Hilbert space of the perspective-neutral (PN) approach \cite{de2021perspective}, which in the compact $G$ case is defined as the space of invariant Hilbert space vectors. We also note that our strategy, clearly motivated by principles very different from the setup of constrained quantization, has the benefit of avoiding the difficulties related to the construction of the physical Hilbert space for non-compact groups.

\newpage
\section{Quantum reference frames}\label{sec:qrfs}

We now give an operational definition of a quantum reference frame.

\begin{definition}\label{def:qrfs}
A \emph{quantum reference frame} $\R$ is a covariant POVM, understood as a map
\begin{equation}
    \E_\R: \mathcal{B}(\Sigma_\R) \to B(\hir).
\end{equation}
It is then specified together with its domain and codomain as $G$-spaces and as such is assumed to satisfy the definition of a covariant POVM (\ref{def:imp}). For brevity, a quantum reference frame will often be called a frame, or reference. We call frames \emph{equivalent} if they are unitarily equivalent as covariant POVMs and refer to $\E_\R$ as the \emph{frame-orientation observable}. We distinguish the following types of frames.
{\normalfont
\begin{itemize} 
    \item A frame  $\R$ is called \emph{principal} if $\Sigma_\R$ is principal, \emph{non}-principal otherwise.
    \item A frame  $\R$ is called \emph{sharp} if $\E_\R$ is sharp, \emph{un}sharp otherwise.
    \item A frame  $\R$ is called \emph{ideal} if it is principal and sharp.
    \item A frame  $\R$ is called \emph{localizable} if $\E_\R$ localizable.
    \item A frame  $\R$ is called \emph{complete} if there is no (non-trivial) subgroup $H_0 \subseteq G$ acting trivially on the effects of $\E_\R$, \emph{in}complete otherwise. Such $H_0$ will be called an \emph{isotropy subgroup}.
    \item A frame $\R$ is called a \emph{coherent system frame} if $\E_\R$ is given via a coherent state system \eqref{ex:css}.
\end{itemize}}
\end{definition}

Thus a quantum reference frame is a quantum system equipped with a covariant observable of orientation. The following remarks are in order.

The sample space of a frame $\Sigma_\R$ is homeomorphic to $G/H$ for some closed subgroup $H \subseteq G$. Sharp frames are characterized by the Imprimitivity Theorem, as described in (\ref{ex:covsharp}): upon trivializing the $H$ representation in the given correspondence, they are equivalent to
\begin{equation}
    P_\R(gH): L^2(G/H) \ni f \mapsto \chi_{gH}f \in L^2(G/H).
\end{equation}
Taking $H=\{e\}$ we arrive at the sharp and principal, i.e. ideal, frames, which are then equivalent to the regular representation and the canonical covariant PVM
\begin{equation}
    P_\R(Y): L^2(G) \ni f \mapsto \chi_{Y}f \in L^2(G).
\end{equation}

A very important from our operational perspective is the class of localizable frames. While there is no operational difference between them and the sharp ones in terms of arbitrarily well-localized distributions they may give rise to, localizability is a strictly weaker condition. Since $G/H$ is always metrizable, we can apply the Proposition \ref{prop:locseqgen} to any localizable frame and find a sequence of pure states whose associated probability distribution weakly converges to the Dirac delta centered at any given class $gH$. In the case of the \emph{principal} localizable frames, we will denote such a localizing sequence for $e \in G$ by $\omega_n$, which by covariance gives a localizing sequence centered at any other group element by $\omega_n(g)=g^{-1}\omega_n$.

In the case of localizable frames, the notions of completeness and principality are intimately related, which allows connecting the two different versions of imposing `insensitivity' of the frame with respect to reorientations by elements of a specified subgroup. To see this, consider
a localizable frame and take any $h \in H_0$ in the isotropy subgroup, fix the sample space as $\mathcal{B}(G/H)$ and write $\omega_n$ for the sequence localizing at the identity coset $eH \in G/H$. We then have
\begin{equation}
    \delta_{hH}(X) = \lim_{n \to \infty} \mu_{\omega_n(hH)}^{\E_\R}(X) = \lim_{n \to \infty}\tr[h^{-1}.\omega_n \E_\R] %=\lim_{n \to \infty} \tr[\omega_n h^{-1}.\E_\R] 
    = \lim_{n \to \infty} \tr[\omega_n \E_\R]  = \delta_{eH}(X),
\end{equation}
so that $hH=eH$ for any $h \in H_0$ and we can conclude that $H_0 \subseteq H$. In particular, localizable principal frames are complete, and in general, a localizable frame-orientation observable on $\Sigma_\R \cong G/H$ factorizes through $B(\hir)^H$, i.e. we can write
\begin{equation}
    \E_\R : \mathcal{B}(G/H) \to B(\hir)^H \hookrightarrow B(\hir).
\end{equation}

In other approaches, the quantum reference frames are often understood in terms of coherent state systems\footnote{Notice however, that we do not consider distributional coherent state systems, contrary to \cite{de2021perspective}.}. In our classification, they are then always coherent system frames, with the corresponding frame-orientation observables given by
\begin{equation}
\E^{\phi}_\R(X) = \frac{1}{\lambda}\int_X \dyad{\phi(g)} d \mu (g),
\end{equation}
where $\{\ket{\phi(g)} = U_\R(g)\ket{\phi(e)}:g \in G\}$ is an orbit of a cyclic vector $\ket{\phi(e)}$ and $\mu$ the Haar measure (\ref{ex:css}). The definition of completeness presented here resembles that given for coherent system frames in \cite{de2021perspective}, however, the relationship between the two is convoluted -- it does not seem either one implies the other, except in trivial cases. The definition of an \emph{ideal} frame above generalizes the one given in \cite{de2021perspective}, which is recovered as an ideal coherent system frame, while the two notions coincide in the case of a countable group. Indeed, if a coherent system POVM is principal, so given on $G$, and sharp, we get $\lambda=1$ and $\ip{\phi(g)}{\phi(g')} = \delta(g,g')$, so that the coherent states can be `perfectly distinguished', which is the defining property of an ideal frame as in \cite{de2021perspective}. However, as noted there, when this is the case the map $G \ni g \mapsto \ket{\phi(g)} \in P(\hir)$ is \emph{invertible}. But notice that, since $\hir$ is assumed separable, without entering the realm of distributional states this can only be the case if $G$ is (at most) countable, as otherwise the orbit $\{\ket{\phi(g)} = U_\R(g)\ket{\phi(e)}|g \in G\}$ would provide an uncountable orthonormal basis.

To the best of our knowledge, the localizability properties of coherent system POVMs have not yet been systematically studied. However, we already see that such POVMs can be sharp only in the very simple setup of a countable group. We avoid the use of the coherent state systems in our formalism since, strictly speaking, it is only when the coherent system POVM is sharp that there is operational justification for understanding $\ket{\phi(g)}$ as representing the frame being `oriented at $g$', and as we have seen this is very restrictive. Moreover, various physically motivated examples of coherent state systems are known \emph{not} to be localizable \cite{Beneduci2013-sm}. Our general feeling is that, when it comes to the operational interpretation of the associated POVMs, the setup of \emph{continuous} groups and their representations on \emph{separable} Hilbert spaces does not combine well with the concept of a coherent state system that maps one into the other.

Various other more or less obvious definitions of frames can be constructed by requiring specific properties from the $G$ action on $\Sigma$ (e.g. effectiveness) or from the POVM (e.g. informational-completeness).

\newpage
\section{Framed descriptions}\label{sec:frmdes}

After the frame observable $\E_\R$ has been chosen, the set of the effects that can be applied to the composite systems $\R\otimes\S$ is operationally constrained. Hence the following definition.

\begin{definition}
    Given a frame $\R$ and a system $\S$, we will refer to the operationally $\Eff(\hir \otimes \his)_{\E_\R}$-equivalent trace-class operators in~$\T(\hirs)$, where
    \begin{align*}
         \Eff(\hir \otimes \his)_{\E_\R} := \rm{conv}\left\{\E_\R(X) \otimes F_\S \hspace{3pt} | \hspace{3pt} X \in \mathcal{F}(\Sigma_\R), F_\S \in \Eff(\his)\right\}^{cl},
     \end{align*}
    with $\rm{conv}$ denoting the convex hull, i.e. the (ultraweak) closure of the set of convex combinations, as \emph{$\E_\R$-equivalent}, while the $\E_\R$-equivalence classes of \emph{states} will be called \emph{$\R$-framed}. Elements of the Banach space $B(\hir \otimes \his)_{\E_\R} = \rm{span}(\Eff(\hir \otimes \his)_{\E_\R})^{cl}$ will be referred to as \emph{$\R$-framed operators}, while the effects in there as the \emph{$\R$-framed effects}.
\end{definition}

The $\R$-framed states are then precisely those that can be distinguished by the observables on the composite system $\R\otimes\S$ that respect the choice of the frame observable. The Propositions \ref{prop:iml} and \ref{prop:statespace} immediately give the following.

\begin{proposition}
    We have the following Banach space isomorphism
    \[
    \left[\T(\hirs)/\hspace{-3pt}\sim_{\E_\R}\right]^\bigstar \cong B(\hir \otimes \his)_{\E_\R}.
    \]
    Moreover, the set of $\R$-framed states
    \[
    \S(\his)^{\E_\R} := \S(\hirs)/\hspace{-3pt}\sim_{\E_\R}
    \]
    is a state space in $\T(\hirs)^{\rm{sa}}/\hspace{-3pt}\sim_{\E_\R}$.
\end{proposition}

Thus the $\R$-framed operators separate the $\R$-framed states. The Banach space of framed operators can thus be thought of as an analog of the full system algebra in the standard approach upon the specification of the reference frame, understood as a reference system \emph{and} the frame observable, with respect to which the system $\S$ is being described. However, the invariance principle (III) has not yet been incorporated. We refer to this setup as the \emph{framed description} of the system $\S$ with respect to the frame $\R$.

\newpage
\section{Relational descriptions}

We now combine the framed and invariant descriptions to give an operationally motivated definition of \emph{relational} states and operators. The operationality (I) and invariance (III) principles impose that the \emph{relative} description of the system $\S$ with respect to the frame $\R$ should be given up to operational equivalence with respect to the \emph{invariant} and \emph{framed} effects. We then define.

\begin{definition}\label{def:rel}
Given a frame $\R$ and a system $\S$, we will refer to operationally $\Eff(\his)^G_{\E_\R}$-equivalent trace-class operators on $\T(\hirs)$, where
\[
    \Eff(\hir \otimes \his)^G_{\E_\R} := \Eff(\hir \otimes \his)_{\E_\R} \cap \Eff(\hirs)^G,
\]
as \emph{$(\E_\R,G)$-equivalent}, while the $(\E_\R,G)$-equivalence classes of \emph{states} will be called \emph{$\R$-relational}. Elements of the Banach space
\[
B(\hir \otimes \his)_{\E_\R}^G := \rm{span}\left\{\Eff(\hir \otimes \his)^G_{\E_\R}\right\}^{cl}
\]
will be called $\R$-relational operators.
\end{definition}

As usual, Propositions \ref{prop:iml} and \ref{prop:statespace} give the following.

\begin{proposition}
    We have the following Banach space isomorphism
    \[
    \left[\T(\hirs)/\hspace{-3pt}\sim_{(\E_\R,G)}\right]^\bigstar \cong B(\his)^G_{\E_\R}.
    \]
    Moreover, the set of $\R$-relational states
    \[
    \S(\his)_G^{\E_\R} := \S(\hirs)/\hspace{-3pt}\sim_{(\E_\R,G)}
    \]
    is a state space in $\T(\hirs)^{\rm{sa}}/\hspace{-3pt}\sim_{(\E_\R,G)}$.
\end{proposition}

Thus the $\R$-relational operators separate the $\R$-relational states. The Banach space of $\R$-relational operators can thus be thought of as an analog of $B(\his)$ upon the specification of the reference frame, with respect to which the system $\S$ is being described \emph{and} imposing the gauge-invariance of principle (III). We refer to this setup as the \emph{$\R$-relational description} of the system $\S$ with respect to the frame $\R$.

Since $\E_\R$ is \emph{covariant} and the action of $G$ on $\his$ arbitrary, it may seem hopeless to look for \emph{invariant} effects among the $\R$-framed ones, i.e. those generated by $\E_\R(X) \otimes F_\S$. The fact that $\Eff(\his)^G_{\E_\R}$ is full of non-trivial effects for arbitrary frames is given by the \emph{$\R$-relativization maps} that we now introduce.

\newpage
\section{Relativization map}\label{sec:yenmap}

To motivate the $\Y$ construction (below), consider first a localizable principal frame on a finite group $G$. The frame observable $\E_\R$ is then given as a map $G \to \Eff(\hir)$ and thus an arbitrary framed effect can be written as a sum
\begin{equation}
A_\alpha := \sum_{g \in G} \E_\R(g) \otimes \alpha(g),
\end{equation}

where $\alpha: G \to \Eff(\his)$ is any function, with the framed effects of the form $\E_\R(Y) \otimes F_\S$ given by fixing $\alpha(g) = \chi_{Y}(g)F_\S$. Acting with $h \in G$ on a framed effect gives
\begin{equation}
h.A_\alpha = h.\sum_{g \in G} \E_\R(g) \otimes \alpha(g) = \sum_{g \in G} \E_\R(hg) \otimes h.\alpha(g) = \sum_{g' \in G} \E_\R(g') \otimes h.\alpha(h^{-1}g'),
\end{equation}
where we put $g'=hg$. Now notice that, since $\E_\R$ is localizable and $G$ finite, for any $g \in G$ we have $\omega(g)$ such that $\tr[\omega(g)\E_\R(g)] = 1$ and thus equality $h.A_\alpha=A_\alpha$ needs to hold term by term. It then amounts to \emph{equivariance} of $\alpha: G \to \Eff(\his)$, so that in the case of a finite group the only invariant framed effects are of the form
\begin{equation}
    A = \sum_{g \in G}\E_\R(g) \otimes g.F_\S
\end{equation}
for some $F_\S \equiv \alpha(e) \in \Eff(\his)$. We then see that by summing over the whole group and \emph{using} the covariance of $\E_\R$, we can generate all the invariant framed effects. This construction can be understood as a map
\begin{equation}
   \Y^\R: F_\S \mapsto \sum_{g \in G}\E_\R(g) \otimes g.F_\S.
\end{equation}

We get the following result.
\begin{proposition}\label{prop:releqinvfrm}
    Let $G$ be a finite group, $\R$ a localizable principal frame for $G$ and write $\Eff(\his)^\R$ to stand for $\Y^\R(\Eff(\his))$. Then for any $\S$ we have
    \[
    \Eff(\his)^\R = \Eff(\hirs)^G \cap \Eff(\hir \otimes \his)_{\E_\R}.
    \]
\end{proposition}

Thus, in this simple case, the image of the $\Y^\R$ map exhausts the set of invariant framed effects. It turns out that this construction can be generalized to arbitrary \emph{principal} frames.\footnote{Generalizing the $\Y$ construction to arbitrary frames is the subject of ongoing work. Results concerning finite groups have recently been achieved in \cite{homogyen2023}. See chapter \ref{ch:discuss} for a discussion of this research direction.}
\begin{definition}
    Given a principal frame $\R$ and a quantum system $\S$, the map
\begin{equation}
\Y^\R: \Eff(\his) \ni F_\S \mapsto \int_G d\E_\R(g)\otimes U_\S(g)F_\S U_\S(g)^* \in \Eff(\hirs)^G,
\end{equation}
where the integral understood as in \cite{lov1}, will be referred to as \emph{$\R$-relativization map}.
\end{definition}

One readily verifies that $\Y^\R(\Eff(\his)) \subseteq \Eff(\hirs)^G$. Indeed, a simple change of variables gives
\[
h.\Y^\R(F_\S) = h.\int_G d\E_\R(g) \otimes g.F_\S = \int_G d\E_\R(hg) \otimes hg.F_\S = \Y^\R(F_\S).
\]

By composing \emph{arbitrary} POVMs on the system $\S$ with the $\Y^\R$ map we get\footnote{We note the following curious fact. The notion of a POVM on a system modeled by $\his = \mathbb{C}$ coincides with that of a non-negative countably additive measure. If we now consider $\his \cong \hir \cong \mathbb{C}$ the standard definition of the convolution of measures on a group is recovered as a special case of the $\Y$ construction.}
\begin{equation}
    \E_\S \mapsto \E_\S * \E_\R := \Y^\R \circ \E_\S: \mathcal{F}(\Sigma_\S) \to \Eff(\hirs)^G \cap \Eff(\hirs)_{\E_\R},
\end{equation}
which assigns \emph{invariant} observables on $\R \otimes \S$ to arbitrary observables on $\S$. The $\Y^\R$ construction can be extended to a map $\Y^\R: B(\his) \to B(\hirs)^G$, and as such, regardless of the choice of the frame-orientation observable $\E_\R$, except being linear and bounded, it is also unital, completely positive, and normal, so a \emph{quantum channel} \cite{lov1}.\footnote{Interestingly, the definition of a relational Dirac observable as in \cite{de2021perspective} is recovered for $g \in G$ by taking $ \Y^{\R}(g.A_S)$ with $\E_\R$ taken to be the coherent system POVM of the frame. Thus in this case the set of relational Dirac observables and relativized operators are the same, while the latter are defined more generally. The twirl map also arises as a special case of this construction, when $\hir$ is taken to be the complex numbers $\mathbb{C}$, with (necessarily) trivial $G$ action. Indeed, the notion of a covariant POVM then coincides with that of a normalized invariant measure, and thus there is exactly one when $G$ is compact, and none otherwise.} The fact that the relativized effects are framed, although seemingly apparent, is proved in (\ref{prop:relativeareframed}).

Now consider a \emph{pair of frames}. The relativization procedure gives an invariant observable that can be easily understood as an observable of their \emph{relative} orientation.

\begin{definition}\label{def:relorobs}
    Given a pair of frames $\R_1$ and $\R_2$ the observable
    \[
    \E_2 * \E_1 = \Y^{\R_1} \circ \E_2 = \int_G d\E_1(g) \otimes g.\E_2(\cdot)
    \]
    will be called an \emph{observable of relative orientation of $\R_2$ with respect to $\R_1$}.
\end{definition}

\begin{proposition}\label{prop:relorobs}
For a pair of localizable frames $\R_1$ and $\R_2$ and corresponding localizing sequences $\omega_n$ and $\rho_m$, writing $\Omega_{n,m}(h) := \omega_n \otimes h^{-1}.\rho_m$ we have
\[
\lim_{n,m\to\infty} \mu_{\Omega_{n,m}(h)}^{\E_2 * \E_1} = \delta_h.
\]
\end{proposition}

\begin{proof}
    We calculate
    \begin{align*}
    \lim_{n,m\to\infty} \tr[(\omega_n \otimes h^{-1}.\rho_m) \int_G d\E_1(g) \otimes g.\E_2(X)] =
    & \lim_{m\to\infty} \tr[h^{-1}.\rho_m \E_2(X)]\\
    = \lim_{m\to\infty} \tr[\rho_m \E_2(h^{-1}.X)]
    &= \delta_e(h^{-1}.X) = \delta_h(X),
    \end{align*}
    where we have used \ref{prop:locseqgen} twice.
\end{proof}

Thus when one frame is localized at $e \in G$, and the other at $h \in G$, the relative orientation observable will give probability distribution localized at $h$. Note that since $\E_2 * \E_1$ is invariant, we could just as well evaluate it on $(h.\omega_n \otimes \rho_m)$, with the same result. Notice also, that if we instead relativized $\E_1$ with respect to $\E_2$ (by taking $\E_1 * \E_2$), the probability distribution of the relative orientation observable evaluated on $\omega_n \otimes h^{-1}.\rho_m$ would be localized at $h^{-1}$, as we would then be measuring the orientation of $\R_1$ with respect to $\R_2$. Indeed, a simple calculation gives
\begin{equation}
\E_2 * \E_1 (X) = SWAP_{1,2} \circ \E_1 * \E_2 (X^{-1}),
\end{equation}
where $SWAP_{1,2}$ takes care of switching the tensor product factors as in \cite{de2020quantum}, i.e. for $A_1 \otimes A_2 \in B(\hio \otimes \hit)$ we have
\[
SWAP_{1,2}(A_1 \otimes A_2) = A_2 \otimes A_1.
\]

\newpage
\section{Relative descriptions}

We now provide a definition of \emph{relative} states and operators given with respect to the \emph{relativized} effects.

\begin{definition}\label{def:Yrel}
Given a frame $\R$ and a system $\S$, we will refer to the operationally $\Eff(\his)^\R$-equivalent trace-class operators on $\T(\hirs)$, where
\[
    \Eff(\his)^\R := \Y^\R(\Eff(\his)) \subset \Eff(\hirs)^G,
\]
as \emph{$\R$-equivalent}, while the $\R$-equivalence classes of \emph{states} will be called \emph{$\R$-relative}.
\end{definition} 

The $\R$-relative states are then precisely classes of states that can be distinguished by the $\R$-relativized effects.

\begin{definition}\label{def:Yrelobs}
Given a frame $\R$ and a system $\S$, elements of the Banach space
\[
B(\his)^\R :=  \rm{span}\{\Eff(\his)^\R\}^{cl} = \Y^\R(B(\his))^{cl} \subseteq B(\hirs)^G
\]
will be called $\R$-relative operators.
\end{definition}

The relative operators are then invariant, as expected. The Propositions \ref{prop:iml}, \ref{prop:statespace} and \ref{generalst} (for $F=\Y^\R$) then give the following.

\begin{proposition}
    We have the following Banach space isomorphism
    \[
    \left[\T(\hirs)/\hspace{-3pt}\sim_{\R}\right]^\bigstar \cong B(\his)^{\R}.
    \]
    Moreover, the set of $\R$-relative states
    \[
    \S(\his)^{\R} := \S(\hirs)/\hspace{-3pt}\sim_{\R}
    \]
    is a state space in $\T(\hirs)^{\rm{sa}}/\hspace{-3pt}\sim_{\R}$ and we have a state space isomorphism
    \[
        \S(\his)^\R \cong \Y^\R_*(\S(\hirs)_G).
    \]
\end{proposition}

In the last claim, we have used the fact that since the image of $\Y^\R$ is in $B(\hirs)^G$, the domain of the predual map $\Y^\R_*$ map is $\T(\hirs)_G$. This characterization will be used in the sequel, so it is worth introducing appropriate notation. We will write $\Omega^\R \in \S(\his)^\R$, and  $[\Omega]_\R$ for the corresponding equivalence class of states on the composite system. We then identify $[\Omega]_\R \simeq \Y^{\R}_*[\Omega]_G \equiv \Omega^\R$, for any $\Omega \in \T(\hirs)$.

Recalling again the Banach space isomorphism $\T(\hi)^\bigstar \cong B(\hi)$ of the standard quantum mechanical setup, the Banach space of relative operators can be thought of as the invariant analog of the full system algebra in the standard approach arising upon the specification of the reference frame, understood as a reference system \emph{and} the frame observable, with respect to which the system $\S$ is described via the \emph{relativization} procedure. We refer to this setup as the \emph{relative description} of the system $\S$ with respect to the frame $\R$, and carefully examine its relation to the standard setup in the next chapter \ref{ch:FrmCond}.

\newpage
\section{Relational and relative}

Crucially for our operational interpretation, the $\R$-relative operators are not only invariant but also $\R$-framed, so relational.

\begin{proposition}\label{prop:relativeareframed}
Given a frame $\R$ and a system $\S$, we have the following inclusion
\[
B(\his)^\R \subseteq B(\hir \otimes \his)_{\E_\R}.
\]
\end{proposition} 

\begin{proof}
We will show that $\Eff(\his)^\R \subseteq \Eff(\hir \otimes \his)_{\E_\R}$, i.e. that the $\E_\R$-equivalence is \emph{finer} than $\R$-equivalence, i.e. that for $\Omega, \Omega' \in \T(\hirs)$ we have
    \[
    \Omega \sim_{\E_\R} \Omega' \Rightarrow \Omega \sim_\R \Omega',
    \]
from which the inclusion of the dual Banach spaces $B(\his)^\R \subseteq B(\hir \otimes \his)_{\E_\R}$ follows. Writing $P^{\E_\S}: \mathcal{B}(\Sigma) \to B(\his)$ for the PVM associated to an arbitrary $\E_\S \in \Eff(\his)$ via the spectral theorem we calculate
    \begin{align*}
    \tr[\Omega \int_G d\E_\R(g) \otimes g.\E_\S] &=
    \int_{G\times \Sigma} d\mu^{\E_\R \otimes g.P^{\E_\S}}_{\Omega}(g,x)\\
    &= \int_{G\times \Sigma} d\mu^{\E_\R \otimes g.P^{\E_\S}}_{\Omega'}(g,x)
     = \tr[\Omega' \int_G d\E_\R(g) \otimes g.\E_\S],
    \end{align*}
    where the $\E_\R$-equivalence of $\Omega$ and $\Omega'$ was invoked in the second equality.
\end{proof}

Given the invariance of the image of $\Y^\R$ can conclude the following.

\begin{proposition}
    For any principal frame $\R$ and an arbitrary system $\S$ we have
\[
    \B(\his)^\R \subseteq B(\his)^G_{\E_\R}.
\]
\end{proposition}
Knowing that those sets are equal in the case of a finite group $G$ and a localizable frame, we leave the following as a \emph{conjecture}.

\begin{proposition}[conjecture]\label{conj:rel=rel}
    Given a localizable principal frame $\R$ and a system $\S$, we have a Banach space isomorphism
    \[
        B(\his)^\R \cong B(\his)^G_{\E_\R}
    \]
\end{proposition}

We reflect upon this and possible lines of attack in chapter \ref{ch:discuss}.

\newpage
\section{Twofold meaning of invariance}\label{sec:dblinv}

Having defined the $\R$-relative description as given by requiring operational equivalence with respect to the relativized effects $\Y^\R(\Eff(\his))$, we may ask if there are effects on $\S$ that are somehow insensitive to the choice of the external reference. To this end, notice that by relativizing an \emph{invariant} effect $F_\S \in \Eff(\his)^G$ we get
\begin{equation}
    \Y^\R(F_\S) = \int_G d\E_\R(g) \otimes g.F_\S = \int_G d\E_\R(g) \otimes F_\S = \mathbb{1}_\R \otimes F_\S,
\end{equation}
and thus we have
\begin{equation}
    \Y^\R(\Eff(\his)^G) \cong \Eff(\his)^G,
\end{equation}
for \emph{any frame $\R$}. The space of invariant effects on $\his$ is then independent of the choice of the external reference frame. Under this interpretation, the setup of $B(\his)^G$ can be called the \emph{perspective-independent description of $\S$}. In this sense, it can be thought of as `objective' or `consensual'.

An alternative interpretation of the invariant description of a single system is possible. Indeed, instead of assuming the \emph{non}-relational, or \emph{absolute} description of $\S$ to exist on its own and serve as a basis for generating the operational and invariant, relative to an \emph{external} frame, descriptions via the $\Y$ construction, we may think of the frame being chosen \emph{internally}, as \emph{subsystem} of a given bigger system, call it $\T$ (total). Such a choice can only be possible if the representation $U_\T:G \to B(\hi_\T)$ can be decomposed as $U_\R \otimes U_\S$ with the corresponding decomposition of $\T$ in terms of its Hilbert space, i.e. $\hi_\T \cong \hir \otimes \his$. Requiring such tensor product decomposition of the action can be understood as assuring that the subsystems may be considered independent elements of our framework. Depending on the representation $U_\T$, such a decomposition may be highly non-unique or impossible. Crucially, upon any choice of an \emph{internal} reference, the relative description will be phrased in terms of the invariant quantities on $\T$ since we have $B(\his)^\R \subseteq B(\hi_\T)^G$ regardless of the choice of the frame-system decomposition. Recall also that the domain of $\Y^\R_*$ is the space of $G$-equivalent trace-class operators so that a \emph{global} state $[\Omega]_G \in \S(\hi_\T)_G$ is just enough to construct all the internal relative states $\Y^\R_*[\Omega]_G$, whatever the choice of the internal frame $\R$. The invariant algebra $B(\hi_\T)^G$ thus contains, in this sense, all possible such \emph{internally-relative} descriptions. Under this interpretation, the invariant description of a system will be referred to as \emph{global}. We may also consider the relations between the relative descriptions corresponding to different choices of such internal reference frames. We address this in chapter \ref{ch:FrChM} by providing frame-change maps.

Notice here that these two interpretations are not in conflict -- in fact, it seems reasonable for the global description to be `objective' in the sense described above. These matters deserve a separate, more philosophically oriented treatment, which will be pursued elsewhere.

\myemptypage

\chapter{Conditioning descriptions}\label{ch:FrmCond}
In this chapter, we consider situations in which the state preparation of the frame is known, which is realized via a \emph{restriction map} (see below) that maps effects on the composite system $\R \otimes \S$ to that on $\S$ alone. In the case of conditioning relative description, it amounts to what can be seen as \emph{probabilistic gauge fixing} -- the resulting description is the same for any frame's states that give the same probability distributions upon evaluation of the frame-orientation observables, and amount to a weighted averaging of the quantities on $\S$. The twirling procedure is recovered in the case of an invariant frame's state. in the case of a localizable frame, we can consider descriptions conditioned upon highly localized frame states. We then find operational agreement between the non-relational description of the system $\S$ in terms of $\his$ alone. Indeed, the relative states can be approximated arbitrarily well by the relative states conditioned upon highly localized frame states. This perspective is applied to a description of a quantum measurement setup a'la quantum measurement theory aligned with the presented framework, with the role of the pointer observable played by the frame-orientation observable. We find that when the measured observable on $\S$ is covariant, the relative orientation observable satisfied a condition stronger than the probability reproducibility property, respecting the symmetry structure of the setup. Lastly, we consider conditioning the invariant description with respect to an internal frame, which will be crucial for the rigorous understanding of the `attaching a frame state' procedure as part of the definition of frame-change maps that follows.

Since the $\Y$ construction is crucial for our considerations in the sequel, \textbf{in what follows we restrict to principal frames.} From now on the terms \emph{frame}, or \emph{reference}, will thus refer to a \emph{principal quantum reference frame} as given in \ref{def:qrfs}. The sample space of the frames considered below will then always be taken as $G$ itself. We see this restriction as one of the main deficiencies of the presented framework. It has recently been addressed in the context of a finite group \cite{homogyen2023}, and we hope to extend the results achieved there to more general situations soon. We briefly reflect upon this research direction in~chapter \ref{ch:discuss}.

\newpage
\section{Restriction map}

We begin by recalling the restriction map (e.g. \cite{loveridge2017relativity}) that allows conditioning observables on the system plus reference with a specified state of the reference.

\begin{definition}
Let $\omega \in \S(\hir)$ be any state of the reference. Then the \emph{$\omega$-restriction map} is given by the continuous linear extension of the following
\begin{equation}\label{eq:res}
   \Gamma_{\omega}: A_R \otimes A_\S \mapsto \tr[\omega A_R] A_\S.
\end{equation}
\end{definition}

For all $A_\R \in B(\hir)$, $A_\S \in B(\his)$ and $\rho \in \mathcal{S}(\his)$ we then have
\begin{equation}
    \tr[\rho\Gamma_{\omega}(A)] = \tr[(\omega \otimes \rho )A].
\end{equation}

Besides being linear and bounded, it is normal, (completely) positive, and trace-preserving. The predual map $\mathcal{V}_{\omega}:= (\Gamma_\omega)_*$ is the embedding
\begin{equation}
\mathcal{V}_{\omega} : \rho  \mapsto \omega \otimes \rho,
\end{equation}
so that we have $\tr[\rho \Gamma_{\omega}(A)] = 
\tr[\omega \otimes \rho A]$ for all $\rho \in \S(\his)$.

The $\omega$-restriction map is understood as conditioning the description of a composite system $\R \otimes \S$ upon a particular choice of the state of the reference. The description of $\S$ arising as $\omega$-conditioning effects on $\R \otimes \S$ can be understood as \emph{probabilistically gauged-fixed}. We distinguish two different scenarios in which this can be done, and explore them separately.

Firstly, we may be interested in conditioning the \emph{relative description} of $\S$ with respect to a state of the reference $\omega \in \S(\hir)$, in which case we would be restricting the relative operators. The corresponding description of $\S$ will be referred to as \emph{$\omega$-conditioned $\R$-relative}. As we will see, the corresponding states take a particularly simple form -- they can be written as
\begin{equation}
    \int_G U_\S(g)^*\rho U_\S(g) d\mu _{\omega}^{\E_\R}(g) \in \S(\his),
\end{equation}
with $\rho \in \S(\his)$. We interpret them as states of $\S$ ``smeared'' with respect to the probabilistic gauge fixing corresponding to $p^{\E_\R}_\omega(g)$. This setup will be useful for analyzing the relation between the framework described here and the standard non-relational quantum mechanics. We find that when the relative description is conditioned upon the highly localized state of a localizable reference frame, the usual description in terms of the whole $B(\his)$ is recovered to arbitrary precision.

Secondly, we may wish to consider a \emph{global} description in terms of $B(\hi_\T)^G$ being conditioned upon a state of the subsystem $\R$ chosen as an \emph{internal} reference, i.e. we have $\hi_\T \cong \hirs$, with a suitably decomposed $G$-action $U_\T = U_\R \otimes U_\S$. Such restricted effects can then be relativized accordingly to the chosen decomposition and frame-orientation observable $\E_\R$. The corresponding description of $\S$ will be referred to as \emph{$\R$-relativized $\omega$-restricted}. The corresponding states are referred to as $\omega$-\emph{lifted} $\R$-relative states since they take the following simple form
\begin{equation}
    [\omega \otimes \Omega^\R]_G \in \S(\hirs)_G,
\end{equation}
with $\Omega^\R \in \S(\his)_\R$. They can be compared to the states in the image of the inverse reduction map of \cite{de2021perspective}. This setup, and the lifting construction, will be useful for the frame-change maps which we discuss in the next chapter \ref{ch:FrChM}.

\newpage
\section{Conditioned relative descriptions}

We begin by the following definition.
\begin{definition}
The map
\[
\Y^\R_\omega:= \Gamma_\omega \circ \Y^\R: B(\his) \to B(\his),
\]
will be called the \emph{$\omega$-conditioned $\R$-relativization map}.
\end{definition}

As a composition of such maps, the $\omega$-conditioned $\R$-relativization map is unital, normal, and (completely) positive. The image of this map, when applied to the effects $\Eff(\his)$, consists of the \emph{relativized} effects that have then been \emph{conditioned} upon the chosen state of the reference $\omega \in \S(\hir)$. This is a conceptually clear restriction on the set of available effects. It can be understood as a next step in `specifying the reference' -- now we do not only clarify \emph{which} quantum system ($\hir$) and \emph{how} ($\E_\R$) is going to be used for the description of another system $\S$, but also constrain its particular \emph{preparation} $\omega \in \S(\hir)$. As we will see in the sequel, upon restriction with a highly localized state of the reference, the standard description of non-relational quantum physics, as given in terms of $B(\his)$ alone, is recovered up to arbitrary precision \cite{lov1}. We introduce the \emph{$\omega$-conditioned $\R$-relative description} of $\S$ with respect to $\R$ \emph{in the (fixed) state $\omega \in \S(\hir)$} as given with respect to the image of the $\Y^\R_\omega$~map.

\begin{definition}\label{def:Yrelcond}
Given a frame $\R$ and a system $\S$, we will refer to operationally $\Eff(\his)^\R_\omega$-equivalent trace-class operators on $\T(\his)$, where
\[
    \Eff(\his)^\R_\omega := \Y^\R_\omega(\Eff(\his)) \subset \Eff(\his),
\]
as \emph{($\R,\omega$)-equivalent}, while the \emph{($\R,\omega$)}-equivalence classes of \emph{states} in $\S(\his)$ will be called \emph{$\omega$-product $\R$-relative states}.
\end{definition} 

The $\omega$-product $\R$-relative states are then precisely classes of states that can be distinguished by the $\R$-relativized and $\omega$-conditioned operators.

\begin{definition}\label{def:Yrelobscond}
Given a frame $\R$ and a system $\S$, elements of the Banach space
\[
B(\his)^\R_\omega := \rm{span}\{\Eff(\his)^\R_\omega\}^{cl} = \Y^\R_\omega(B(\his))^{cl} \subseteq B(\his)
\]
will be called \emph{$\omega$-conditioned $\R$-relative operators}.
\end{definition}

The operator $\Y^\R_\omega(A_\S)$ can be seen as \emph{weighted} average of the operators on the orbit of $A_\S$  with respect to the probability distribution of the frame-orientation observable as we have
\begin{equation}
\Y^\R_\omega(A_\S) = \int_G d\mu^{\E_\R}_\omega(g) U_\S(g)A_\S U_\S^*(g).
\end{equation}
Thus they are defined up to the operational equivalence of $\omega$ with respect to the effects of $\E_\R$. Indeed, if $\tr[\omega \E(X)] = \tr[\omega' \E(X)]$ for all $X \in \mathcal{B}(G)$, we have $\Y^\R_\omega(A_\S)=\Y^\R_{\omega'}(A_\S)$ for all $A_\S \in B(\his)$. Similar formula holds for $\omega$-product $\R$-relative states, generalizing the $G$-twirling procedures. Indeed, taking $G$-compact for any invariant $\omega \in \S(\hir)$ we get $\Y^\R_\omega=\mathcal{G}$ (see e.g. \cite{lov1}), providing operational justification for the use of the $G$-twirl in the information-theoretic considerations of QRFs \cite{brs,castro2021relative}. The Propositions \ref{prop:iml}, \ref{prop:statespace} and \ref{generalst} (for $F=\Y^\R_\omega$) give the following.

\begin{proposition}
    We have the following Banach space isomorphism
    \[
    \left[\T(\his)/\hspace{-3pt}\sim_{\Eff(\his)^\R_\omega}\right]^\bigstar \cong B(\his)_\omega^{\R},
    \]
    Moreover, the set of $\omega$-product $\R$-relative states as given below is a state space
    \[
    \S(\his)^{\R}_\omega := \S(\his)/\hspace{-3pt}\sim_{\Eff(\his)^\R_\omega} \cong (\Y^\R_\omega)_*(\S(\his)) \subset \T(\his)^{\rm{sa}}/\hspace{-3pt}\sim_{\Eff(\his)^\R_\omega}.
    \]
\end{proposition}

\section{Product-relative states}

We will now examine some special cases of product-relative states. Notice first that they take the following form
\begin{equation}
    \rho ^{(\omega)} := \Y^\R_*(\omega \otimes \rho ) = \int_G d\mu _{\omega}^{\E_\R}(g) g.\rho,
\end{equation}
for some $\omega \in \S(\hir)$, where the measure $\mu_{\omega}^{\E_\R}(X)$ is given as usual by $\mu_{\omega}^{\E_\R}(X)= \tr[\E_\R(X)\omega]$, and we have introduced the notation $\rho ^{(\omega)}$ to indicate the particular frame's state that is used for conditioning. The product-relative states arise from product states on the global system and in some sense generalize the alignable states of \cite{krumm2021quantum}. Notice here, that the state $\rho ^{(\omega)}$ depends only on the $\E_\R$ equivalence class of the state $\omega \in \S(\hir)$. We note the following.

\begin{proposition}
Product-relative states satisfy the following symmetry condition
\begin{equation}
    \rho^{(h.\omega)} = (h^{-1}.\rho)^{(\omega)}
\end{equation}
\end{proposition}

\begin{proof}
We calculate:
\begin{align*}
    \rho^{(h.\omega)} &= \int_G d\mu^{\E_\R}_{h.\omega}(g)g.\rho = \int_Gd\mu^{\E_\R}_{\omega}(hg)g.\rho\\
    &= \int_Gd\mu^{\E_\R}_{\omega}(g')(h^{-1}g').\rho = \int_Gd\mu^{\E_\R}_{\omega}(g')g'.(h^{-1}.\rho) = (h^{-1}.\rho)^{(\omega)}.
\end{align*}
We have changed the integration variable $g'=hg$ and used the fact that $\mu^{\E_\R}_{h.\omega}(g) = \mu^{\E_\R}_{\omega}(hg)$ which follows directly from the covariance of $\E_\R$:
\begin{displaymath}
\mu^{\E_\R}_{h.\omega}(X) = tr[\E_\R(X)h.\omega] = tr[h.\E_\R(X)\omega] = tr[\E_\R(h.X)\omega] = \mu^{\E_\R}_{\omega}(h.X).
\end{displaymath}
\end{proof}

Thus `rotating' the frame by a group element $h \in G$ is equivalent in the sense of the product-relative state description to keeping it fixed but `rotating' the system instead with $h^{-1} \in G$. This is again just a consequence of the fact that the image of $\Y^\R$ is invariant with respect to the diagonal action on $B(\hirs)$ and thus for any $A_\S \in B(\his)$ we have
\begin{equation}
[U_\R(h)(\_)U_\R(h)^* \otimes \mathbb{1}_\S]\Y^\R(A_\S) = [\mathbb{1}_\R \otimes U_\S(h^{-1})(\_)U_\S(h^{-1})^*]\Y^\R(A_\S).
\end{equation}

The next proposition demonstrates the intuitively plausible claim that invariant system states are defined without reference to an external frame or, more precisely, independent from the chosen reference. 
\begin{proposition}
Let $\rho \in \S(\his)^G$. Then $\rho ^{(\omega)} = \rho$ for and any choice of frame $\R$,  and any state $\omega \in \S(\his)$.
\end{proposition}

\begin{proof}
    We calculate:
    \[
        \rho ^{(\omega)} = \int_G d\mu _{\omega}^{\E_\R}(g) U_\S(g)^*\rho  U_\S(g) = \int_G d\mu _{\omega}^{\E_\R}(g) \rho = \rho,
    \]
    where we only needed to use invariance of $\rho$ and normalization of $\E_\R$.
\end{proof}

Thus an invariant system state is a (product) relative state with respect to any frame in any state, irrespective of the choice of the covariant POVM $\E_\R$. This follows directly from our framework, and does not need to be stated as an assumption. Another plausible intuition -- that a reference in an invariant state can only give rise to invariant relative states -- is confirmed by the following proposition.

\begin{proposition}
Let $G$ be compact and $U_\S$ and $U_R$ be any unitary representations of 
$G$ in $\his$ and $\hir$. Suppose $\omega$ is invariant. Then $\rho ^{(\omega)} = \mathcal{G}(\rho )$ for any $\rho $.
\end{proposition}

\begin{proof}
    Notice first that if $\omega$ is an invariant state, then $\mu_{\E_\R}^{\omega}$ is an invariant measure, and thus Haar measure, denoted $d\mu$ as before. We then have
    \[
        \rho ^{(\omega)} = \int_G d\mu(g) U_\S(g)^*\rho U_\S(g) = \mathcal{G}(\rho).
    \]
\end{proof}

Thus if the reference is in an invariant state, the only relative states defined with respect to it are also invariant.

As a final example consider $\hir = \mathbb{C}$. Since any unitary representation
on $\mathbb{C}$ will be trivial, the only state is trivially invariant, and a covariant POVM $\E: \mathcal{B}(G) \to B(\mathbb{C}) \simeq \mathbb{C}$ is a (positive localizable) measure $\mu_\E$ on $G$ satisfying $\mu_\E(g.X) = g.\mu_\E(X) = \mu_\E(X)$, and thus it needs to be the normalized Haar measure. So there is a unique such covariant POVM iff $G$ is compact, and in that case, $\Y^\R_*$ is just the twirl as above.

\newpage
\section{Localizing Reference}

In this section, we no longer consider $\omega \in \S(\hir)$ to be fixed, but instead, we are interested in what can be achieved by taking $\omega=\omega_n$ to be elements of the localizing sequence (centered at $e\in G$). The $\omega$-conditioned $\R$-relative operators have been investigated in
e.g. \cite{lov1,lov4}, where the goal was to analyze the extent to which the relative description, as given in terms of $B(\his)^\R$, can be in \emph{agreement} with the non-relative one given on $B(\his)$ upon \emph{conditioning} the relative description. We now generalize Theorem 1. from \cite{lov1}, which lies at the core of such considerations.
\begin{theorem}\label{th:con1}
Let $\R$ be a localizable (principal) frame and $\omega_n$ a localizing sequence centered at $e \in G$. Then for any $A_\S \in B(\his)$ we have
\begin{equation}
    \lim_{n \to \infty}\Y^\R_{\omega_n}(A_\S) = A_\S,
\end{equation}
where the limit is understood in the ultraweak sense. Equivalently, for any $\Omega \in \T(\his)$, in the operational topology we have
\begin{equation}
    \lim_{n \to \infty}(\Y^\R_{\omega_n})_*(\Omega) = \Omega.
\end{equation}
\end{theorem}

\begin{proof}
    It is enough to check the agreement of expectation values of both sides. We then take an arbitrary state $\rho \in \S(\his)$ and calculate
    \[
    \tr[\rho(\Gamma_{\omega_n}\circ \Y)(A_\S)] = \tr[\rho\int_G g.A_\S d\mu^{\E_\R}_{\omega_n}(g)] = 
    \int_G \tr[\rho (g.A_\S)] d\mu^{\E_\R}_{\omega_n}(g)
    \]
    
    The function $g\mapsto \tr[\rho (g.A_\S)]$ is continuous and bounded, and by Proposition \ref{prop:locseqgen} we know that the sequence of measures $\mu^{\E_\R}_{\omega_n}$ converges weakly to $\delta_e $, so by the porte-manteau theorem we have: 
    
    \[\lim_{n \to \infty} \int_G \tr[\rho (g.A_\S)] d\mu^{\E_\R}_{\omega_n}(g)= \int_G \tr[\rho (g.A_\S)] \delta_e (g)= \tr[\rho A_\S] , \]
    
    and so the sequence of operators $ (\Gamma_{\omega_n}\circ \Y^\R)(A_\S) $ converges ultraweakly to $A_\S $.
\end{proof}

Thus the usual non-relational kinematics of quantum mechanics is recovered in the operational sense as a limiting procedure of localizing the reference system - any observable on the system $\S$ can be approximated to arbitrary precision by an observable relative to a localizable frame prepared in a highly localized state. Moreover, any state of $\S$ can be similarly approximated by the relative states of the form $\Y^\R_*[\omega_n \otimes \rho]_G$. We state it as a separate Proposition.

\begin{proposition}\label{cor:locrelst}
Given a localizable (principal) frame $\R$ and a system $\S$, the set of relative states $\S(\his)_\R$ is \emph{operationally dense} in $\S(\his)$.
\end{proposition}

Notice here that if $G$ is countable and the frame is ideal, no such limiting procedure is needed and the agreement is exact, i.e., there exists a state $\omega_{\R} = \dyad{e}$ and we have $(\Gamma_{\omega_R} \circ \Y^{P}) = \mathbb{1}_\S$. Indeed, for $\E_\R(g) = P(g) = \dyad{g} \in B(L^2(G))$, any state is an $\ket{e}$-conditioned $\R$-relative state since we have
\begin{equation}
\Y^\R_*(\dyad{e} \otimes \rho) = \sum_{g \in G}\tr[P(g)\dyad{e}]g.\rho = \sum_{g \in G} \delta(g,e)U^*_\S(g)\rho U_\S(g) = \rho.
\end{equation}

This is the setting considered in \cite{de2020quantum}, and states of the form $\dyad{e} \otimes \rho$ are called `aligned' \cite{krumm2021quantum}. Since $\Y^\R_*$ is constant on $G$-orbits, it is not sensitive to whether the state is aligned or only \emph{alignable} \cite{krumm2021quantum}.

Keeping the localizability assumption, we can go even further with identifying the relative notions with the standard non-relative ones upon high localization of the reference. This shows how the standard quantum-mechanical framework, developed in the context of macroscopic, classical measuring apparatuses, is realized as a limiting case of the relational framework presented here when the classical-like features of the reference are exploited.

\begin{proposition}\label{prop:locyen}
    Given a localizable (principal) frame $\R$ and a system $\S$ the operator space $B(\his)^\R$ inherits the von Neumann algebra structure from $B(\his)$ making the relativization map
    \[
        \Y^\R: B(\his) \xrightarrow{\cong} B(\his)^\R
    \]
    an isometric *-isomorphism. If further $\R$ is sharp then $\Y^\R$ is multiplicative and we~get
    \[
    B(\his)^\R \subseteq B(\hirs)^G \subseteq B(\hirs)
    \]
    as von Neumann algebras.
\end{proposition}

\begin{proof}
We will first show that if $\R$ is localizable then $\Y^\R$ is isometric. Given $A_\S\in B(\his)$ it is shown in \cite{lov1} that $||\Y^{\R}(A_\S)||\leq ||A_\S||$. So it just remains to show that $||A_\S|| \leq ||\Y^{\R}(A_\S)||$. By Theorem \ref{th:con1} there is a localizing sequence of states $\omega_n $ such that for all $\rho\in\S(\his) $:
\begin{align*}
|\tr[\rho A_\S ]|
&=\lim_{n\to \infty} |\tr[\rho(\Gamma_{\omega_n } \circ \Y^{\R}) (A_\S)]|\\
&=\lim_{n\to \infty} |\tr[(\omega_{n}\otimes \rho) \Y^{\R} (A_\S)]|\\
&\leq \sup_{\Omega} |\tr[\Omega \Y^{\R} (A_\S)]|=||\Y^{\R} (A_\S)||,
\end{align*}
where $\Omega \in \S(\hirs)$. Since $\rho$ above is arbitrary and we can conclude
\[
||A_\S||= \sup_{\rho}|\tr[\rho A_\S ]| \leq ||\Y^{\R} (A_\S)||,
\]
Thus if $\R$ is localizable then $\Y^{\R}$ is injective which allows to define an algebra structure on $B(\his)^\R $ by $\Y^{\R}(A_\S)\cdot \Y^{\R}(B):=\Y^{\R}(A_\S B_\S)$. Since $\Y^\R$ is isometric and preserves adjoints, the $C^*$-identity is lifted from $B(\his)$ as we have
\begin{align*}  
||\Y^{\R}(A_\S)^* \cdot  \Y^{\R}(A_\S) || &= ||\Y^{\R}(A_\S^*) \cdot \Y^{\R}(A_\S) || =||\Y^{\R}(A_\S^* A_\S )||\\
=||A_\S^* A_\S||&=||A_\S^* || ||A_\S||=||\Y^{\R}(A_\S)^* || ||\Y^{\R}(A_\S)||,
\end{align*}
and we have the predual $[\T(\his)^\R]^\bigstar \cong B(\his)^\R $, it is a von Neumann algebra. The relativization map $\Y^{\R}$ is then an isometric *-isomorphism by definition. If furthermore, $\R$ is sharp, then $\Y^{\R}$ is multiplicative making $B(\his)^\R $ a von Neumann subalgebra of $B(\hirs)$.
\end{proof}

Notice here, that if the frame $\R$ is localizable but \emph{not} sharp, and the system is given as a composite $\hi_\S\otimes \hi_{\S'}$, this tensor product structure will \emph{not} be preserved under the above isometric isomorphism, which suggests a novel approach to `relativity of subsystems' issue \cite{Ali_Ahmad2022-da}.

\newpage
\section{Relationality of measurement}\label{measuringY}

In this section, we show concretely how our principle of measuring only relativized observables interacts with the standard measurement-theoretic setup. We show that the measurement of covariant observables on the system $\S$ can be modeled as a measurement of relative orientation observables upon localization of the reference system under the assumption that the evolution of the composite system commutes with the group action. This direction will be explored further in future work and is here presented as an illustration of the claims made above regarding the compatibility of the presented formalism with the more standard approaches to quantum measurement under symmetries. This section is independent of the main flow of the narrative.

We begin by recalling a description of a measurement setup as in Quantum Measurement Theory (see e.g. \cite{loveridge2012quantum}). The measurement of a POVM $\E_\S: \mathcal{B}(\Sigma_\S) \to B(\his)$ on a system $\S$ can be described by specifying a \emph{measuring apparatus} (frame) $\R$, together with a \emph{pointer (frame-orientation) observable}  $\E_\R: \mathcal{B}(\Sigma_\R) \to B(\hir)$ on a sample space $\Sigma_\R$, a \emph{Borel map} $f: \Sigma_\R \to \Sigma_\S$, and a \emph{pointer state} $\omega_p \in \S(\hir)$. Such a setup can be depicted~as

\begin{center}
    \begin{equation}\label{qmsetup}
    \begin{tikzcd}
        B(\hir)                                   & B(\his)                                                       \\
        \mathcal{B}(\Sigma_\R) \arrow[u, "\E_\R"] & \mathcal{B}(\Sigma_\S) \arrow[l, "f^{-1}"] \arrow[u, "\E_\S"].
    \end{tikzcd}
    \end{equation}
\end{center}

Further, an \emph{interaction} between $\R$ and $\S$, i.e. a short time unitary evolution $U_{\Delta t} = e^{iH\Delta t} \in B(\hir \otimes \his)$, is introduced that entangles an initially separable state $\omega_p \otimes \rho_\S \in \S(\hir \otimes \his)$ in a way that allows recovering the measurement statistics of $\E_\S$ on $\rho_\S$ by measuring the pointer observable $\E_\R$. More precisely, for any $X \in \mathcal{B}(\Sigma_\S)$ we should have the \emph{probability reproducibility condition} (PRC) \cite{loveridge2012quantum}

\begin{equation}\label{probrepcond}
    \tr[U_{\Delta t}^*(\omega_p \otimes \rho_\S)U_{\Delta t}\E_\R(f^{-1}(X)) \otimes \mathbb{1}_\S ] = \tr[\rho_\S \E_\S(X)].
\end{equation}

Thus the interaction between the system and the measuring apparatus allows an experimenter to access the measurement statistics of the observable $\E_\S$ while the Borel function $f$ allows relating the sample space $\Sigma_\S$, which is an abstract space, e.g. the spectrum of a self-adjoint operator, with the space $\Sigma_\R$ to which the experimenter has direct access to, e.g.~a~detection screen. One can also see the equation above as \emph{defining} the observable $\E_\S$ by the measurement setup~($U_{\Delta t}$,~$f$,~$\E_\R$,~$\omega_p$). If we now assume, which is a common practice, both sample spaces to be identical and set $f=\text{Id}_\Sigma$, the PRC can be then rewritten in the following form
\begin{equation}
\tr[(\omega_p \otimes \rho_\S) \E_{\R}(\Delta t)(X)] = \tr[\rho_\S \E_\S(X)],
\end{equation}
where we have introduced the notation for the evolved frame observable
\begin{equation}
\E_{\R}(\Delta t) := U_{\Delta t}[\E_\R(\cdot) \otimes \mathbb{1}_\S]U_{\Delta t}^*.
\end{equation}

In the presence of symmetries, i.e. when we have an action of $G$ on $\Sigma_\R$, we should have a corresponding action of $G$ on $\hir$ making $\E_\R$ covariant. A short calculation shows that if the interaction $U_{\Delta t}$ commute with frame 'rotations', i.e. $[U_{\Delta t}, U_\R(g) \otimes \mathbb{1}_{\S}]$ for all $g \in G$, a stronger condition holds. Indeed, for any $h \in G$ we then have
\begin{equation}\label{RRP}
\tr[(h.\omega_p \otimes \rho_\S) E_{\R}(\Delta t)(h.X)] = \tr[\rho_\S \E_\S(X)].
\end{equation}

Thus when we 'rotate' the sample space of the apparatus by $h \in G$, which corresponds to acting with $h$ on the $\epsilon$-pointer state, to reproduce the statistics of $\E_\S$ by measuring $E_{\R}(\Delta t)$ on the product state we need to first 'rotate' the input accordingly. We will refer to the above as the \emph{relational reproducibility condition} (RRC). If we further assume that both systems are principal quantum reference frames on the same group $G$, so that $\E_S$ is now assumed covariant as well, and that $\E_\R$ is localizable, for given $\E_S$ and $\E_R$ we can generate the relative orientation observable that will satisfy the RRC in the localized limit. Indeed, take $\omega_n$ to be a localizing sequence for $\E_\R$. Then for any $X \in \mathcal{B}(G)$ we have
\begin{align*}
&\lim_{n \to \infty}\tr[(h.\omega_n \otimes \rho_\S) \E_\S*\E_\R (h.X)]\\
&= \lim_{n \to \infty}\int_G d\mu_{h.\omega_n}^{\E_\R}(g)\tr[\rho_\S g.\E_\S(hX)]\\
&= \tr[\rho_\S h^{-1}.\E_\S(hX)] = \tr[\rho_\S \E_\S(X)].
\end{align*}
If we thus localize the frame at $h\in G$, i.e. take the input state to be $(h^{-1}.\omega_n \otimes \rho_\S)$, to reproduce the statistics of $\E_\S$ by measuring $\E_\S*\E_\R$ on the product state we need to first 'rotate' the input set the \emph{opposite} way. This is an instance of an \emph{active} transformation of the frame -- we localize it at a given point in $G$ -- in contrast to the \emph{passive} transformation that we had before when the whole sample space was being 'rotated'. Given a localizable frame, measuring covariant observables can thus be modeled by a measurement of the corresponding relative orientation observables.

\newpage
\section{Relativized restricted descriptions}\label{sec:relresdes}

Consider now a \emph{global} description in terms of $B(\hi_\T)^G$. Within this setup, two systems $\S$ and $\R$ are distinguished, with $\hi_\T \cong \hirs$, and the \emph{invariant} effects may be conditioned upon a particular state of the reference system by applying the restriction map. If further, the factorization into $\hirs$ respects the $G$-action, i.e. we have $U_\T = U_\R \otimes U_\S$, and $\R$ is indeed a frame, so comes equipped with a covariant POVM $\E_\R: \mathcal{B}(G) \to B(\hir)$, we can relativize the restricted observables, arriving at the subset of the relative observables. Thus the following definition.

\begin{definition}\label{def:liftmap}
The map
\[
\Gamma^\R_\omega := \Y^\R \circ \Gamma_\omega: B(\hirs)^G \to B(\his)^\R.
\]
will be called the \emph{$\R$-relativized $\omega$-restriction~map}. We will write $\mathcal{L}^\R_\omega := (\Gamma^\R_\omega)_*$ for the predual map and refer to it as the \emph{$\R$-relative $\omega$-lifting map}.
\end{definition}

As before, since the $\R$-relativized $\omega$-restriction map is a composition of $\Y^\R$ and $\Gamma_\omega$ which are both unital, normal, and (completely) positive, so is $\Gamma^\R_\omega$. The \emph{$\R$-relative $\omega$-lifting map} can be seen as our operational analog of the `disentangler map' of \cite{de2021perspective}. Indeed, we have
\begin{equation}
\mathcal{L}^\R_\omega: \T(\his)_\R \ni \Omega^\R \mapsto [\omega \otimes \Omega^\R]_G \in \T(\hirs)_G,
\end{equation}
so that the image of $\mathcal{L}^\R_\omega$ consists of $G$-equivalence classes of product states. This map allows to \emph{lift} the $\R$-relative states of an internal frame $\R$ up to the global description of $B(\hi_\T)^G$.

\begin{definition}\label{def:condYrel}
Given a frame $\R$ and a system $\S$, we will refer to operationally $\Eff(\his)^\omega_\R$-equivalent trace-class operators on $\T(\his)$, where
\[
    \Eff(\his)^\omega_\R := \Gamma^\R_\omega(\Eff(\hirs)^G) \subset \Eff(\his)^\R,
\]
as \emph{($\omega,\R$)}-equivalent, while the \emph{($\omega,\R$)}-equivalence classes of \emph{states} in $\S(\hirs)$ will be called \emph{$\omega$-lifted ($\R$-)relative states}.
\end{definition} 

\begin{definition}\label{def:condYrelobs}
Given a frame $\R$ and a system $\S$, elements of the Banach space
\[
B(\his)_\R^\omega := \rm{span}\{\Eff(\his)^\omega_\R\}^{cl} = \Gamma^\R_\omega(B(\hirs)^G)^{cl} \subseteq B(\his)^\R
\]
will be called \emph{$\R$-relativized $\omega$-conditioned operators}.
\end{definition}

The Propositions \ref{prop:iml}, \ref{prop:statespace} and \ref{generalst} (for $F=\Gamma^\R_\omega$) then give the following.

\begin{proposition}
    We have the following Banach space isomorphism
    \[
    \left[\T(\his)/\hspace{-3pt}\sim_{\Eff(\his)^\omega_\R}\right]^\bigstar \cong B(\his)^\omega_{\R},
    \]
    Moreover, the set of $\omega$-lifted \emph{($\R$)}-relative states as given below is a state space
    \[
    \S(\his)_{\R}^\omega := \S(\his)/\hspace{-3pt}\sim_{\Eff(\his)^\omega_\R} \cong (\Gamma^\R_\omega)_*(\S(\hirs)_G) \subset \T(\hirs)^{\rm{sa}}/\hspace{-3pt}\sim_{\Eff(\his)^\omega_\R}.
    \]
\end{proposition}

Notice here, that in the case of a \emph{localizable} frame $\R$, using the lifting map $\mathcal{L}^\R_{\omega_n}$ we can lift an arbitrary $\R$-relative state to a global one which gives back initial relative state upon applying $\Y^\R_*$ up to arbitrary precision, thus providing an \emph{approximate left-inverse} to~$\Y^\R_*$. Indeed, the Theorem \ref{th:con1} gives
    \begin{equation}
    \lim_{n \to \infty} \Y^\R_* \circ \mathcal{L}^\R_{\omega_n} (\Omega^\R) = \lim_{n \to \infty} \Y^\R_*[\omega_n \otimes \Omega^\R]_G= \Omega^\R.
    \end{equation}
Thus in the case of a localizable principal frame $\R$, any $\R$-relative state can be approximated to arbitrary precision by $\omega_n$-product relative states.
\myemptypage

\chapter{Frame-change maps}\label{ch:FrChM}
In a relational framework, like the one presented in this work, we may be concerned with different relative descriptions. Being able to consistently translate between them is a very natural requirement for such a setup, which we address in this chapter. In fact, solving the problems concerning such a translation motivated much if not all the new content of the presented framework -- the proposed notion of relative states, the framing procedure, and the lifting map were all developed in order to provide a rigorous, fairly general, and operational frame-change map. We will use all the concepts provided in the previous chapters in what follows.

In the special theory of relativity, on which this framework in some sense is based as explained in the Introduction, the procedure of changing the frame is performed by applying a Poincar{\'e} transformation to the description given in a particular coordinate system chosen for the Minkowski space. This choice of coordinates is now lifted to the choice of a quantum system with a frame-orientation observable, so a frame, possibly also conditioned upon a frame's state. The usual, non-relational approach to the relativistic description of quantum systems, where we simply assume the action of the Poincar{\'e} group on the system's Hilbert space to exist and be unitary, which is motivated by the interpretation of (pure) states as initial and final with respect to the performed transformation, is recovered as \emph{frame rotations upon localizing a principal frame}. Indeed, as easily confirmed, rotating a localized frame is nothing other than acting on the system itself since for $h \in G$ and $\omega_n$ a localizing sequence centered at $e \in \G$, for any $A_\S$ and in the usual ultraweak topology we have
\begin{equation}
    \lim_{n \to \infty}\Y^\R_{h^{-1}.\omega_n} (A_\S) = U_\S(h)A_\S U_\S(h)^*.
\end{equation}
As noted before, in Special Relativity all frames are `the same', so `changing' the frame and `rotating' the frame is precisely the same thing. In quantum reference frames framework like this ones, however, these notions are very different, and alternative conceptions of what `changing frame' can mean are valid. Here we consider internal, state-based frame-change maps for principal frames, discussing some other possible lines of inquiry in the Discussion chapter \ref{ch:discuss}. We are then interested in exploiting the relations between different internal descriptions in terms of states relative to different frames and providing state space maps that translate one such description to another. Assuming the initial frame $\R_1$ to be \emph{localizable}, allows us to use the $\Y$ construction and combine the intuitions explicated in \cite{de2021perspective} -- that to pass from one relative description to another we should go through a `perspective-neutral' one, which in our setup is given by the \emph{global description} in terms of $B(\hi_\T)^G$ -- and those of \cite{de2020quantum} suggesting that this passage from the relative state to the somehow `corresponding' global one is achieved by `attaching the identity state'. While the states $\ket{0}$ used in \cite{de2020quantum} are not available as normal states in the case of continuous groups, we can make this intuition precise in the context of arbitrary locally compact topological groups in terms of the localizing sequences of states. It is not clear under what conditions the maps used in the context of similar construction in the perspective-neutral approach \cite{de2021perspective} -- the Heisenberg reduction maps, and its weak inverse -- exist, leaving a mathematically cautious reader doubtful of the generality of the results that are being provided there. The construction proposed here, although for now slightly less general than that of \cite{de2021perspective} in the sense of allowing only principal frames, is more general in the types of such frames and groups that can be considered, fully operational, and rigorously defined. In the simple setting of finite groups and ideal frames, where \cite{de2020quantum} and \cite{de2021perspective} agree, our construction, differing from those previously proposed on the level of the relative states, gives the same results up to the operational equivalences that we work with. We believe that the claims concerning the behavior of entanglement, superposition, entropy, and so on under the frame-change maps should be treated with much care regarding operationality.

The first section below begins with describing the relative descriptions containing a pair of frames. To this end, we then consider a global description in terms of $B(\hi_\T)^G$, from which the \emph{internal} quantum reference frames are chosen in terms of the decomposition of $\hi_\T\cong \hio \otimes \hit \otimes \his$ respecting the $U_\T$ action and fixing the frame-orientation observables for the frames. The appropriate operational state spaces will then serve as domains and codomains for the frame change maps.

In the next section, the frame-change maps are defined and proof of their invertibility and composability is provided. The frame-change maps are \emph{not} channels in the usual sense since they are not given at the level of operator algebras, but are perfectly good state space maps appropriate in the convex-theoretic setup that we end up working~in.

In the closing section, we compare our setup with the QRF change maps as presented in \cite{de2020quantum} in the finite group and ideal frames case. We find that when evaluated on the basis states $\ket{g}$ (see \ref{ex:covsharpfin}) our construction yields exactly the same output (pure) states as the one in \cite{de2020quantum}. Thus we recover the classical intuition that underlines and indeed motivates the setup as presented there. In the mentioned work the ``principle of coherent change of reference system'' is stated and invoked in order to extend the map by linearity to the Hilbert space level. This approach is based on a conviction about a special role played in the foundations of quantum mechanics by the pure states and linearity at the level of Hilbert spaces, for which we do not see the operational justification. As we discover, when evaluated on more general (pure) states, the two constructions differ but remain \emph{operationally equivalent}.

As a closing remark, we present a simple procedure of constructing a relative state with respect to a different frame when both frames are treated \emph{externally} to the system in question. This promising direction will be explored elsewhere.

\newpage
\section{Framed relative descriptions}

 To consider a pair of `independent', i.e. described by distinct quantum systems\footnote{A more general version of the two-frame setup could also be considered, see chapter \ref{ch:discuss}.}, internal reference frames, we look for the decomposition of the total Hilbert space as $\hi_\T \cong \hio \otimes \hit \otimes \his$ compatible with the representation $U_\T$ of $G$ on $\hi_\T$, i.e. such that $U_\T = U_1 \otimes U_2 \otimes U_\S$, and provide a pair of frame-orientation observables denoted $\E_i: \mathcal{B}(G) \to \Eff(\hi_i)$, with $i=1,2$. The set of effects corresponding to this setup~is
 \begin{equation}
\Eff(\hi_\T)_{\E_1,\E_2}^G := \Eff(\hi_\T)^G \cap \Eff(\hi_\T)_{\E_1} \cap \Eff(\hi_\T)_{\E_2}.
 \end{equation}
 
We will only be concerned with translating \emph{relative} descriptions.\footnote{This can be done without loss of generality if the conjecture \eqref{conj:rel=rel} holds since we will assume the localizability of frames.} Picking $\R_1$ as a first frame, the relative effects respecting the setup, i.e. the choice of the second frame-orientation observable, are given by the $\R_1$-relativized effects generated by those of the form $\E_2(X) \otimes F_\S$. Hence the following.

\begin{definition}\label{def:framedrelst}
    Let the $(\R_1,\E_2)$-equivalence relation on $\T(\hio \otimes \hit \otimes \his)$, denoted $\Omega \sim_{\R_1,\E_2} \Omega'$, be the operational equivalence with respect to the following set
    \begin{align*}
         B(\hi_2 \otimes \his)^{\R_1}_{\E_2} &:= \rm{span}\left\{\Y^{\R_1}(\E_2(X) \otimes F_\S) \hspace{3pt} | \hspace{3pt} X \in \mathcal{F}(G), F_\S \in \Eff(\his)\right\}^{cl},\\
         \Eff(\hi_2 \otimes \his)^{\R_1}_{\E_2} &:= \{F \in B(\hi_2 \otimes \his)^{\R_1}_{\E_2} \hspace{3pt} | \hspace{3pt} \mathbb{0} \leq F \leq \mathbb{1}\}.
     \end{align*}
    Operators in $B(\hi_2 \otimes \his)^{\R_1}_{\E_2}$ will be referred to as $\R_2$-framed $\R_1$-relative, elements of the corresponding operational state space
    \[
\S(\hi_2 \otimes \his)^{\R_1}_{\E_2} := \S(\hio \otimes \hit \otimes \his)/\hspace{-3pt}\sim_{\Eff(\hi_2 \otimes \his)^{\R_1}_{\E_2}}
\]
will be called \emph{$\R_2$-framed $\R_1$-relative states} and denoted as $[\Omega^{\R_1}]_{\E_2}$ or $\Omega^{\R_1}_{\E_2}$. The state space maps projecting onto the framed trace-class operators will be denoted as $\pi_{\E_i}$, e.g.
\[
\pi_{\E_2}: \T(\hit \otimes \his)^{\R_1} \to \T(\hit \otimes \his)^{\R_1}_{\E_2}.
\]
\end{definition}

Notice here that in the absence of the system $\S$, for $\Omega, \Omega' \in \S(\hio \otimes \hit)$ we have
    \[        \S(\hit)^{\R_1}_{\E_2} = \S(\hio \otimes \hit)/\hspace{-3pt}\sim_{\E_1 * \E_2},
    \]
where $\sim_{\E_2 * \E_1}$ denotes operational equivalence with respect to the relative orientation observable $\E_2 * \E_1 = \Y^{\R_1} \circ \E_2$ (\ref{def:relorobs}). Thus in this situation, the $\R_2$-framed $\R_1$-relative states are only fixed up to the probability distributions they give rise to upon evaluation of the relative orientation observable. This resembles the classical way of thinking about the relative orientation of frames, lifted to the probabilistic regime, with the full-fledged classicality of a single group element describing the relative orientation being recovered upon localization of both (principal) frames \ref{prop:relorobs}.

\newpage
\section{Changing Reference}

We now turn to present the main results of this section. As argued above, the quantum reference frame-change map should translate the internal $\R_2$-framed $\R_1$-relative description to the $\R_1$-framed $\R_2$-relative one and thus it should be a state space map between the following state spaces
\begin{equation}
\S(\hit \otimes \his)^{\R_1}_{\E_2} \to \S(\hio \otimes \his)^{\R_2}_{\E_1}.
\end{equation}
Under the assumption of localizability of $\R_1$, we choose the following strategy. First use the lifting map $\mathcal{L}_{\omega_n}^{\R_1}$ to take the $\R_1$-relative input states to the $G$-equivalent states on the global invariant algebra $B(\hio \otimes \hit \otimes \his)^\G$. Then apply the $\Y^{\R_2}_*$ map to get the corresponding $\R_2$-relative states. This procedure looks as follows.
\begin{center}
\begin{tikzcd}
& \mathcal{S}(\hio \otimes \hit \otimes \his)_G  \arrow[rd, "\Y^{\R_2}_*"] &   \\
\mathcal{S}(\hit \otimes \his)^{R_1} \arrow[ru, "\mathcal{L}^{\R_1}_{\omega_n}"] \arrow[rr,"\Y^{\R_2}_* \circ \mathcal{L}^{\R_1}_{\omega_n}"] &          &  \mathcal{S}(\hio \otimes \his)^{R_2}.
\end{tikzcd}
\end{center}

When the framed operational equivalences are taken care of, in the limit of $n \to \infty$, we arrive at the following notion of a frame-change map.
\begin{definition}\label{def:frchm}
    Assume $\R_1$ to be a localizable (principal) frame. The map 
    \[
        \Phi^{loc}_{1 \to 2}:= lim_{n \to \infty} \pi_{\E_1} \circ \Y^{\R_2}_* \circ \mathcal{L}^{\R_1}_{\omega_n}: \S(\hit \otimes \his)_{\E_2}^{\R_1} \to \S(\hio \otimes \his)^{\R_2}_{\E_2},
    \]
    where $\omega_n$ is any localizing sequence for $\E_1$, will be called a (localized) \emph{frame-change map}. On the $\R_2$-framed $\R_1$-relative states the frame-change map then acts as
    \[
        \Phi^{loc}_{1 \to 2}: [\Omega^{\R_1}]_{\E_2} \mapsto \lim_{n \to \infty} [\Y^{\R_2}_* \circ \mathcal{L}^{\R_1}_{\omega_n}(\Omega^{\R_1}_{\E_2})]_{\E_1}.
    \]
\end{definition}

We will now prove that the frame-change maps are well-defined, invertible exactly when both frames are localizable, composable in the setup of three frames, and translate consistently between the different framed relative descriptions, thus deserving the given name.

\begin{theorem}\label{thm:frmchm}
The frame change map $\Phi_{1\to 2}^{loc}$ is a well-defined state space map making the following diagram \emph{commute}

\begin{center}
\begin{tikzcd}\label{diag:fcmconsist}
            & \mathcal{S}(\hio \otimes \hit \otimes \his)_G \arrow[ld, "\pi_{\E_2} \circ \Y^{\R_1}_*"'] \arrow[rd, "\pi_{\E_1} \circ \Y^{\R_2}_*"] &   \\
\mathcal{S}(\hit \otimes \his)_{\E_2}^{R_1}  \arrow[rr,"\Phi_{1 \to 2}^{loc}"] &                                                                             &  \mathcal{S}(\hio \otimes \his)^{R_2}_{\E_1},
\end{tikzcd}
\end{center}
If further $\R_2$ is also localizable the analogously defined map in the opposite direction provides an inverse, i.e. we have
\[
    \Phi_{2 \to 1}^{loc} \circ \Phi_{1 \to 2}^{loc} = \text{Id}_{\S(\hit \otimes \his)^{R_1}_{\E_2}}.
\]
Moreover, given $\hi_{\T'} \cong \hio \otimes \hit \otimes \hith \otimes \his$ and assuming $\R_1$ and $\R_2$ to be localizable (principal) frames, and $\R_3$ and arbitrary (principal) frame, we have
\[
    \pi_{\E_2} \circ \Phi_{1 \to 3}^{loc} = \Phi_{2 \to 3}^{loc} \circ \Phi_{1 \to 2}^{loc}.
\]
\end{theorem}

The projection $\pi_{\E_2}$ in the last claim is necessary to acknowledge the choice of the second frame-orientation observable that is needed for the right-hand side and could be omitted on the left-hand side if only the first and third frames were taken into account as such.

\begin{proof}
    Take $\Omega_1, \Omega_2 \in \S(\hio \otimes \hit \otimes \his)/\hspace{-3pt}\sim_G$ and write  $\Omega^{\R_i} = \Y^{\R_i}_*(\Omega)$ as usual. We need to show that for localizable $\E_1$ and any localizing sequence $\omega_n$ whenever $[\Omega_1^{\R_1}]_{\E_2} = [\Omega_2^{\R_1}]_{\E_2} $, i.e. whenever we have
    \[
        \tr[\Omega_1^{\R_1}(\E_2(X) \otimes F_\S)] =\tr[\Omega_2^{\R_1}(\E_2(X) \otimes F_\S)] \text{ for all } X \in \mathcal{B}(G), F_\S \in \Eff(\his)
    \]
    we will also have $\Phi_{1\to2}^{loc}(\Omega_1^{\R_1}) = \Phi_{1\to 2}^{loc}(\Omega_2^{\R_1})$, i.e.
    \begin{align*}   
        \lim_{n \to \infty}\tr[(\omega_n \otimes \Omega_1^{\R_1})\Y^{\R_2}(\E_1(X) \otimes F_\S)] = \lim_{n \to \infty}\tr[(\omega_n \otimes \Omega_2^{\R_1})\Y^{\R_2}(\E_1(X) \otimes F_\S)]
    \end{align*}
    for all $X \in \mathcal{B}(G)$ and $F_\S \in \Eff(\his)$, so that $\Phi_{1\to 2}^{loc}$ is well-defined on the equivalence classes. We then calculate\footnote{See \cite{busch1997operational} for integration theory of measurable functions with respect to POVMs.}
    \begin{align*}
    \tr[\Phi_{1\to2}^{loc}(\Omega_1^{R_1})\E_1(X) \otimes F_\S] &=
       \lim_{n \to \infty} \tr[(\omega_n \otimes \Omega_1^{\R_1})\Y^{\R_2}(\E_1(X) \otimes F_\S)]\\ &=
       \lim_{n \to \infty} \tr[(\omega_n \otimes \Omega_1^{\R_1})\int_G d\E_2(g) \otimes \E_1(g.X) \otimes g.F_\S]\\ &=
        \tr[\Omega_1^{\R_1} \int_G d\E_2(g) (\lim_{n \to \infty} \mu_{\omega_n}^{\E_1}(g.X))\otimes g.F_\S]\\ &=
        \tr[\Omega_1^{\R_1} \int_G d\E_2(g) \delta_e(g.X) \otimes g.F_\S]\\ &=
        \tr[\Omega_1^{\R_1} \int_G d\E_2(g) \chi_{g.X}(e) \otimes g.F_\S]\\ &= 
        \tr[\Omega_1^{\R_1} \int_G d\E_2(g) \chi_{X}(g^{-1}) \otimes g.F_\S],
    \end{align*}
    where we have used that $\lim_{n \to \infty} \mu_{\omega_n}^{\E_1} = \delta_e$ and $\delta_e(g.X) = \chi_{g.X}(e) = \chi_X(g^{-1})$. Now we see that by hypothesis we can replace $\Omega_1^{\R_1}$ by $\Omega_2^{\R_1}$ and get the same number for any $X \in \mathcal{B}(G)$ and $F_\S \in \Eff(\his)$. Running this calculation backward gives the first claim as the calculation does not depend on the choice of the localizing sequence.
    
    To prove the second claim (diagram commutativity), we need to show that for arbitrary $\Omega \in \S(\hio \otimes \hit \otimes \his)/\hspace{-3pt}\sim_G$, $X \in \mathcal{B}(G)$ and $F_\S \in \Eff(\his)$ we have
    \[
    \tr[\Phi_{1\to2}^{loc}(\Omega^{\R_1})\E_1(X) \otimes F_\S] = \tr[\Omega^{\R_2}\E_1(X) \otimes F_\S].
    \]
    Using what we already know we calculate
    \begin{align*}
        \tr[\Phi_{1\to2}^{loc}(\Omega^{R_1})\E_1(X) \otimes F_\S] &=
        \tr[\Omega \int_G d\E_1(h) \otimes h.\int_Gd\E_2(g)\chi_X(g^{-1}) \otimes g.F_\S]\\ &=
        \tr[\Omega \int_G d\E_1(h) \otimes \int_Gd\E_2(hg)\chi_X(g^{-1}) \otimes hg.F_\S] 
    \end{align*}
    Now we perform the change of variables $l := hg$ in the second integral and change the order of integration to write
    \begin{align*}
        \tr[\Phi_{1\to2}^{loc}(\Omega^{R_1})\E_1(X) \otimes F_\S] &=
        \tr[\Omega \int_G d\E_1(h) \otimes \int_Gd\E_2(l)\chi_X(l^{-1}h) \otimes l.F_\S] \\&=
        \tr[\Omega \int_Gd\E_2(l) \otimes \int_G d\E_1(h)\chi_X(l^{-1}h) \otimes l.F_\S].
    \end{align*}
Since the $h$ variable appears only in the second tensor factor the second integral can be evaluated giving
    \[
        \int_G d\E_1(h)\chi_X(l^{-1}h) = \int_G d\E_1(h)\chi_{l.X}(h) = \E_1(l.X) = l.\E_1(X),
    \]
   and finally, we get
    \begin{align*}
        \tr[\Phi_{1\to2}^{loc}(\Omega^{R_1})\E_1(X) \otimes F_\S] &=
        \tr[\Omega \int_Gd\E_2(l) \otimes l.(\E_1(X) \otimes F_\S)] \\&=
        \tr[\Omega \Y^{\E_2}(\E_1(X) \otimes F_\S)] =
        \tr[\Omega^{\R_2}(\E_1(X) \otimes F_\S)].
    \end{align*}

To show the next claim (invertibility), writing $\eta_m$ for a localizing sequence of $\R_2$, we calculate
\begin{align*}
    &\tr[\Phi^{loc}_{2 \to 1} \circ \Phi^{loc}_{1 \to 2}(\Omega^{\R_1})\E_2(X) \otimes F_\S]=\\
    &= \lim_{m \to \infty}\tr[\Y^{\R_1}_* \circ \mathcal{L}^{\R_2}_{\eta_m} \circ \Phi^{loc}_{1 \to 2}(\Omega^{\R_1})\E_2(X)\otimes F_\S]\\
    &= \lim_{m \to \infty}\tr[\mathcal{L}^{\R_2}_{\eta_m} \circ \Phi^{loc}_{1 \to 2}(\Omega^{\R_1}) \int_G d\E_1(g) \otimes g.(\E_2(X) \otimes F_\S)]\\
    &= \lim_{m \to \infty}\tr[\Phi^{loc}_{1 \to 2}(\Omega^{\R_1}) \int_G d\E_1(g) \mu^{\E_2}_{\eta_m}(g.X) \otimes g.F_\S]\\
    &= \lim_{m \to \infty}\lim_{n \to \infty}\tr[\Y^{\R_2}_* \circ \mathcal{L}^{\R_1}_{\omega_n}(\Omega^{\R_1}) \int_G d\E_1(g) \mu^{\E_2}_{\eta_m}(g.X) \otimes g.F_\S]\\
    &= \lim_{m \to \infty}\lim_{n \to \infty}\tr[\mathcal{L}^{\R_1}_{\omega_n}(\Omega^{\R_1}) \int_G d\E_2(h) \otimes h.(\int_G d\E_1(g) \mu^{\E_2}_{\eta_m}(g.X) \otimes g.F_\S)]\\
    &= \lim_{m \to \infty}\lim_{n \to \infty}\tr[\mathcal{L}^{\R_1}_{\omega_n}(\Omega^{\R_1}) \int_G d\E_2(h) \otimes \int_G d\E_1(hg) \mu^{\E_2}_{\eta_m}(g.X) \otimes hg.F_\S]\\
    &= \lim_{m \to \infty}\lim_{n \to \infty}\tr[\Omega^{\R_1} \int_G d\E_2(h) \otimes \int_G d\mu^{\E_1}_{\omega_n}(hg) \mu^{\E_2}_{\eta_m}(g.X) hg.F_\S]\\
    &= \lim_{m \to \infty}\tr[\Omega^{\R_1} \int_G d\E_2(h) \mu^{\E_2}_{\eta_m}(h^{-1}.X) \otimes  F_\S]\\
    &= \tr[\Omega^{\R_1}\int_G d\E_2(h)\chi_X(h) \otimes F_\S] = \tr[\Omega^{\R_1}\E_2(X) \otimes F_\S],
\end{align*}
where we have used $\lim_{n \to \infty} \mu^{\E_1}_{\omega_n}(gh) = \delta_e(gh) = \delta_{g^{-1}}(h)$ and $\lim_{m \to \infty} \mu^{\E_2}_{\eta_m}(h^{-1}.X) = \chi_X(h)$. From commutativity it follows that the map $\Phi_{1 \to 2}^{loc}: \S(\hit \otimes \his)^{\R_1}_{\E_2} \to \S(\hio \otimes \his)^{\R_2}_{\E_1}$ is well-defined in the sense that taking the limit $n \to \infty$ does not take the outcome out of the codomain. Since $\Y^{\R_2}_*$ and $\mathcal{L}^{\R_1}_{\omega_n}$ are all linear, it is a state space map.

We now turn to the proof of composability of the frame-change maps. Keeping the notation unchanged, we calculate
\begin{align*}
    &\tr[\Phi^{loc}_{2 \to 3} \circ \Phi^{loc}_{1 \to 2}(\Omega^{\R_1})
    \E_1(X) \otimes \E_2(Y) \otimes F_\S]=\\
    &\hspace{2cm}= \lim_{m \to \infty} \tr[\Y^{\R_3}_* \circ \mathcal{L}_{\eta_m}^{\R_2} \circ \Phi^{loc}_{1 \to 2}(\Omega^{\R_1})\E_1(X) \otimes \E_2(Y) \otimes F_\S]\\
    &\hspace{2cm}= \lim_{m \to \infty} \tr[\mathcal{L}_{\eta_m}^{\R_2} \circ \Phi^{loc}_{1 \to 2}(\Omega^{\R_1})\int_G d\E_3(g) \otimes g.(\E_1(X) \otimes \E_2(Y) \otimes F_\S)]\\
    &\hspace{2cm}= \lim_{m \to \infty} \tr[\Phi^{loc}_{1 \to 2}(\Omega^{\R_1})\int_G d\E_3(g)\mu_{\eta_m}^{\E_2}(g.Y) \otimes g.(\E_1(X) \otimes F_\S)]\\
    &\hspace{2cm}= \lim_{m \to \infty} \lim_{n \to \infty} \tr[\Y^{\R_2}_* \circ \mathcal{L}_{\omega_n}^{\R_1}(\Omega^{\R_1})\int_G d\E_3(g)\mu_{\eta_m}^{\E_2}(g.Y) \otimes g.(\E_1(X) \otimes F_\S)]\\
    &\hspace{2cm}= \lim_{m \to \infty} \lim_{n \to \infty} \tr[\mathcal{L}_{\omega_n}^{\R_1}(\Omega^{\R_1})\int_G d\E_2(h) \otimes h.(\int_G d\E_3(g)\mu_{\eta_m}^{\E_2}(g.Y) \otimes g.(\E_1(X) \otimes F_\S))]\\
    &\hspace{2cm}= \lim_{m \to \infty} \lim_{n \to \infty} \tr[\mathcal{L}_{\omega_n}^{\R_1}(\Omega^{\R_1})\int_G d\E_2(h) \otimes \int_G d\E_3(hg)\mu_{\eta_m}^{\E_2}(g.Y) \otimes hg.(\E_1(X) \otimes F_\S)].
\end{align*}
If we now change the integration variable in the second integral for $g':=hg$ and change the order of integration we can write the operator above as
\begin{align*}
    &\int_G d\E_2(h) \otimes \int_G d\E_3(hg)\mu_{\eta_m}^{\E_2}(g.Y) \otimes hg.(\E_1(X) \otimes F_\S) \\&= \int_G d\E_3(g') \otimes \int_G d\E_2(h)\mu_{\eta_m}^{\E_2}(h^{-1}g'.Y) \otimes g'.(\E_1(X) \otimes F_\S).
\end{align*}
Exchanging the order of limits and taking $m \to \infty$ the integral in the second tensor factor can then be evaluated giving
\[
    \lim_{m \to \infty}\int_G d\E_2(h)\mu_{\eta_m}^{\E_2}(h^{-1}g'.Y) = \int_G d\E_2(h)\chi_{g'.Y}(h) = \E_2(g'.Y) = g'.\E_2(Y).
\]
We then get
\begin{align*}
    &\tr[\Phi^{loc}_{2 \to 3} \circ \Phi^{loc}_{1 \to 2}(\Omega^{\R_1})
    \E_1(X) \otimes \E_2(Y) \otimes F_\S]=\\
    &= \lim_{m \to \infty} \tr[\mathcal{L}_{\omega_n}^{\R_1}(\Omega^{\R_1})\int_G d\E_3(g') \otimes g'.(\E_1(X) \otimes \E_2(Y) \otimes F_\S)]\\
    &= \lim_{m \to \infty} \tr[\Y^{\R_3}_* \circ \mathcal{L}_{\omega_n}^{\R_1}(\Omega^{\R_1})\E_1(X) \otimes \E_2(Y) \otimes F_\S]\\
    &= \tr[\Phi_{1\to 3}^{loc}(\Omega^{\R_1})\E_1(X) \otimes \E_2(Y) \otimes F_\S].
\end{align*}
Since $X,Y \in \mathcal{B}(G)$ and $F_\S \in \Eff(\his)$ were all arbitrary this finishes the proof.
\end{proof}

We have thus shown that the construction $\Phi_{1 \to 2}^{loc} = \lim_{n \to \infty} \pi_{\E_1} \circ \Y^{\R_2}_* \circ \mathcal{L}_{\omega_n}^{\R_1}$ provides an invertible and composable frame-change map for localizable principal frames. Due to the general notion of operational state spaces and state space maps that we use, this is achieved without requiring the maps to be unitary quantum channels.

\newpage
\section{Comparison with other work}
The approach to frame-change maps as presented in \cite{de2020quantum} is briefly summarized in the Appendix \ref{app:ACT}. Since the procedure is rigorously defined only in the finite $G$ case, we now fix such a group. We perform our comparison in the context of ideal frames on $G$ and adopt the convention of the `left-right' $G$-action on $L^{2}(G)$ given by $U(g)\ket{h} = \ket{hg^{-1}}$. Then $G$ acts on $B(L^2 (G))$ as $g.A = U(g)AU(g)^*$, and dually on the states by $g.\Omega = U(g)^*\Omega U(g)$. The unique (up to the unitary equivalence) covariant PVM for such an action is given by
\[
P(h) = \dyad{h^{-1}}
\]
Indeed, we have
\[
g.P(h) = U(g)\dyad{h^{-1}}U(g)^* = \dyad{h^{-1}g^{-1}} = \dyad{(gh)^{-1}} = P(g.h).
\]

Consider now a collection of $N$ such frames and the total Hilbert space $\hi_\T = \bigotimes_{i=1}^N L^2 (G)_i$ with the factors ordered such that $n=1$ system is the first reference and $n=2$ the second. We will write $\S_i$ for the system modeled on $\hi_{n} := \bigotimes_{i \neq n} L^2 (G)_i$. In this case, the lifting map $\mathcal{L}_{\ket{e}}^{\R_1}$ provides an inverse to $\Y^{\R_1}_*$ since we have 
\begin{equation}
    \Y^{\R_1}_*(\dyad{e} \otimes \Omega^{\R_1}) = \sum_{g \in G} \tr[P(g)\dyad{e}] \otimes g.\Omega^{\R_1} = \sum_{g \in G} \delta(e,g) g.\Omega^{\R_1} = \Omega^{\R_1}
\end{equation}
Thus an arbitrary state $\Omega^{\R_1} \in \S(\hi_{1})$ can be understood as an $\R_1$-relative state without any limiting procedure, just like in \cite{de2020quantum}. The map $\Y^{\R_2}_* \circ \mathcal{L}_{\ket{e}}^{\R_1}$ acts on \emph{product} states as follows
\[
    (\omega_2 \otimes \dots \otimes \omega_N)^{\R_1} \mapsto \dyad{e}_1 \otimes \omega_2 \otimes \dots \otimes \omega_N \mapsto \sum_{g \in G} \tr[P(g)\omega_2] g.(\dyad{e}_1 \otimes \omega_3 \otimes \dots \otimes \omega_N)
\]
Now taking $\omega_i=\dyad{h_i}$ with $h_i \in G$ we get (on a representative the operational $\E_2$-equivalence class)
\begin{align*}
    &\Phi_{1 \to 2}^{loc}\dyad{h_2} \otimes \dots \otimes \dyad{h_N} \\
    &= \sum_{g \in G}\ip{h_2}{P(g)h_2} g.\left(\dyad{e}\otimes \dyad{h_3}\otimes \dots \otimes \dyad{h_N}\right) \\
    &= \sum_{g \in G}\delta(g^{-1},h_2)U^*_{\S_2}(g)\left(\dyad{e}\otimes \dyad{h_3}\otimes \dots \otimes \dyad{h_N}\right)U_{\S_2}(g) \\
     &= \dyad{h_2^{-1}}\otimes\dyad{h_3 h_2^{-1}}\otimes\dots\otimes\dyad{h_N h_2^{-1}},
\end{align*}
where we have used that $U^*(h^{-1})\ket{h'} = U(h)\ket{h'} = \ket{h'h^{-1}}$. Written in the ket notation we then have
\begin{equation}
    \Phi_{1 \to 2}^{loc}: \ket{h_2} \otimes \ket{h_3} \dots \otimes \ket{h_N} \mapsto \ket{h^{-1}_2} \otimes \ket{h_3 h_2^{-1}} \otimes\dots \otimes \ket{h_N h_2^{-1}},
\end{equation}

just like in \cite{de2020quantum} (see eq. (24) on pg.7 and also (\ref{ACTomchannel1}) in Appendix \ref{app:ACT}). There the map above is extended by linearity, i.e. assuming \emph{coherence}, to $U_{1 \to 2}:\hi_{2} \to \hi_{1}$ which reads (eq. (25) on pg.7, and (\ref{ACTomchannel2}) in Appendix \ref{app:ACT})
\begin{equation}
U_{1 \to 2} = \sum_{g \in G}\dyad{g^{-1}}{g}_{2 \to 1} \otimes U_{\S}(g),
\end{equation}
where $\hi_{\S}:=\bigotimes_{i \neq 1,2} L^2(G)_i$ and $\hspace{-2pt}\dyad{g^{-1}}{g}_{2 \to 1}: L^2(G)_2 \to L^2(G)_1$. Let us now
check what \cite{de2020quantum} gives when extended to mixed states. Consider first $\omega_2 \in L^2 (G)_2$ and $\rho \in \S(\his)$. We then get
\begin{equation}
U_{1 \to 2}(\omega \otimes \rho)U^*_{1 \to 2} = \sum_{g,h} \bra{g} \omega \ket{h} \left(\dyad{g^{-1}}{h^{-1}} \otimes U_\S (g) \rho U^*_\S (h)\right).
\end{equation}
Our map on the other hand reads
\begin{equation}
    \Phi_{1 \to 2}^{loc} (\omega \otimes \rho) = \sum_{g} \bra{g} \omega \ket{g} \left( \dyad{g^{-1}} \otimes U_\S (g) \rho U^*_\S (g)\right).
\end{equation}
It is now clear that the difference stems from a different approach to implementing symmetry -- for us the group $G$ acts on the operator algebras $B(\hi)$ and state spaces $\S(\hi)$ by duality, while in \cite{de2020quantum}, and also in \cite{de2021perspective} the primary object is the Hilbert space. Perhaps surprisingly it turns out that the two states resulting from these procedures are \emph{$\E_1$-equivalent}. Indeed, a simple calculation gives the following.

\begin{proposition}
    Consider a finite group $G$, a pair of ideal frames (for $G$) $\R_1$ and $\R_2$ with the left-right $G$-action and a system $\S$. We then have
    \[
        \Phi_{1 \to 2}^{loc} = \pi_{\E_1} \circ U_{1 \to 2}(\_)U^*_{1 \to 2}: \S(\hio \otimes \his)_{\E_2}^{\R_1} \to  \S(\hit \otimes \his)^{\R_2}_{\E_1}.
    \]
\end{proposition}

\begin{proof}
    Since we are in the countable $G$ case with ideal frames it suffices to evaluate the states on the effect of the form $\dyad{l^{-1}}\otimes F_\S$, with $l \in G$ and $F_\S \in \Eff(\his)$ arbitrary. We calculate
\begin{align*}
    &\tr[U_{1 \to 2}\Omega^{\R_1}U^*_{1 \to 2}\dyad{l^{-1}}\otimes F_\S] \\
    &=\tr[\sum_{g,h}\dyad{g^{-1}}{g} \otimes U_\S(g)\Omega^{\R_1}\dyad{h}{h^{-1}} \otimes U^*_\S(h)\dyad{l^{-1}} \otimes F_\S] \\
    &=\tr[\sum_g\dyad{g^{-1}}{g} \otimes U_\S(g)\Omega^{\R_1}\dyad{l}{l^{-1}} \otimes U^*_\S(g)F_\S] \\
    &=\tr[\sum_g\Omega^{\R_1}\dyad{l}{l^{-1}} \otimes U^*_\S(g)F_\S\dyad{g^{-1}}{g} \otimes U_\S(g)] \\
    &=\tr[\Omega^{\R_1}\dyad{l} \otimes U^*_\S(l)F_\S U_\S(l)] =\tr[\Omega^{\R_1}\dyad{l} \otimes l^{-1}.F_\S],
\end{align*}
while the operational frame-change maps gives
\begin{align*}
    &\tr[\Phi^{loc}_{1 \to 2}(\Omega^{\R_1})\dyad{l^{-1}}\otimes F_\S] \\
    &=\tr[\dyad{e} \otimes \Omega^{\R_1} \sum_g \dyad{g^{-1}} \otimes g. \left(\dyad{l^{-1}} \otimes F_\S\right)]\\
     &=\tr[\dyad{e} \otimes \Omega^{\R_1} \sum_g \dyad{g^{-1}} \otimes \dyad{l^{-1}g^{-1}} \otimes g.F_\S]\\
     &=\tr[\Omega^{\R_1} \sum_g \dyad{g^{-1}} \delta(l,g^{-1}) \otimes g.F_\S]=\tr[\Omega^{\R_1} \dyad{l} \otimes l^{-1}.F_\S]
\end{align*}
\end{proof}

We have thus shown that in this simple setup when all the procedures of changing frames are defined, they agree up to operational equivalence.

The ``relational Schr{\"o}dinger picture'' frame change map of the perspective-neutral framework can be written (setting $g_i=g_j=e$) in a very similar form (Thm. 4 on pg. 40 of \cite{de2021perspective})
    \begin{equation}
    V_{1\to 2} = \int_G \dyad{\phi(g)}{\psi(g)} \otimes U_\S(g)d\mu(g).
    \end{equation}
where $\{\ket{\phi(g)}\}_{g \in G} \subset \hi_1$ and $\{\ket{\psi(g)}\}_{g \in G} \subset \hi_2$ are systems of coherent states of the first and second frame, respectively, understood in the sense of \eqref{ex:css}, and $\dyad{\phi(g)}{\psi(g)}$ as a map from $\hi_2$ to $\hi_1$. Within the Hilbert space framework of quantum mechanics, the seed state $\ket{\phi(e)}$ admits the interpretation as `localized at identity' only in the ideal case, and thus with the countable group. Then, with the second coherent state system arbitrary, mimicking the calculation above gives a slightly more general agreement up to operational equivalence.

\begin{proposition}
    Consider a finite group $G$, an ideal frame (for $G$) $\R_1$, and a coherent system frame $\R_2$. When we adopt the convention of left $G$-action on the states rather than operators, i.e. $g.\omega = U(g)\omega U^*(g)$, we get
    \[
        \Phi_{1 \to 2}^{loc} = \pi_{\E_1} \circ V_{1 \to 2}(\_)V^*_{1 \to 2}: \S(\hio \otimes \his)_{\E_2}^{\R_1} \to  \S(\hit \otimes \his)^{\R_2}_{\E_1}.
    \]
\end{proposition}

The operational quantum reference frame transformations as developed in this work are thus in operational agreement with other constructions presented in the literature whenever the latter can be stated inside the standard Hilbert space setting and respecting the operational interpretation of the frame-orientation observables.

In this light, the claims about the `relativity of superposition and entanglement' that can be found eg. in \cite{de2020quantum} are shown to have little operational meaning from the perspective of the present work. Indeed, consider an ideal setup with the second frame $\R_2$ prepared in a superposed state, e.g. the input state of the frame change~is
\begin{equation}
    \ket{\psi}=\left((\alpha \ket{h_1}_{\R_2}+\beta \ket{h_2}_{\R_2}) \otimes \ket{g}_\S\right)^{\ket{e}_{\R_1}}.
\end{equation}
At an operational level, there is no difference between the state transformed according to the coherent frame-change map of \cite{de2020quantum}
\begin{equation}
     \left( \alpha \ket{h_1^{-1}}_{\R_1} \otimes \ket{gh_1^{-1}}_\S  +\beta \ket{h_2^{-1}}_{\R_1} \otimes \ket{gh_2^{-1}}_\S \right)^{\ket{e}_{\R_2}},
\end{equation}
and the output of the operational frame-change map \eqref{def:frchm} evaluated on a representative of the $\E_2$-equivalence class of $\dyad{\psi}$, which reads
\[
    \left(|\alpha|^2 \dyad{h_1^{-1}}_{\R_1} \otimes \dyad{gh_1^{-1}}_\S 
 +|\beta|^2 \dyad{h_2^{-1}}_{\R_1} \otimes \dyad{gh_2^{-1}}_\S + \right)^{\dyad{e}_{\R_2}},
\]
as they occupy the same $\E_1$-equivalence class in $\S(\hi_1 \otimes \his)^{\R_2}$. The operationally transformed state is the L\"{u}ders mixture corresponding to the one transformed coherently, and is \emph{not} entangled.

Moreover, the input state being superposed is also an operationally questionable claim, not only because it is basis-dependent, but also because we have
\[
    \left(\alpha \ket{h_1}_{\R_2}+\beta \ket{h_2}_{\R_2}\right)  \left(\alpha\bra{h_1}_{\R_2}+\beta \bra{h_2}_{\R_2}\right) \hspace{5pt} \sim_{\E_2} \hspace{5pt} |\alpha|^2\dyad{h_1}_{\R_2} + |\beta|^2\dyad{h_2}_{\R_2},
\]
and thus both states are equally good at representing the probabilistic uncertainty of the second frame's orientation, one being superposed in the sense of \cite{de2020quantum}, and the other one being mixed.

\newpage
\section{Triangular reconstruction of relative states}\label{sec:extfcm}

We finish with some preliminary ideas on how a system $\S$ is described relative to two different quantum reference frames, $\R_1$ and $\R_2$ that are treated \emph{externaly}. We find a simple procedure that mirrors the sort of frame changes often encountered in classical relativistic physics, and allows for the examination of a different range of questions than those available in the internal QRF setting. The full analysis of such a setup will be carried out elsewhere.

The procedure presented here is concerned with relative states and is in a way closest to the classical intuitions: to describe a system from a different perspective, we need to know how the two frames are related. For example in special relativity, a change of a reference frame is specified by the Lorentz transformation that relates the initial and final frames.

Consider then a pair of frames $\E_i: \B(G) \to \Eff(\hi)$. Given also a relative state $\Omega \in \S(\hi_1 \otimes \hi_2)_G$, the quantum analog of the relation of the frames is given by the probability measure corresponding to the relative orientation observable \eqref{def:relorobs}
\[
\E_2 * \E_1 = \Y^{\R_1} \circ \E_2.
\]
We then seek a procedure to determine a relative state $\rho^{\R_2} \in \S(\his)^{\R_2}$ given $\rho^{\R_1} \in \S(\his)^{\R_1}$ \emph{and} $\Omega_G \in \S(\hi_1 \otimes \hi_2)_G$. Let us stress here, that the states $\Omega$ and $\rho^{\R_1}$ are here treated \emph{independently} -- we do not assume any relation between them. By analogy with the classical world, if the frames would be `the same' and the second frame is 'rotated' by $h \in G$ with respect to the first one, the description it would give to the system should be the one from the first perspective but 'rotated' accordingly. In general, i.e. when the relative orientation observable is not localized, we should provide a probabilistic version of this 'rotation', and thus we may want to try

\begin{equation}
\rho^{\R_2} := \int_G d\mu_{\Omega}^{\E_2 * \E_1}(h) h.\rho^{\R_1}.%\int_G d\mu^{\E_2}_{\omega_2 }(h) h.\rho_\S .
\end{equation}

Another plausible thing to do, given $\rho^{\R_1} \in \S(\his)^{\R_1}$ and $\Omega^{\R_1} = \Y^{\R_1}_*(\Omega)  \in \S(\hi_{\R_2})^{\R_1}$, would be to simply assign to $\R_2$ the corresponding product relative state

\begin{equation}
\tilde{\rho}^{\R_2} := \Y^{\R_2}_*(\rho^{\R_1} \otimes \Omega^{\R_1}) =  \int_G d\mu^{\E_2}_{\Omega^{\R_1}}(g) g.\rho^{\R_1}
\end{equation}

We readily verify that the two lines of thought bring us to the same formula since

\begin{equation}
\tr[\Omega \E_2 * \E_1(X)] =
\tr[\Omega \Y^{\E_1} \circ \E_2(X)] =
\tr[\Omega^{\R_1} \E_2(X)] =
\mu^{\E_2}_{\Omega^{\R_1}} (X).
\end{equation}

We are now also sure that $\rho^{\R_2}_\S$ can indeed be understood as a state relative to $\R_2$, which was not immediately clear from the definition.
\myemptypage

\chapter{Discussion}\label{ch:discuss}
In this thesis, we formulated an operational approach to quantum reference frames, a fully operational approach to relational quantum kinematics. Building on the results presented elsewhere, we provide a first complete account of this framework by spelling out the physical principles that it is based. Presented as such, it is clear how this formalism combines intuitions from Special Relativity, Quantum Measurement Theory, and Gauge Theory, and can be made explicitly relativistic by the right choice of the underlying symmetry group. The concept of the relational description, crucial for our understanding of the whole setup, has not been introduced before. Also the first steps in the direction of justifying the $\Y$ construction as providing access to \emph{all} relational effects are nowhere to be found except here. Supplemented by some plausible conjectures (see below), we see the framework as an almost complete conceptually and mathematically clear theory of operational relational kinematics of localizable quantum reference frames. The proofs of theorems needed to fill the remaining gaps seem to be within reach, as we outline below \eqref{sec:furtherpersp}.

Crucially for the consistency of the framework with the usual, non-relational approach to quantum physics, the latter is recovered to arbitrary precision under the localization of the reference system. In the relativistic case, the framework then boils down to the usual description of quantum systems as the representations of the Poincar{\'e} group with the action corresponding to Poincar{\'e} transformations of the external (classical) frame of reference. In this light, we understand the presented approach as pointing towards an alternative to Quantum Field Theory operational account of relativistic quantum physics, with \emph{no} background structure. Indeed, the notion of space-time, even in its weakest possible form as a set of events, is absent altogether from formalism. Nevertheless, its \emph{operational} features - namely the observables of relative position, time differences, relative velocities, angles, and so on -- can all be modeled; formalism is designed to make sense of these quantities in the quantum realm. One can imagine the space-time entering only as a \emph{useful abstraction} in the case when the relational spatiotemporal observables of relative orientations give rise to localized probability distributions, and thus \emph{definite} spatiotemporal relations between (quantum) systems. We outline how such an approach to relational space-time emergence could be realized within the framework below; we hope to deepen our understanding of these issues in the future by providing specific realizations of the ideas we are alluding to. Needless to say, a background-less approach to relativistic quantum physics could shed some light on the problems of incorporating the description of gravity into the picture.

\newpage
\section{Summary}

In the course of the presentation, we went through the following stages. After some general technical remarks about the operator-algebraic setup for quantum mechanics under symmetries, topologies that we use, and the notion of a state space relevant to the purpose, we introduced the main tool that we use for imposing operationality -- the \emph{operational equivalence} -- that allows quotienting state spaces in such a way that the equivalence classes represent operationally distinguishable states. We also briefly discussed the notion of \emph{covariance} and provided some motivation and the definition of \emph{localizability} of POVMs.

After such preliminaries, we moved to present the operational approach to relational quantum kinematics. Beginning by introducing the \emph{invariant descriptions} that take care of the gauge-invariance in our setup, we discussed how this approach relates to others. We then provided the definition of a quantum reference frame as a group representation equipped with a covariant POVM, thus encompassing all the definitions previously proposed. We also distinguished important classes of frames, generalizing the existing classification. Further, the operational equivalence was invoked again in order to acknowledge the choice of the frame-orientation observable by the notion of \emph{framing}. Combined with the gauge-invariance it gave rise to the \emph{relational descriptions} that make operational justice to all our symmetry requirements. Finally, the \emph{$\Y$ construction} was invoked as a way to access the observables satisfying such requirements, generating the \emph{relative descriptions}. An argument was provided in the context of a localizable frame on a finite group that the $\Y$ construction in fact provides \emph{all} such observables. At the end of this section, the double meaning of the invariant algebras was discussed -- they were interpreted as \emph{perspective-independent} descriptions of the system in question or as the arena for a \emph{global description} where the frames can be understood as \emph{internal}.

We then moved on to discuss two different ways in which the gauge-invariant descriptions can be conditioned upon the choice of the state of the reference system. To this end, we introduced the restriction maps $\Gamma$ and first applied them to the relative operators, defining \emph{conditioned relative descriptions}. We then presented a series of results confirming that upon localization of the reference system, the usual non-relational quantum kinematics is recovered, also making direct contact with the Quantum Measurement Theory. Further, we applied the restriction maps to the global description, which upon relativization gave rise to the \emph{relativized restricted descriptions}. In this context, the \emph{lifting maps} were also introduced as a way to lift relative states to global ones by `attaching' a chosen state of the reference.

With all the kinematics in hand, we turned to present an approach to frame-change maps aligned with the presented operational setup. We began by defining yet another class of \emph{framed relative descriptions}, corresponding to choosing a pair of internal frames, one being treated as a reference, and the other acknowledged by introducing the relevant framing. Finally, a definition of frame-change maps as a mapping between such framed relative descriptions was given, and proof of consistency, invertibility, and composability was provided, establishing the proposed construction as a viable notion of a frame-change map. Closing the analysis, we compared the proposed construction with others present in the literature in the context of a finite group and ideal frames discovering agreement up to operational equivalence. As a final comment, we proposed an alternative -- \emph{external} -- way of performing frame changes.

\newpage
\section{Further perspectives}\label{sec:furtherpersp}

We end this thesis by briefly presenting some research directions aiming at a better understanding of the operational QRF framework, suggesting its novel applications and further generalizations.\footnote{The reader interested in contributing to this research is most welcome to reach out to the author at glowacki@cft.edu.pl or janmarcinglowacki@gmail.com.}

\subsection{Missing pieces}

Perhaps the biggest deficiency of the framework as developed so far is the lack of operational justification of the $\Y$ construction in the case of continuous groups and its restriction in the continuous $G$ case to the \emph{principal} frames. To this end consider a non-principal localizable frame $\E_\R: \mathcal{B}(G/H) \to B(\hir)$ with $H \subseteq G$ a non-trivial (closed) subgroup. We conjecture that in such a case the relational, i.e. framed and invariant, effects come from relativizing the $H$-invariant ones, i.e. that the $\Y$ construction extends to
\begin{equation}
    \Y^\R: B(\his)^H \ni A_\S \mapsto \int_{G/H} d\E_\R(gH) \otimes gH.A_\S \in B(\hirs)^G,
\end{equation}
and that the image of this map on the $H$-invariant effects exhausts the relational ones, i.e. we have
\begin{equation}
    \Eff(\his)^\R := \Y^\R(\Eff(\his)^H) \cong \Eff(\hirs)^G_\R.
\end{equation}
This strengthens the conjecture \eqref{conj:rel=rel} stating the relational effects are precisely relative ones in the case of localizable \emph{principal} frames. We believe that convergence and the relevant properties of the homogenous $\Y$ formula above should be possible to prove by adapting the original proof as presented in \cite{lov1}. Regarding the \eqref{conj:rel=rel} conjecture and its generalization above, we propose to investigate the following approach to the integration theory of operator-valued functions with respect to (positive) operator-valued measures. Consider a POVM $ \E_\R: \mathcal{F}(\Sigma) \to \Eff(\hir)$ and a function $f: \Sigma \to B(\his)$ such that $\Sigma \ni x \mapsto \tr[\rho f_\S(x)] \in \mathbb{C}$ is \emph{measurable} for any state $\rho \in \S(\his)$. The operator $\int_\Sigma d  \E_\R(x) \otimes f(x) \in B(\hirs)$ is then defined as a continuous functional on $\T(\hirs)$ given by the linear and continuous extension of the map
\begin{equation}
    \omega \otimes \rho \mapsto \int_\Sigma d\mu^{\E_\R}_\omega(x)\tr[\rho f_\S(x)].
\end{equation}
For $\hir = \mathbb{C}$, the map $ \E_\R$ is a probability measure $p$ on $(\Sigma,\mathcal{F})$ and we have
\[
\tr[\rho \int_\Sigma p(x) f_\S(x)] = \int_\Sigma p(x)\tr[\rho f_\S(x)],
\]
while for $\hir=\his=\mathbb{C}$ integration theory of complex-valued~functions is~recovered. Moreover, the tensor product decomposition seems too strong, and could perhaps be replaced with commutativity of $[f(x),\E(X)]$, with $f: \Sigma \to B(\hi), \E: \mathcal{F}(\Sigma) \to \Eff(\hi)$, or even further by compatibility of these effects. Further, we see no reason for this definition, if correct, not to extend to the convex theoretic setup, i.e. replacing $B(\his)$ by the space of bounded affine functionals $A^b(\S)$ \eqref{def:gpteff} on a total convex subset $\S \subset V$ of a real vector space \eqref{def:stsp}.

\subsection{Generalizations and applications}

Given that we always have an action of $G$ on the considered operator spaces $\rm{span}(\mathcal{O})^{cl}$, it is interesting to ask about the possibility of characterizing those order-unit Banach spaces that can support the actions of the continuous groups. Such considerations can possibly lead to some insights in the direction of \emph{reconstruction} of an operational relational quantum formalism like the one presented in this work, or its \emph{generalization} into the realm of General Probabilistic Theories. It is a known fact that the theory of von Neumann algebras is intimately interlinked with the theory of representations of continuous groups, which was one of the reasons for its development. Moreover, the results of Alfsen and Shulz \cite{Alfsen2002-kh} characterizing state spaces of von Neumann algebras among those of Jordan algebras in terms of the dynamical correspondence point to deep links between the von Neumann algebras setup and the ability to support dynamics in general, that we wish to investigate in the convex-theoretic context.

The operational notion of a frame-change map calls for both applications and generalizations. We believe it should be extended besides localizable frames. This may be possible to achieve either by replacing the localizing sequence for the lifting maps by an arbitrary state of the reference or by exchanging the lifting by some other map taking relative states to global ones, e.g. the preimage map $(\Y^{\R_1}_*)^{-1}\{\_\}$. Extending the $\Y$ construction to the non-principal frames should also open the possibility of considering more interesting localizable frame changes than the ones available now.

Another potential generalization of the setup of the frame-change maps would be to drop the tensor product decomposition in favor of a weaker commutativity requirement. Indeed, when we want to consider different internal frames on $\hi_\T$, instead of looking for a decomposition $\hi_\T \cong \hio \otimes \hit \otimes \his$, we could instead consider a pair of frame-orientation observables on $\hi_\T$ and assume commutativity of their effects, i.e. $[\E_1(X_1),\E_2(X_2)]=0$ for all $X_i \in \mathcal{B}(\Sigma_{\R_i})$, $i=1,2$, or, more generally \emph{joint measurability} of $\E_1$ and $\E_2$. We can also imagine a pair of internal frames \emph{not} satisfying any such `independence' requirement, somehow `overlapping'. Providing a frame-change map for such setups would be an interesting generalization of the construction developed so far.

When it comes to applications, we believe that the (sufficiently developed) operational frame-change maps will provide a final resolution of the conceptual difficulties connected to the Wigner's Friend type scenarios as they allow for relating perspectives taking into account the available observables and the underlying symmetry structure. Moreover, a whole series of questions opens up here regarding how different properties of quantum states which can be considered resources behave under the operational frame-change maps in general, providing an operational perspective on the issue of `quantum resource covariance'~\cite{qrcov}.

We believe that most, if not all, of the experimental setups in relativistic quantum physics could, and maybe should, be modeled along the following lines. A starting point is a global description in terms of a strongly continuous unitary representation of the Poincar{\'e} group on $\hi_\T$. Then the observer specifies the reference system for performing the observations, so a quantum reference frame $\R$ with the frame-orientation observable, the quantum system $\S$ that will be observed relative to $\R$, and the rest of the world, the environment $\Eff$, providing a decomposition $\hi_\T \cong \hir \otimes \his \otimes \hi_\Eff$. It can be done operationally all at once by specifying the set of available effects~to~be
\begin{equation}
    \mathcal{O} = \{\E_\R(X) \otimes F_\S \otimes \mathbb{1}_\Eff | X \in \mathcal{B}(G/H), F_\S \in \Eff(\his)\},
\end{equation}
where $H \subset G$ is the isotropy subgroup of the frame $\R$. Then an observable of interest on $\E_\S: \mathcal{F}(\Sigma) \to B(\his)^H$ should be specified and relativized to
\begin{equation}
\Y^\R \circ \E_\S \otimes \mathbb{1}_\Eff: \mathcal{F}(\Sigma) \to B(\hi_\T)^G.
\end{equation}
 Notice here, that in such an approach the problem of the arbitrariness of Heisenberg's cut is no longer present -- all the systems are treated as quantum from the outset. The experimental predictions should then arise by adequate modeling of the preparation procedure, upon localization of the reference frame and understanding the time translations as generating relative time evolution of $\S$ with respect to $\R$ \emph{and} $\Eff$ \cite{loveridge2019relative}. To support this claim, we recall \eqref{measuringY} that a general quantum measurement setup in the presence of symmetries, assuming that the evolution of a composite system $\R \otimes \S$ entangling the initially separable state commutes with the group action, measurement of any covariant POVM on $\S$ can be modeled as a measurement of the relative-orientation observable. We consider providing a relational model for the scattering experiments that are traditionally described by relativistic quantum field theory and delivering consistent predictions as one of the ultimate goals of the presented framework.

\subsection{What about space-time?}

Another pressing research direction is to look for ways to make contact with the algebraic formulation of quantum field theory. In particular, recently Edward Witten et al. \cite{Witten2023-zw},\cite{wittenetal} considered algebras generated by (smeared) paths of `observers' traveling along time-like geodesics, which we could perhaps try to interpret as the frame algebras. Another direction could be to try putting a group of diffeomorphisms as the underlying symmetry structure. This seems challenging but maybe not hopeless, as there are ways to put smooth structure on such groups (see e.g. \cite{Michor}). This direction, however, seems to compromise operationality in the sense of introducing the space-time manifold through the back door.

We propose the following more intrinsic, and operational, perspective on relational spatiotemporality that also suggests some links to theories of gravity. Consider a collection of principal internal quantum reference frames, so a setup given by $\hi_\T \cong \bigotimes_{i=0}^N \hi_i$ and $\E_i: \B(G) \to \Eff(\hi_i)$.\footnote{One could also imagine specifying just one, `global' frame orientation observable, from which the family of covariant POVMs may arise as a result of decomposing the `global' frame's Hilbert space and treating each factor independently.} Fixing the $0$th frame as our reference, we can generate an `$N$-point relative orientation observable'~as
\[
    \B(G^{\times N}) \ni (X_1,X_2,\dots,X_N) \mapsto \int_G d\E_0(g) \otimes g.\left(\E_1(X_1)\otimes E_2(X_2) \otimes \dots \otimes \E_N(X_N)\right).
\]
Provided a global state $\Omega \in \S(\hi_\T)_G$ this gives rise to $N$ probability distributions over $G$, corresponding to the relative orientation of the $0$th frame with respect to all others. If the global state happens to be in a highly localized form according to all frames\footnote{Or perhaps the frames have been chosen such that this is the case.}, i.e. it is a product state of the form
\begin{equation}
    \Omega \sim_G \omega_{n_0}(e) \otimes \omega_{n_1}(h_1)\otimes \dots \otimes \omega_{n_N}(h_N),
\end{equation}
the $N$-point relative orientation observable will with good approximation return the collection of group elements $(h_1,h_2,\dots,h_N)$, describing the configuration of the frames in an invariant manner. The fact that the group elements $h_i$ can be chosen arbitrarily \emph{suggests} picturing the whole collection of frames as placed in an ambient background space. In the case of the restricted Poincar{\'e} group, we make direct contact with the space-time picture along the following lines. Since $\mathcal{P}^+ \cong T^{1,3} \rtimes O^+(1,3)$, where $T^{1,3}$ is the (abelian) group of translations of the Minkowski space and $O^+(1,3)$ the restricted Lorentz group (component containing identity). In the case of such a localizable global state, the relative orientations are approximated by $h_i=(v_i,\Lambda_i)$. Now we can \emph{imagine} the $0$th frame placed at the origin of the Minkowski space, and all other frames placed at other points accordingly to the $v_i$ vectors. The elements $\Lambda_i$ then describe the relative orientation of other frames with respect to the $0$th one that we have chosen to be localized at $e \in \mathcal{P}^+$. Such a choice of a Lorentz frame at each point of the Minkowski space occupied with a frame can be seen as providing a system of {local inertial coordinates} around that point. Under this identification, and the picture we have in such a localized limit resembles one of the ways of describing the \emph{gravitational field} -- indeed, it can be given as a section of a principal fiber bundle of tetrads over the Minkowski space (see e.g. \cite{kijowski}). Investigating these connections further is an exciting research direction that we wish to pursue in the future.

\subsection{What about the observers?}

We finish with some general conceptual considerations concerning the notion of an observable, the role of an observer, and possible mathematical implementation of the presented views. The basic entities of the formalism are quantum reference frames understood as covariant POVMs, so quantum observables on homogeneous spaces respecting the $G$-action. We believe this notion should be revisited on both conceptual and mathematical grounds. To this end, we share the following intuitions. An observable in an operational framework should primarily be seen as a map \emph{from the state space of a system}, that can be given by the positive trace-one operators on the Hilbert space, normalized elements of the predual of a $W^*$-algebra, a base norm space, a convex subset of a real vectors space, and so on, depending on the level of generality that is needed, \emph{to the space of probability distributions}. We find plausible the interpretation of these probability distributions as \emph{representing the observer's statistical state of knowledge about the statements concerning the chosen quantity}. In this light, following Jayne's view on Probability Theory as an extension of logic \cite{Jaynes2003-dr}, probability distributions are best understood as \emph{normalized positive continuous valuations on complete Boolean algebras}. Valuations are simply functionals preserving the underlying lattice structure. More precisely, we have the following definition~\cite{Pavlov22}.
\begin{definition}
    A positive valuation on a distributive lattice $L$ (e.g., a locale or a Boolean algebra) is a
map $\nu: L \to [0,\infty)$ such that $\nu(0)=0, \nu(p \vee q) + \nu(p \wedge q) = \nu(p) + \nu(q)$ and $p \leq q$ implies $\nu(p) \leq \nu(q)$. A positive valuation is continuous if it preserves the existing suprema of directed subsets.
\end{definition}
Complete Boolean algebras that admit `enough' such valuations, meaning that for any proposition $p \neq 0$ there is a continuous valuation such that $\nu(p) \neq 0$, are called \emph{localizable}. We understand this requirement as a natural one -- localizable Boolean algebras are those for which we can reason probabilistically about any non-trivial proposition. Perhaps surprisingly, the notion of a localizable Boolean algebra is \emph{equivalent} (in fact, in the categorical sense, see \cite{Pavlov22}) to that of a \emph{commutative von Neumann algebra}. The latter is always isometrically isomorphic to the algebra of measurable functions on a measurable space $L^\infty(\Sigma,\mu)$, and under this identification, the probability distributions, i.e. normalized positive continuous valuations, correspond to the normalized functions in $L^1(\Sigma,\mu)$, so precisely the normalized elements of the predual of the von Neumann algebra $L^\infty(\Sigma,\mu)$, and the propositions in the Boolean algebra to the measurable subsets of $\Sigma$. We then suggest grounding the notion of a quantum observable in logic via this Gelfand-type duality between \emph{localizable} Boolean algebras and commutative von Neumann algebras by endorsing the following definition.
\begin{definition}
    A \emph{von Neumann observable} is a normal positive unital map
    \[
    \hat{\E}: L^\infty(\Sigma,\mu) \to B(\hi).
    \]
    It is called \emph{sharp} if it is multiplicative, and \emph{localizable} if it is norm-preserving. When $\Sigma$ is a homogeneous space for a locally compact second countable Haussdorf group $G$ and we have a strongly continuous projective unitary representation of $G$ on $B(\hi)$, $\hat{\E}$ is called \emph{covariant} if it is an \emph{equivariant} map.
\end{definition}
This can be generalized even further by replacing $B(\hi)$ with a general $W^*$-algebra, or an order-unit Banach space, and treating $L^\infty(\Sigma,\mu)$ accordingly \cite{Beltrametti}, \cite{kuramochiCompactConvexStructure2020}. The point is that the predual map $\hat{\E}_*$ will map states in $\S(\hi) \subset \T(\hi)$ to probability distributions, which is all and precisely what we want it to be, respecting the mentioned interpretation of the probability distributions. This research direction requires further investigation, which is the subject of ongoing work. For now, we conjecture that the definition above is aligned with the usual usage of the terms \emph{sharp}, \emph{localizable} and \emph{covariant}. As easily confirmed, a POVM is indeed associated with $\hat{\E}$ via
\begin{equation}
    \E(Y) := \hat{\E}(\chi_Y),
\end{equation}
where $\chi_Y$ is the characteristic function of the measurable subset $Y \subseteq \Sigma$. The class of $\mu$-continuous POVM on a compact space with a bounded variation has been characterized precisely as such maps \cite{Roumen2014-za}. We propose to turn the logic around and \emph{require} observables to be such maps, realizing the role we want them to play as assigning normalized positive continuous valuations on localizable Boolean algebras of propositions considered by the observer to the states.

The algebraic perspective on POVMs as presented above provides in our view not only conceptually advantageous, which is perhaps a controversial claim but also brings more structural coherence to the framework by allowing us to speak about all the basic notions of the framework in a unified language, namely inside the category of normal positive unital maps between von Neumann algebras (or more generally the state space morphisms). Indeed, the $\R$-framed effects can now be seen as elements in the image of the $\hat{\E} \otimes \mathbb{1}_\S$ map, the $\Y^\R$ map is also such a normal positive and unital, and thus the relativizing procedure will always take von Neumann observables to the alike. 

The interpretational shift in the direction of logic and the observer's state of knowledge puts the \emph{continuity} of possible outcomes of (localized) spatiotemporal relative-orientation observables on the side of the observer. The space-time, in which we imagine physical objects to live, is then continuous precisely because we think it is.
\myemptypage

\addcontentsline{toc}{chapter}{\protect\numberline{}References}
\printbibliography[title={References}] %you may change the title in the toc here if you want
\myemptypage

\chapter*{\LARGE \textbf{Appendices}}
\fancyhf{} %clear the header, it should be empty for the appendices
\renewcommand{\headrulewidth}{0pt} %no rule
\fancyfoot[C]{\thepage} %set the page numbers in the center of the footer instead 

%it is possible to set a different page numbering style for the appendix, but I personally just continued with the same page numbering as the main content as I find that more tidy
%\pagenumbering{roman}
%\setcounter{page}{1}

\addcontentsline{toc}{chapter}{\protect\numberline{}Appendices:}
\appendix

\section{Notation}\label{app:notation}
For the reader's convenience, we collect below the definitions of spaces and maps which are used in the framework, together with the relevant notation, mentioning also the crucial results.
\begin{itemize}
    \item \emph{Invariant operators/effects}: von Neumann subalgebra of $B(\hi)$/subset of $\Eff(\hi)$ consisting of operators/effects invariant under the given unitary representation of a group $G$
    \begin{align*}
    B(\hi)^G &:= \{A \in B(\hi)\hspace{3pt}|\hspace{3pt}g.A \equiv U(g)AU^*(g)=A\} \subseteq B(\hi), \\
    \Eff(\hi)^G &:= \{F \in \Eff(\hi)\hspace{3pt}|\hspace{3pt}g.F \equiv U(g)FU^*(g)=F\} \subseteq \Eff(\hi) \subset B(\hi).
    \end{align*}
    \item \emph{Invariant states/trace-class operators}: total convex subset/subspace of states/trace-class operators which are invariant under the given unitary representation of $G$
     \begin{align*}
         \T(\hi)^G &:= \{T \in \T(\hi)\hspace{3pt}|\hspace{3pt}g.T \equiv U^*(g)T U(g)=T\} \subseteq \T(\hi), \\
        \S(\hi)^G &:= \{\Omega \in \S(\hi)\hspace{3pt}|\hspace{3pt}g.\Omega \equiv U^*(g)\Omega U(g)=\Omega\} \subset \T(\hi)^G \subseteq \T(\hi).
     \end{align*}
     \item \emph{Global states/trace-class operators}: the quotient space of classes of states/ trace-class operators that can not be distinguished by the invariant effects (or, equivalently, by the invariant operators)
     \begin{align*}
         \T(\hi)_G &:= \T(\hi)/\hspace{-3pt}\sim_G, \\
        \S(\hi)_G &:= \S(\hi)/\hspace{-3pt}\sim_G,
     \end{align*}
     where $\Omega \sim_G \Omega' \hspace{5pt} \Leftrightarrow \hspace{5pt} \tr[\Omega F]=\tr[\Omega' F] \text{ for all } F \in \Eff(\hi)^G$. In the case of compact $G$ we have
     \begin{align*}
         \T(\hi)_G &\cong \T(\hi)^G \text{ (as Banach spaces),} \\
         \S(\hi)_G &\cong \S(\hi)^G \text{ (as state spaces).}
     \end{align*}
     \item \emph{Framed operators/effects}: subspace of $B(\hirs)$/subset of $\Eff(\hirs)$ consisting of operators/effects respecting the choice of the frame-orientation observable $\E_\R: \mathcal{F}(\Sigma) \to \Eff(\hir)$
     \begin{align*}
         \Eff(\hirs)_{\E_\R} &:= \rm{conv}\left\{\E_\R(X) \otimes \F_\S \hspace{3pt} | \hspace{3pt} X \in \mathcal{F}(\Sigma_\R), \F_\S \in \Eff(\his)\right\}^{cl},\\
         B(\hirs)_{\E_\R} &:= \rm{span}(\Eff(\hirs)_{\E_\R})^{cl}.
     \end{align*}
     \item \emph{Relational operators/effects}: subspace of $B(\hirs)$/subset of $\Eff(\hirs)$ consisting of invariant and framed operators/effects
     \begin{align*}
         \Eff(\his)^G_{\E_\R} &:=  \Eff(\his)_{\E_\R} \cap \Eff(\hirs)^G,\\
         B(\his)^G_{\E_\R} &:= {\rm span}\{\Eff(\his)^G_{\E_\R}\}^{cl} \subset B(\hirs).
     \end{align*}
     \item \emph{Relative operators/effects}: subspace of $B(\hirs)$/subset of $\Eff(\hirs)$ consisting of operators/effects relativized with the relativization map
     \[
        \Y^\R (A_\S) := \int_G d\E_\R(g) \otimes g.A_\S,
     \]
     so that we have
     \begin{align*}
         \Eff(\his)^\R &:=  \Y^\R(\Eff(\his)) \subseteq \Eff(\his)^G_{\E_\R},\\
         B(\his)^\R &:= {\rm span}\{\Eff(\his)^\R\}^{cl} = \Y^\R(B(\his)) \subseteq B(\his)^G_{\E_\R}.
     \end{align*}
     When $\R$ is localizable (i.e. $||\E_\R(X)||=1$ for all $X \in \B(G)$), we have
     \[
     B(\his)^\R \cong B(\his)
     \]
     as von Neumann algebras, and when $\E_\R$ is sharp $B(\hir)^\R$ is a subalgebra of $B(\hirs)$. If $G$ is finite we have
     \[
     \Eff(\his)^\R=\Eff(\his)^G_{\E_\R}.
     \]
     \item \emph{Relative states/trace-class operators}: the quotient space of classes of states/trace-class operators that can not be distinguished by the relative effects (or, equivalently, by the invariant operators)
     \begin{align*}
         \T(\his)^\R &:= \T(\hirs)/\hspace{-3pt}\sim_\R, \\
        \S(\hi)^\R &:= \S(\hirs)/\hspace{-3pt}\sim_\R,
     \end{align*}
     where $\Omega \sim_\R \Omega' \hspace{5pt} \Leftrightarrow \hspace{5pt} \tr[\Omega F]=\tr[\Omega' F] \text{ for all } F \in \Eff(\his)^\R$. We have
      \begin{align*}
         \T(\his)^\R \cong \Y^\R_*(\T(\hirs)) &= \Y^\R_*(\T(\hirs)/\hspace{-3pt}\sim_G) \subseteq \T(\his), \\
        \S(\hi)^\R \cong \Y^\R_*(\S(\hirs)) &= \Y^\R_*(\S(\hirs)/\hspace{-3pt}\sim_G) \subseteq \S(\his).
     \end{align*}
      When $\R$ is localizable the inclusions above are dense (in operational topology). We use the following notation for relative states
      \[
        \Omega^\R \equiv \Y^\R_*(\Omega) \cong [\Omega]_\R,
      \]
      where $\Omega$ is a state on the composite system $\Omega \in \S(\hirs)$.
    \item \emph{The $\omega$-conditioned relative observables/effects}: subspace of $B(\his)$/subset of $\Eff(\his)$ consisting of $\omega$-conditioned relativized observables/effects
      \begin{align*}
        \Eff(\his)^\R_\omega &:= \Y^\R_\omega (\Eff(\his)) \subseteq \Eff(\his),\\
         B(\his)^\R_\omega &:= \rm{span}\{\Eff(\his)^\R_\omega\}^{cl} = \Y^\R_\omega (B(\his)) \subseteq B(\his),
      \end{align*}
      where $\Y^\R_\omega := \Gamma_\omega \circ \Y^\R$ and 
      \[
        \Gamma_\omega: B(\hirs) \ni A_\R \otimes A_\S \mapsto \tr[\omega A_\R] A_\S \in B(\his)
      \]
      is extended by linearity and continuity. The $\omega$-conditioned relative effects take the form
     \[
        \Y^\R_\omega(F_\S) = \int_G d\mu^{\E_\R}_\omega(g) U(g)F_\S U^*(g).
     \]
      \item \emph{The $\omega$-product relative states/trace-class operators}: the quotient space of classes of states/trace-class operators that can not be distinguished by the $\omega$-conditioned relative effects (or, equivalently, by the invariant operators)
     \begin{align*}
         \T(\his)_\omega^\R &:= \T(\his)/\hspace{-3pt}\sim_{(\R,\omega)} \subseteq \T(\his), \\
        \S(\hi)_\omega^\R &:= \S(\his)/\hspace{-3pt}\sim_{(\R,\omega)} \subseteq \S(\his),
     \end{align*}
     where $\rho \sim_{(\R,\omega)} \rho' \hspace{5pt} \Leftrightarrow \hspace{5pt} \tr[\rho F]=\tr[\rho' F_\S] \text{ for all } F_\S \in \Eff(\his)_\omega^\R$. We have
      \begin{align*}
         \T(\his)_\omega^\R &\cong (\Y^\R_\omega)_*(\T(\his)) \subseteq \T(\his), \\
        \S(\hi)_\omega^\R &\cong (\Y^\R_\omega)_*(\S(\his)) \subseteq \S(\his).
     \end{align*}
     The $\omega$-conditioned relative states take the form
     \[
        \rho^{(\omega)} := (\Y^\R_\omega)_*(\rho) = \int_G d\mu^{\E_\R}_\omega(g) U(g)^*\rho U(g).
     \]
     \item \emph{The $\omega$-lifted relative states}: global states in $\S(\hirs)_G$ that arise by attaching a frame state $\omega$ to a relative state via the lifting map $\mathcal{L}^\R_\omega$ defined as
     \begin{align*}
        \Gamma^\R_\omega &:= \Y^\R \circ \Gamma_\omega: B(\hirs)^G \to B(\his)^\R, \\
        \mathcal{L}^\R_\omega &:= (\Gamma^\R_\omega)_*: \T(\his)^\R \to \T(\hirs)_G.
     \end{align*}
     The $\omega$-lifting map acts on states as
     \[
        \S(\his)^\R \ni \Omega^\R \mapsto \mathcal{L}^\R_\omega(\Omega^\R) = [\omega \otimes \Omega^\R]_G \in \S(\hirs)_G.
     \]
     \item \emph{Framed relative operators/effects}: subspace of $B(\hi_1 \otimes \hi_2 \otimes \his)^G$/subset of $\Eff(\hi_1 \otimes \hi_2 \otimes \his)^G$ consisting of relativized framed effects, e.g.
     \begin{align*}
        \Eff(\hi_2 \otimes \his)^{\R_1}_{\E_2} &:= \Y^{\R_1}\left(\Eff(\his)_{\E_2}\right) \\
        &= \{\Y^{\R_1}(\E_2(X) \otimes F_\S) \hspace{3pt} | \hspace{3pt} X \in \mathcal{F}(\Sigma_2), F_\S \in \Eff(\his)\}, \\
        B(\hi_2 \otimes \his)^{\R_1}_{\E_2} &:= \rm{span}\left\{\Eff(\hi_2 \otimes \his)^{\R_1}_{\E_2}\right\}^{cl} \subseteq B(\hi_1 \otimes \hi_2 \otimes \his)^G.
     \end{align*}
     
     \item \emph{Framed relative states/trace-class operators}: the quotient space of classes of states/trace-class operators that can not be distinguished by the relativized framed effects, e.g.
     \begin{align*}
         \T(\hi_2 \otimes \his)^{\R_1}_{\E_2} &:= \T(\hi_1 \otimes \hi_2 \otimes \his)/\hspace{-3pt}\sim_{\Eff(\hi_2 \otimes \his)^{\R_1}_{\E_2}}, \\
         &= \T(\hi_2 \otimes \his)^{\R_1}/\hspace{-3pt}\sim_{\E_2},\\
         \S(\hi_2 \otimes \his)^{\R_1}_{\E_2} &:= \S(\hi_1 \otimes \hi_2 \otimes \his)/\hspace{-3pt}\sim_{\Eff(\hi_2 \otimes \his)^{\R_1}_{\E_2}}, \\
         &= \S(\hi_2 \otimes \his)^{\R_1}/\hspace{-3pt}\sim_{\E_2},
     \end{align*}
     where $\E_2$ denotes the equivalence with respect to $\Eff(\his)_{\E_2}$. For the corresponding quotient projections we write e.g.
     \[
        \pi_{\E_2}: \S(\hi_2 \otimes \his) \supseteq \S(\hi_2 \otimes \his)^{\R_1} \to \S(\hi_2 \otimes \his)^{\R_1}/\hspace{-3pt}\sim_{\E_2}.
     \]
\end{itemize}

\section{Quantum mechanics of operators}\label{basics}
Here we provide a semi-didactic introduction to the operator-algebraic perspective on the quantum mechanics of von Neumann. In this setup, a quantum system $\S$ is modeled on a separable Hilbert space $\his$, as in most standard textbook presentations, but the focus is shifted from the Hilbert space vectors to the space $B(\his)$ of bounded operators acting on them. As it turns out, nothing really hinges on those vectors. On the contrary, the whole setup can be phrased in terms of the elements of $B(\his)$, and indeed this perspective seems more appropriate e.g. when we want to view Quantum Theory as an extension of Probability Theory, or when attempting to reconstruct the framework from principles. Since the material presented in this Appendix is standard, we do not cite any particular sources, referring an interested reader to \cite{landsman2017foundations} and references there.

\subsection*{Banach and Hilbert spaces}

The concept of an $n$-dimensional Euclidean space $\mathbb{R}^n$ has been generalized in various ways. One of the least restrictive notions is that of a vector space.

\begin{definition}
    A \emph{vector space} is a set $V$, elements of which will be called vectors, that can be added together and multiplied by numbers. Formally the set $V$ comes equipped with the following operations
    \begin{itemize}
        \item (addition) $+: V \times V \ni (v_1,v_2) \mapsto v_1 + v_2 \in V$,
        \item (scalar multiplication) $\cdot: k \times V \ni (\lambda,v) \mapsto \lambda\cdot v \equiv \lambda v \in V$,
    \end{itemize}
    where $k$ is a field. We will only use the field of real $k=\mathbb{R}$ or complex $k=\mathbb{C}$ numbers.
\end{definition}

This formalizes and generalizes the operations $(x_1,y_1,z_1) + (x_2,y_2,z_2) = (x_1+x_2,y_1+y_2,z_1+z_2)$ and $\lambda \cdot (x,y,z) = (\lambda x,\lambda y.\lambda z)$ in $\mathbb{C}^n$. As a next step, the notion of a \emph{length} of a vector can be introduced on a vector space.

\begin{definition}
A \emph{normed space} is a vector space $V$ equipped with a norm, which is a map $V \ni v \mapsto ||v|| \in [0,\infty)$ satisfying the following conditions
\begin{enumerate}
    \item The norm of a vector is zero only if it is the zero vector, i.e. $||v||=0 \Rightarrow v=0$,
    \item The norm scales with the module under multiplication by numbers, i.e. for any number $\lambda$ and any vector $v$ we have $||\lambda v|| = |\lambda|||v||$, and
    \item The norm satisfies the \emph{triangle inequality}, i.e. for any pair of vectors $v_1,v_2$ we have $||v_1 + v_2|| \leq ||v_1|| + ||v_2||$.
\end{enumerate}
\end{definition}

The notion of a norm  generalizes the map $v = (x,y,z) \mapsto \sqrt{x^2+y^2+z^2}$ on $\mathbb{R}^3$, when the triangle inequality is the one thought in school. A norm naturally induces the notion of a distance $d(v_1,v_2) := ||v_2-v_1||$ between vectors in $V$, which is symmetric and positive and makes $(V,||\_||)$ a \emph{metric space}.% and a \emph{topological vector space}. For completeness, we also recall these definitions.

\begin{definition}
    A \emph{metric space} is a set $M$ equipped with a \emph{metric}, or a \emph{distance function} $d:M \times M \to [0,\infty)$ satisfying the following properties for any $p,q,r \in M$
    \begin{itemize}
        \item The distance of a point from itself is always zero, i.e. $d(p,p)=0$,
        \item The distance is symmetric, i.e. $d(p,q) = d(q,p)$, and
        \item the triangle inequality holds, i.e. $d(p,q) \leq d(p,r) + d(r,q)$.
    \end{itemize}
\end{definition}

As is easily seen, the two triangle inequalities are in agreement for normed vector spaces. In a Euclidean space, the distance function is the one that we use for calculating distances using the Pythagorean theorem, e.g. in three dimensions we have $d(v_1,v_2) = ||v_2-v_1|| = \sqrt{(x_2-x_1)^2+(y_2-y_1)^2 + (z_2-z_1)^2}$. In a metric space, we have an important notion of a \emph{limit} of a \emph{sequence} of elements, i.e. of a collection of points labeled by the natural numbers, written $\{p_n\}_{n \in \mathbb{N}}=\{p_0,p_1,p_2,\dots\}$.

\begin{definition}
    A \emph{limit} of a sequence of elements of a metric space is an element with the property that for large enough indices all the elements of the sequence are arbitrarily close to it. Formally we say that $p_\infty$ is a limit of $\{p_0,p_1,p_2,\dots\}$ if for any $\epsilon > 0$ there is $N_\epsilon \in \mathbb{N}$ such that for any element $p_m$ with $m > N_\epsilon$ the distance $d(p_\infty,p_m)$ is smaller than $\epsilon$. The element $p_\infty$ is denoted by $\lim_{n\to \infty}p_n$.
\end{definition}

Clearly, not every sequence will have a limit, as the elements of a sequence can be distributed arbitrarily. Sequences that do have a limit -- if there is one, it needs to be unique -- are called \emph{converging}. There is a class of sequences, called \emph{Cauchy sequences}, that we would expect to converge.

\begin{definition}
    A sequence of elements in a metric space is a \emph{Cauchy sequence} if, for big enough indices, the distance between any pair of elements is arbitrarily small. Formally, we say that $\{p_0,p_1,p_2,\dots\}$ is Cauchy if for any $\epsilon > 0$ there is $N_\epsilon \in \mathbb{N}$ such that for any $n,m > N_\epsilon$ we have $d(p_n,p_m) < \epsilon$.
\end{definition}

Intuitively if we have a Cauchy sequence that does not converge, the metric space can be seen to be lacking this limiting point -- the elements of the sequence get arbitrarily close to one another as the indices rise, yet there is no point in the metric space that would be arbitrarily close to all of the sufficiently highly labeled elements. Hence the following definition.

\begin{definition}
    A metric space $M$ is \emph{complete} if all the Cauchy sequences converge.
\end{definition}

As a simple example to have in mind consider the following. The real numbers, equipped with the absolute-value norm and the associated distance function, are complete metric spaces, while the rational numbers are not. Indeed, the real numbers can be seen as a \emph{completion} of the rational numbers in the sense of adding all the missing limits of the Cauchy sequences. In such a situation, we say that the set of rational numbers is (norm) \emph{dense} in the set of real numbers. As we often would like to consider limiting procedures we will require all the metric spaces to be complete.

Now going back to the vector spaces, we present the two most important, at least in terms of applications, infinite-dimensional generalizations of Euclidean spaces.

\begin{definition}
    A \emph{Banach space} is a complete normed vector space.
\end{definition}

Juts like the notion of a norm generalizes the length of a vector in an Euclidean space, that of \emph{inner product} provides a generalization of the notion of an \emph{angle}.

\begin{definition}
    An \emph{inner product} on a vector space $V$ is a map $\langle,\rangle: V \times V \to k$ such that
    \begin{itemize}
        \item The inner product is conjugate-symmetric, i.e. for any pair of vectors we have $\langle v_1,v_2\rangle = \overline{\langle v_2,v_1\rangle}$,
        \item The inner product is linear in the first argument, i.e. for any pair of vectors we have $\langle \lambda_1 v_1 + \lambda_2 v_2,v_3\rangle = \lambda_1 \langle v_1,v_3\rangle + \lambda_2 \langle v_2,v_3\rangle$, and
        \item The inner-product of a vector with itself is a non-negative (real) number, zero only for the zero vector, i.e. $\langle v,v \rangle \geq 0$, $\langle v,v \rangle = 0 \Rightarrow v=0$
    \end{itemize}
\end{definition}

A few comments are in order. The first requirement -- that of conjugate-symmetricity -- assures that $\langle v,v\rangle$ is a real number. It also follows that the inner product can not be linear in the second argument (the choice of the first or second slot to be linear is conventional), as it needs to be anti-linear to satisfy the first requirement. Such a map is often called a \emph{sesquilinear form}. In a $3$-dimensional real Euclidean space the inner product is given by $\langle v_1,v_2\rangle = x_1x_2 + y_1y_2 + z_1z_2 = ||v_1||||v_2||cos(\theta)$, where $\theta$ is the angle between the two vectors and $||\_||$ the usual norm. The notions of an inner product and the norm are thus related by $\langle v,v \rangle = ||v||^2$. In fact, this is a general feature -- having fixed an inner product on a vector space, this formula \emph{defines} an associated norm. We then naturally arrive at the notion of a Hilbert space.

\begin{definition}
    A Hilbert space is a vector space equipped with an inner product that is complete in the associated norm.
\end{definition}

Thus the Hilbert spaces can be seen as special kinds of Banach spaces -- those complete normed spaces for which the norm comes from an inner product. The simplest non-trivial complex Hilbert space is the space of complex numbers themselves, with the inner product given by $\langle z,w \rangle = z \overline{w}$, which generalizes to arbitrary dimension $\mathbb{C}^n$ by $\langle z,w \rangle = z_1\overline{w_1} + z_2\overline{w_2} + \dots z_n\overline{w_n}$. 

A pair of Hilbert spaces $\hi_1,\hi_2$ can be considered the same if there is a linear map $T:\hi_1~\to~\hi_2$ that is bijective and preserves the inner product, i.e. we have $\langle T\xi,T\eta\rangle_{\hi_2} = \langle \xi,\eta\rangle_{\hi_1}$ for any pair of vectors $\xi,\eta \in \hi_1$. Such a map is called \emph{unitary}, and if there is one, $\hi_1$ and $\hi_2$ are called \emph{unitarily equivalent}. The spaces $\mathbb{C}^n$ are, up to unitary equivalence, the only examples of finite-dimensional complex Hilbert spaces. We distinguish a class of Hilbert spaces that are particularly tractable and traditionally used in quantum mechanics.

\begin{definition}
    A Hilbert space $\his$ is called \emph{separable} if it contains a countable subset of vectors that is dense in $\his$.
\end{definition}

As is a common practice we will only consider separable Hilbert spaces. Assuming the Axiom of Choice\footnote{Following an accurate description from Wikipedia: ``the Axiom of Choice is an axiom of set theory equivalent to the statement that a Cartesian product of a collection of non-empty sets is non-empty. Informally put, the Axiom of Choice says that given any collection of sets, each containing at least one element, it is possible to construct a new set by arbitrarily \emph{choosing} one element from each set, even if the collection of sets is infinite.} one can show that a Hilbert space is separable if and only if it admits a countable \emph{orthonormal basis}, defined as follows.

\begin{definition}
    An \emph{orthonormal basis} of a Hilbert space is a set of orthogonal vectors of norm one that span the Hilbert space. More precisely, it is a subset $F \subset \his$ such that for any pair of vectors $f_1,f_2 \in F$ we have $\langle f,f' \rangle = 0$ unless $f=f'$ in which case we have $\langle f,f' \rangle = ||f||^2=1$, with the property that for any vector $\xi \in \his$ whenever $\langle f, \xi \rangle = 0$ for all $f \in F$ we can conclude that $\xi = 0$.
\end{definition}

This generalizes the cartesian coordinates that are commonly used in Euclidean spaces. Indeed, given a countable orthonormal basis $\{\phi_0,\phi_1,\dots\}$, any Hilbert space vector $\xi \in \his$ can be written as $\sum_{n=0}^\infty \langle \phi_n,\xi\rangle \phi_n$.
%As noted before, the rational numbers provide a dense subset of the reals, which is countable making $\mathbb{R}$ a real separable Hilbert space. Similarly one can show that all $\mathbb{C}^n$ spaces are complex separable Hilbert spaces.
Up to unitary equivalence, there is only one complex infinite-dimensional \emph{separable} Hilbert space -- it can be constructed as a space of infinite sequences of complex numbers $\{\lambda_0,\lambda_1,\dots\}$ for which the sum $\sum_{n=1}^\infty |\lambda_n|^2$ is finite, with the inner product defined as $\langle \lambda, \lambda' \rangle = \sum_{n=1}^\infty \lambda_n \overline{\lambda'_n}$, and is denoted~by~$l^2(\mathbb{N})$.

\subsection*{Bounded operators}

We can now turn to the main object of interest in the operator-algebraic approach to quantum mechanics -- the algebras of operators. An linear operator on a Hilbert space $\his$ is a linear map $A: \his \ni \xi \mapsto A\xi \in \his$. The space of linear operators on $\his$, denoted $\mathcal{L}(\his)$, is a \emph{vector space} under the operations $(A_1+A_2) = A_1\xi + A_2\xi$ and $(\lambda A)\xi = \lambda A \xi$. It can be equipped with the norm\footnote{Notice here, that this norm can be put on the space of operators on any Banach space, not necessary a Hilbert space.} defined by $||A|| = sup_{\xi \in \his}\frac{||A\xi||}{||\xi||}$. The operators for which this norm is finite are called \emph{bounded} and form a \emph{Banach space} denoted $B(\his)$. This Banach space is very rich in additional structures that we now briefly summarise.

As maps, operators on $\his$ can be composed, i.e. $AB \equiv A \circ B: \xi \mapsto A(B\xi)$. Further, since the composition of two bounded operators remains bounded, we have an \emph{algebra} structure on $B(\his)$, i.e. a binary operation $B(\his) \times B(\his) \ni (A,B) \mapsto AB \in B(\his)$. Moreover, the inner product in the underlying Hilbert space brings about an operation, called \emph{involution}, that can be performed on a bounded operator returning another one. Namely, it turns out that given $A \in B(\his)$ there is a unique $A^* \in B(\his)$ satisfying $\langle A\xi,\eta \rangle = \langle \xi,A^*\eta \rangle$ for all pairs of vectors $\xi,\eta \in \his$, with $A^*$ referred to as an \emph{adjoint} operator of $A$.

It is useful to distinguish special classes of bounded operators on a Hilbert space.

\begin{definition}\label{def:bdops}
    A bounded operator $A \in B(\his)$ is:
    \begin{enumerate}
        \item \emph{unitary} if it preserves the inner product, i.e. $\langle A\xi,A\eta \rangle = \langle \xi,\eta\rangle$,
        \item \emph{self-adjoint} if it equals its adjoint, i.e. $A\xi=A^*\xi$ for all $\xi \in \his$,
        \item \emph{positive} if there is $B \in B(\his)$ such that $A=B^*B$,
        \item a \emph{projection} if it is self-adjoint and idempotent, i.e. $A=A^*=A^2$,
        \item \emph{trace-class} if its trace, defined as $\tr[A] = \sum_{n=1}^\infty \langle A f_n,f_n \rangle$ for $\{f_n\}$ an orthonormal basis, is finite.
    \end{enumerate}
\end{definition}

As a simple example to have in mind take $\mathbb{C}$ as the simplest non-trivial Hilbert space and $B(\mathbb{C}) \cong \mathbb{C} $ -- complex numbers are precisely the linear bounded maps $\mathbb{C} \to \mathbb{C}$, with the complex conjugation providing the involution, and multiplication given by the usual multiplication of complex numbers. Another simple example is the space $B(\mathbb{C}^n)$ of bounded operators on the $n$-dimensional complex vector space~$\mathbb{C}^n$, so all $n\times n$ complex matrices since all linear operators are matrices and all matrices are bounded as operators on a finite dimensional space. In this case, the hermitian adjoint provides the involution operation, matrix multiplication the algebra operation and entry-wise scaling the scalar multiplication. A few comments are in order. 

\begin{enumerate}
    \item As easily seen, an operator $A$ is unitary if and only if it is invertible and $A^*=A^{-1}$, so that $A^*A=AA^*=\mathbb{1}_{\his}$. This generalizes the complex numbers of the form $e^{i\phi}$. The unitary operators thus form a group, generalizing the circle group $S^1$.
    \item The defining property of a self-adjoint operator is analogous to that of a real number among the complex ones -- $z \in \mathbb{C}$ is real iff $\overline{z}=z$.
    \item Any positive number can be written as $z\overline{z}$, which is generalized in point $3.$ above. Positive operators are then automatically self-adjoint. If $A$ is positive, for any $\eta \in \his$ we have $\langle \eta, A\eta \rangle = \langle B \eta, B \eta \rangle = ||B\eta||^2 \geq 0$, which is in fact equivalent -- $A$ is positive iff $\langle \eta, A\eta \rangle \geq 0$ for any $\eta \in \his$.
    \item One can show that there is a one-to-one correspondence between subspaces of the Hilbert space $\his$, i.e. subsets $\mathcal{K} \subseteq \his$ that are Hilbert spaces themselves, and projections as defined in the point $4.$ above. For any subspace, we have a unique projection that maps onto it and does not change the vectors that are already on that subspace. One easily verifies that such maps are indeed self-adjoint and idempotent.
    \item The definition of the trace in $5.$ generalizes that of a trace of a matrix and can do not depend on the choice of an orthonormal basis in a separable Hilbert space. The space of trace-class operators will be denoted by $\T(\his)$, and for its elements, we will usually use the letter $\Omega \in \T(\his)$.
\end{enumerate}

\subsection*{Elementary quantum mechanics}

In the textbook quantum mechanics (pure) states are defined as equivalence classes of unit vectors in a Hilbert space $\his$, so-called rays, written as kets $\ket{\xi}$. The two unit vectors $\xi,\xi'$ define the same state if they are related by a phase, i.e. $\xi' = e^{i\varphi}\xi$ for some $\phi \in [0,2\pi)$. Such states define functionals on $B(\his)$ via associating the \emph{expectation value} to an operator, i.e. $A \mapsto \langle \xi, A \xi \rangle$, which is well defined by the class $\ket{\xi}$ since $\langle e^{i\varphi}\xi, A  e^{i\varphi}\xi\rangle = \langle \xi, A \xi \rangle$ for all $\phi \in [0,2\pi)$ and $A \in B(\his)$. This can be seen as the first instance \emph{operational equivalence} that we encounter -- in principle, any unit vector in $\his$ defines a state understood as a functional that associates expectation values to bounded operators. We however want to identify those vectors that are not distinguishable as such functional. Those that give the same expectation values on all operators can be called "operationally equivalent". This notion plays a prominent role in the presented framework and is discussed separately in the main part of the thesis -- see \ref{sec:opeq}.

The quantum observables in the textbook setting are defined as self-adjoint operators in $B(\his)^{\rm{sa}}$. This assures that the expectation values they give will always be real numbers. Moreover, it is (often implicitly) assumed, that the relevant observables $A\in B(\his)^{\rm{sa}}$ define an orthonormal basis of \emph{eigenvectors}, i.e. the vectors satisfying $A\phi_n=\lambda_n \phi_n$ with $\lambda_n \in \mathbb{C}$, with the corresponding numbers $\lambda_n$ called \emph{eigenvalues}. This is always true when the dimension of $\his$ is finite but otherwise holds only for compact\footnote{An operator on a Hilbert space is compact iff it is a limit (in the operator norm) of a sequence of operators of finite-rank, i.e. such that have an finite-dimensional image. This is an interesting technical condition that we will not need to use.} self-adjoint operators, as in general, a self-adjoint operator may have no eigenvectors. We will go back to this issue shortly, assuming for now that such a basis is given for the operator in question. Then $A$ can be written as $A = \sum_{n=0}^\infty \lambda_n \P_{\lambda_n}$, where $P_{\lambda_n}$ is a projection operator projecting onto the subspace spanned by the eigenvectors with the eigenvalue $\lambda_n$. The expectation values the eigenvectors give as functionals when evaluated on $A$ are then simply the corresponding eigenvalues, i.e. $\langle \phi_n, A\phi_n \rangle = \lambda_n$. But since $\{\lambda_0,\lambda_1,\dots\}$ form an orthonormal basis, any vector $\xi \in \his$ can be written as $\sum_{n=0}^{\infty}\langle \phi_n, \xi \rangle \phi_n$ and thus the expectation value it gives when evaluated on $A$ reads $\langle \xi,A\xi\rangle = \sum_{n=0}^\infty\langle \phi_n, \xi \rangle \lambda_n$. The normalization condition, i.e. the norm equals one, of the states allows for a probabilistic interpretation of the procedure just described along the following lines. The set of eigenvalues of $A$, i.e. the real numbers $\{\lambda_0,\lambda_1,\dots\}$ is understood as possible outcomes of the measurement of the quantity modeled by the self-adjoint operator $A$. The probability of getting the outcome $\lambda_n$ if the system \emph{is prepared in a state $\ket{\xi}$} is given by the \emph{Born rule}
\[
p_{\ket{\xi}}(\lambda_n) := \langle \xi,P_{\lambda_n}\xi \rangle = |\langle \phi_n,\xi\rangle|^2.
\]
Since the norm of $\xi$ is one, we have\footnote{We need to use convexity of the $|\_|^2: \mathbb{C} \to \mathbb{C}$ function to get the sum inside, which is assured since e.g. the norm is always convex, and so is the $z \mapsto z^2$ function, and hence their compositon.}
\[
\sum_{n=0}^\infty p_{\ket{\xi}}(\lambda_n) = \sum_{n=0}^\infty |\langle \phi_n,\xi \rangle|^2 = ||\xi||^2=1.
\]

The epistemic content of the elementary textbook setup is then that the (pure) states assign probability distributions over the spectra of (compact) self-adjoint operators. As such, we find this highly unsatisfactory as a basis for quantum theory for the following reasons.

\begin{enumerate}
    \item[1)] The pure states do not allow for the description of situations in which the state is a \emph{probabilistic mixture} of states.
    \item[2)] The operators that we would like to use as observables may not be expressible as operators with a spanning set of eigenvalues, not even as bounded ones.
    \item[3)] The mathematical machinery is introduced without any conceptual justification.
\end{enumerate}

Regarding the first point above, imagine a situation in which we would like to assign a state to a system for which there is probability $2/3$ of being prepared in $\ket{\xi}$, and $1/3$ of being prepared in $\ket{\xi'}$. There is no pure state that will reproduce the probability distributions corresponding to this situation -- this is our first worry.

A simple instance of the second worry is the position observable of a particle on the real line that acts on wave functions as $\psi(x) \mapsto x\psi(x)$. There is no representation of this observable as a bounded self-adjoint operator with the continuum of real numbers as eigenvalues.

The third worry is a fact -- in the hundred years after the advent of quantum mechanics we still do not have a complete and compelling justification for the mathematical setup of our arguably most empirically successful theory. Moreover, this setup seems to be in serious trouble when we demand its compatibility with Special Relativity, as approached via Quantum Field Theory -- to date, no interacting QFT in $4$ space-time dimensions has been rigorously constructed.

There are different strategies for avoiding the first two worries, the last one, to our understanding, being an open problem. The perspective we would like to advocate for in this thesis begins with a standard move of generalizing the pure states to density operators, and a slightly less popular but still fairly standard one of generalizing observables to positive operator-valued measures, thus addressing the first two worries above. The third worry should perhaps be addressed by a suitable reconstruction of the framework (or its generalization). We suggest some research directions of this kind in chapter \ref{ch:discuss}.

\subsection*{State spaces}
Here we address the first worry of our list above, introduce the trace-class/bonded operators duality, the related topologies that we will use, and briefly discuss maps between algebras of operators preserving the state spaces.

\subsubsection{Density operators}

The first worry can be phrased mathematically as \emph{``the pure state space is not convex''}, meaning that we do not have a pure state corresponding to a convex combination of pure states like $\{p_i,\ket{\xi}\}_{i=1}^N$ with $p_i > 0$ and $\sum_{i=1}^N p_i = 1$. We can of course add the vectors weighted by their probabilities to get $\ket{\xi} = \ket{\sum_{i=1}^Np_i\xi_i}$ -- this is a \emph{superposed state} -- but the resulting state will not provide the expected probabilities. Indeed, using the notation from the previous section, we write
\[
    p_{\ket{\xi}}(\lambda_n) = |\langle \phi_n,\xi\rangle|^2 = \sum_{i=1}^N p_i^2|\langle \phi_n,\xi_i\rangle|^2 = \sum_{i=1}^Np_i^2p_{\ket{\xi_i}}(\lambda_n) \neq 
    \sum_{i=1}^Np_i p_{\ket{\xi_i}}(\lambda_n).
\]

A way out leads through the following observation. There is $1-1$ correspondence between pure states $\ket{\xi}$ and one one-dimensional projections $P_{\ket{\xi}} = \dyad{\xi}$. The expectation values that $\ket{\xi}$ assigns to operators $A \in B(\his)$ can then be expressed by $\langle \xi,A\xi \rangle = \tr[P_{\ket{\xi}}A]$. If we now instead of taking weighted combinations of vectors in the Hilbert space, do the same with the corresponding projection operators on the algebra $B(\his)$ level, we get the correct probabilities. Indeed, writing $\rho = \sum_{i=1}^N p_i P_{\ket{\xi_i}}$ we have
\[
    \tr[\rho P_{\lambda_n}] = 
    \tr[\sum_{i=1}^N p_i P_{\ket{\xi_i}}P_{\lambda_n}] = \sum_{i=1}^Np_i\tr[P_{\ket{\xi_i}}P_{\lambda_n}] = 
    \sum_{i=1}^N p_i|\langle \xi, \phi_n\rangle|^2 =
    \sum_{i=1}^N p_i p_{\ket{\xi_i}}(\lambda_n).
\]

Further, since the trace of a $1$-dimensional projection equals one, we have
\[
\tr[\rho] = \tr[\sum_{i=1}^N p_i P_{\ket{\xi_i}}] = \sum_{i=1}^N p_i \tr[P_{\ket{\xi_i}}] = \sum_{i=1}^N p_i,
\]
so that the condition $\sum_{i=1}^N p_i = 1$ translates to $\tr[\rho] = 1$. Moreover, positivity of $\rho$ is equivalent to $p_i > 0$ since for any vector $\eta \in \his$ we have
\[
\langle \eta, \rho \eta \rangle = \langle \eta, \sum_{i=1}^N p_i P_{\ket{\xi_i}} \eta \rangle = \sum_{i=1}^N p_i\langle \eta, P_{\ket{\xi_i}} \eta \rangle = \sum_{i=1}^N p_i|\langle \xi_i,\eta|^2\rangle,
\]

which is always non-negative only if all $p_i > 0$. Conversely, since trace-class operators are compact and thus come with an orthonormal basis of eigenvectors, any positive operator of trace one can be diagonalized and written in the form of $\rho = \sum_{n=0}^\infty p_n P_{\ket{\lambda_n}}$, thus representing a convex combination of pure states $\ket{\lambda_n}$. Summarizing, the general notion of a quantum state in the Hilbert space formalism, which allows for convex combinations of pure states, is that of a density operator.

\begin{definition}
    A density operator is a positive operator of trace one. 
\end{definition}

The density operators are then a subset of the trace-class operators. The span of one-dimensional projections is \emph{dense} in the \emph{Banach space} of trace-class operators (see below), which then appear naturally as the Banach space generated by the pure states. Assuming positivity and normalization then allows for the probabilistic interpretation. The set of density operators will be denoted by $\S(\his) \subset \T(\his)$, elements of which will be referred to as quantum states and denoted with small Greek letters.

\subsubsection{States-observables duality}

We will now take a closer look at the space of trace-class operators as it is the home of quantum states. It turns out that, like for positive numbers, we can take a square root of any positive operator, i.e. find an operator that squares to it.\footnote{We refrain from entering the functional calculus as it won't play any role in what follows as we would like to keep the presentation focused.} The space of trace-class operators becomes a Banach space on its own when equipped with the following norm $||A||_{tr} = \tr[\sqrt{A^*A}]$. This does not work for arbitrary operators -- the trace-class property is needed for the trace norm to be finite. Notice here that if $A$ is \emph{positive}, we have $\sqrt{A^*A}=A$ and the trace-norm is just the trace, i.e. $||A||_{tr} = \tr[A]$. Now given an arbitrary operator $A \in B(\his)$ we can define a functional
\[
\tr[\_A]: \T(\his) \ni \Omega \mapsto \tr[\Omega A] \in \mathbb{C}.
\]

This map is easily shown to be linear and bounded in the sense that $\sup_{\rho} \frac{|\tr[\rho A]|}{||\rho||_{tr}} = ||A|| \leq \infty$. The space of all bounded linear functionals on a Banach space $V$ is called the Banach dual space and will be denoted by $V^\bigstar$. It is always a Banach space itself when equipped with the supremum-norm $||\phi||_\infty = \sup_{v \in X} \frac{|\phi(v)|}{||v||}$ for $\phi: V \to \mathbb{C}$ as above. Perhaps surprisingly, one can prove that the map above is the most general kind of such a functional, i.e. that for any bounded linear $\phi: \T(\his) \to \mathbb{C}$ there is $A_\phi \in B(\his)$ such that $\phi(\Omega) = \tr[\Omega A_\phi]$ for any $\Omega \in \T(\his)$. The functional $\phi$ is real-valued iff $A_\phi$ is self-adjoint. The assignment $\phi \mapsto A_\phi$ is linear, bijective, and norm-preserving and thus provides a Banach space isomorphism
\begin{equation}\label{eq:stobsduality}
    \T(\his)^\bigstar \cong B(\his).
\end{equation}

The space $\T(\his)$ is then a \emph{predual} of $B(\his)$. One can show that such predual is unique in the sense that if there was a different Banach space with such property, it needs to be isomorphic to $\T(\his)$. The self-adjoint operators can then be understood as precisely the real-valued bounded linear functionals on $\T(\his)$, perhaps pointing to a special role played by the state spaces in the quantum mechanical formalism. The isomorphism above inspires a definition of convergence of operators that is suited to treating them as functionals on the states.

\begin{definition}
    A sequence of operators $\{A_0,A_1,\dots\}$ is \emph{ultraweakly convergent} if for any trace-class operator $\Omega \in \T(\his)$ the sequence of expectation values $\tr[\Omega A_n]$ converges in $\mathbb{C}$.
\end{definition}

Since $\S(\his)$ spans $\T(\his)$, ultraweak convergence can be equivalently defined with $\rho \in \S(\his)$ instead of $\Omega \in \T(\his)$. We say that a subset $\mathcal{O} \subseteq B(\his)$ is weakly closed if it contains all the limits of weakly converging sequences. This is a stronger requirement than the usual closeness as a metric space. Given $\mathcal{O} \subseteq B(\his)$ we can weakly close it by extending it to contain all the possibly missing limits. The result of this operation will be denoted by $\mathcal{O}^{cl}$. There is a corresponding notion of convergence in the spaces $\T(\his)$ and $\S(\his)$.

\begin{definition}\label{def:opertop}
    A sequence of trace-class operators (states) $\{\Omega_0,\Omega_1,\dots \}$ is \emph{operationally convergent} if for any bounded operator $A \in B(\his)$ the sequence of expectation values $\tr[\Omega_n A]$ converges in $\mathbb{C}$.
\end{definition}

The quantum states form a convex subset $\S(\his) \subseteq \T(\his)$ that is \emph{operationally closed}. The quantum states can also be seen as normalized bounded linear functionals $\varphi: B(\hi) \to \mathbb{C}$, and indeed this is the general definition used in the context of operator algebras. However, not every such $\varphi$ corresponds to a density operator via $\varphi_\omega: B(\hi) \ni A \mapsto \tr[\omega A] \in \mathbb{C}$. Those that do are called \emph{normal states} (on $B(\hi)$), and can be characterized by the following continuity condition:
\begin{definition}
    A state $\varphi: B(\hi) \to \mathbb{C}$ is \emph{normal} if for \emph{any orthogonal} family of projections $\{e_i\}_{i \in I}$ we have
    \[
        \varphi \left(\sum_{i \in I}e_i\right) = \sum_{i \in I} \varphi(e_i).
    \]
\end{definition}
For us, states are thus always \emph{normal} as functionals on $B(\hi)$.

\subsection*{Channels and preduals}

Here we distinguish some properties of bounded linear maps between algebras of operators on Hilbert spaces relevant in the context of state spaces.

\begin{definition}
    A bounded linear map $f: B(\hi_1) \to B(\hi_2)$ is called
    \begin{itemize}
        \item \emph{Unital} if it takes identity to identity, i.e. $f(\mathbb{1}_{\hi_1}) = \mathbb{1}_{\hi_2}$.
        \item \emph{Positive} if it takes positive elements to the alike.
        \item \emph{Completely positive} (CP) if all maps of the form $f \otimes \mathbb{1}_{\mathbb{C}^n}$%, for any $n \in \mathbb{N}$, 
        are positive.
        \item \emph{Normal} if for \emph{any orthogonal} family of projections $\{e_i\}_{i \in I}$ in $B(\hi_1)$ we have
             \[
                f \left(\sum_{i \in I}e_i\right) = \sum_{i \in I} f(e_i).
             \]      
        \item \emph{Quantum channel} if it is unital, completely positive, and normal.
        \item \emph{Unitary channel} if it is given by a unitary operator $U:\hi_2 \to \hi_1$ via
        \[
        B(\hi_1) \ni A \mapsto U^* A U \in B(\hi_2).
        \]
    \end{itemize}
\end{definition}
 
As we see, states are exactly normal maps to complex numbers. Normal maps have an important property of admitting a \emph{predual} map: given $f: B(\hi_1) \to B(\hi_2)$ that is \emph{normal} and a state on $B(\hi_2)$ we can use $f$ to construct a state $\Omega \in \T(\hi_1)$ by simply composing
\[
    f_*: \T(\hi_2) \ni \Omega \mapsto \varphi_\Omega \circ f \in \T(\hi_1),
\]
where we have used the identification of density operators and \emph{normal} (bounded, linear) functionals on $B(\hi)$ mentioned above. Equivalently, we may think of the map $f_*$ as being the unique one making the following true
\[
    \tr[\Omega f(A)] = \tr[f_*(\Omega) A]
\]
for all $A \in B(\hi_1)$ and $\Omega \in \T(\hi_2)$. The map $f$ needs to satisfy the normality condition to map (normal) states to (normal) states. If we want $f_*$ to map states to states, it needs to preserve trace and positivity. As easily seen, $f_*$ preserve trace iff $f$ is unital. Indeed, we have
\[
    \tr[f_*(\Omega)] = \tr[f_*(\Omega) \mathbb{1}_{\S'}] = \tr[\Omega f(\mathbb{1}_{\S'})].
\]
Similarly, $f_*$ preserves positivity iff $f$ does, so that $f_*$ maps states to states iff $f$ is normal, unital and positive. It is a quantum channel iff also any map of the form $f\otimes \mathbb{1}_{\mathbb{C}^n}$ maps states to states. This is understood as making sure that if $\S$ is considered as a part of a bigger system (see below), the trivially extended $f$ map will still preserve states. It turns out, that this is a non-trivial condition. As a final comment we note that unitary channels are indeed quantum channels, and as easily confirmed the only \emph{invertible} ones.

\subsection*{Composite systems}

Given a pair of quantum systems $\S_1$ and $\S_2$, modeled on Hilbert spaces $\hi_1$ and $\hi_2$, a composite system, denoted $\S_1 \otimes \S_2$, is modeled on the \emph{tensor product} Hilbert space $\hi_1 \otimes \hi_2$. This is the Hilbert space spanned by the pairs of basis elements of $\hi_1$ and $\hi_2$, i.e. any element $v \in \hi_1 \otimes \hi_2$ can be written as
\[
    v = \sum_{n,m}\lambda_{n,m}f_n \otimes g_m,
\]
where $\{f_n\}_{n \in \mathbb{N}}$ and $\{g_m\}_{m \in \mathbb{N}}$ form basis in $\hi_1$ and $\hi_2$, respectively. This way any pair of (pure) \emph{states} $(\ket{\xi},\ket{\eta})$ can be considered a (pure) state of the composite system, with the dimension of the Hilbert space of $\S_1 \otimes \S_2$ equal the \emph{product} of the dimensions of $\hi_1$ and $\hi_2$. The full state space is generated from those as before by requiring convexity, yielding $\T(\hi_1 \otimes \hi_2) \cong \T(\hi_1 ) \otimes \T(\hi_2)$, where $\T(\hi_i)$ are treated as vector spaces on their own right. The states in $\S(\hi_1\otimes \hi_2) \subseteq \T(\hi_1 \otimes \hi_2)^{\rm{sa}}$ that can be written $\omega \otimes \rho$ (without the sum) are called \emph{product states}. The duality $[\T(\hi)]^\bigstar \cong [B(\hi)]$ gives
\[
    [\T(\hi_1 \otimes \hi_2)]^\bigstar \cong B(\hi_1) \otimes B(\hi_2).
\]
The natural \emph{inclusions} of algebras, e.g. for $B(\hi_1)$ the map
\[
    i_1: B(\hi_1) \ni A \mapsto A \otimes \mathbb{1}_2 \in  B(\hi_1) \otimes B(\hi_2)
\]
are normal, with the predual maps given by the partial trace, e.g. we have
\[
(i_1)_*: \T(\hi_1 \otimes \hi_2) \ni \Omega \mapsto \tr_2[\Omega] \in \T(\hi_1).
\]
Indeed for product states we get $\tr[(i_1)_*(\omega \otimes \rho)A] = \tr[\omega \otimes \rho A \otimes \mathbb{1}_2] = \tr[\omega A]$.

\subsection*{Positive operator-valued measures}

We now address the second worry of the elementary quantum mechanical setup, namely that it is not clear how to represent physically relevant observables as self-adjoint operators with eigenvalues understood as possible measurement outcomes. The solution that we adapt is a relatively simple one, we also find it conceptually compelling. The idea is that an observable on a quantum system should be primarily thought of as a map that takes quantum states to probability distributions, which is perfectly aligned with the epistemic content of the textbook quantum mechanics.

Let us begin by noting the following. One way of phrasing the famous spectral theorem for bounded self-adjoint operators is that for any $A \in B(\his)^{\rm{sa}}$ there is a closed, bounded (and hence compact), non-empty subset of the real line, called the \emph{spectrum of $A$} and denoted by $\sigma(A)$, such that $A$ can be written as
\begin{equation}\label{eq:spthmsa}
    A_\S = \int_{\sigma(A_\S)}\lambda dP^{A_\S}(\lambda),
\end{equation}
where $P^{A_\S}: \mathcal{B}(\sigma({A_\S})) \to B(\his)$ is the projection-valued measure associated with ${A_\S}$ and the integral is interpreted in terms of weak convergence. This means that $P^{A_\S}$ a map that assigns projections on $\his$ to the Borel subsets of $\sigma({A_\S})$ such that for any $\rho \in \S(\his)$ the map $\mathcal{B}(\sigma({A_\S})) \ni X \mapsto \tr[\rho P^{A_\S}(X)]$ is a probability measure on $\sigma({A_\S})$. It follows that $P^{A_\S}$ is normalized in that $P^{A_\S}(\sigma({A_\S}))=\mathbb{1}_{\his}$. Borel subsets are elements of the smallest $\sigma$-algebra containing all the opens. The measurable subsets $X \in \mathcal{F}(\Sigma_\S)$ represent \emph{propositions} about the system $\S$, and the numbers $p^{\E_\S}_\rho(X) \in [0,1]$ probabilities of these propositions being true given that the system $\S$ has been prepared in the state $\rho \in \S$. This extends the standard setting since the spectrum of a compact operator is the (discrete) set of eigenvalues and the integral above becomes the sum $\sum_{n=0}^\infty \lambda_n P_{\lambda_n}$ that we have seen before. The normalization condition of $P^{A_\S}$ is then equivalent to the fact that the eigenvectors span the Hilbert space $\his$. The equation \ref{eq:spthmsa} provides a $1-1$ correspondence between self-adjoint operators and such projection-valued measures. However, this is not enough, since e.g. the position operator that we mentioned is not bounded. The projection-valued measure perspective on self-adjoint operators is readily generalized as follows.

\begin{definition}\label{def:povms}
    A positive operator-valued measure (POVM) ${\E_\S}$ is a map 
    \[
        \mathcal{F}(\Sigma) \ni X \mapsto {\E_\S}(X) \in B(\his),
    \]
    where $\mathcal{F}(\Sigma)$ is a $\sigma$-algebra of subsets of $\Sigma$, such that for any state $\rho \in \S(\his)$ the assignment
    \[
       X \mapsto p^{\E_\S}_\rho(X) := \tr[\rho{\E_\S}(X)]
    \]
    is a (non-negative, countably additive) normalized measure on $\mathcal{F}(\Sigma)$. The operators ${\E_\S}(X)$ are referred to as the \emph{effects of ${\E_\S}$}. If all the effects of ${\E_\S}$ are projections, it is called a projection-valued measure (PVM), also referred to as \emph{sharp} POVM.
\end{definition}    

When probability measure is understood as a non-negative, countably additive normalized measure on a measurable space, this is the most general form of a quantum observable for which the maps $X^{\E_\S}: \S(\his) \ni \rho \mapsto p^{\E_\S}_\rho(X) \in [0,1]$ are \emph{continuous}. It follows that the effects of ${\E_\S}$ are positive operators of the norm at most one, satisfying ${\E_\S}(X\cup Y) = {\E_\S}(X)+{\E_\S}(Y)$ for disjoint subsets $X\cap Y = \emptyset$ and that we have ${\E_\S}(\Sigma) = \mathbb{1}_{\his}$. We refer to such general POVMs as (quantum) \emph{observables}. 

The set of effects on $\his$, i.e. positive operators of the norm at most one, will be denoted by $\mathcal{E}(\his) = \{\F_\S \in B(\his) | 0 \leq \F_\S \leq \mathbb{1}_{\his}\}$, where we introduced the order on $B(\his)$ given by $A \leq B$ iff $0 \leq B-A$ meaning that $B-A$ is positive. If we are interested only in probability distributions that arise by evaluating POVMs on quantum states, the effects in $\mathcal{E}(\his)$ are the only operators that we should really be concerned about. Such a radical attitude does not however change much since they generate the whole algebra of bounded operators as the span of the effects is dense in $B(\his)$.

Notice that sharp POVMs are more general than bounded self-adjoint operators since they can be defined on non-compact measurable spaces. Indeed, our motivating example of the position observable can be modeled as a sharp POVM on the $\sigma$-algebra of Borel subsets of the reals as
\begin{equation}\label{pozobs}
    \E^Q: \mathcal{B}(\mathbb{R}) \ni Y \mapsto M_{\chi_Y} \in B(L^2(\mathbb{R})),
\end{equation}
where $M_{\chi_Y}$ denotes the operator of multiplication by the characteristic function of the Borel subset $Y$, i.e. $M_{\chi_Y}\psi(x) = \chi_Y(x)\psi(x)$. Thus if the state of a particle is given as a wave function $\psi(x)$, the probability of
measuring its position in the subset $Y \subseteq \mathbb{R}$ is $\langle \psi,\chi_Y\psi \rangle = \int_{Y} |\psi(x)|^2 dx$, as desired. 

The setup of POVMs allows for a great variety of observables. In particular, they can be modeled on arbitrary measurable spaces (and not only on $\mathbb{R}$), which is necessary for the framework presented in this work.

\newpage
\section{Coherent frame-change maps}\label{app:ACT}
Here we review the approach to quantum reference frames and frame-change maps as presented in \cite{de2020quantum}. It is developed on the setup of a collection of identical ideal quantum reference frames. The idea is to pick one of these systems and use it as a reference for the description of all the others. Such a reference system is called an \emph{internal} quantum reference frame since it can be seen as a subsystem of the composite system composed by the whole collection. The notion of primary interest is that of a relative pure state, that is given with respect to a chosen reference system. Given such a relative state a specific prescription for assigning a relative state from the perspective of some other reference system from the collection is proposed and dubbed the `coherent quantum reference frame change'.

To stay rigorous, we will present this approach in the context of a \emph{finite} group $G$. We are concerned with a collection of $N$ identical quantum systems $\{\S_0,\S_1,\dots,\S_{N-1}\}$ each with Hilbert space $L^2(G)$. Lets pick a system and the numeration so that our first quantum reference frame is $\S_0$. The relative state from the perspective of $\S_0$ is a ray in $\hi_0 := \bigotimes_{j=1}^{n-1} L^2(G)_j$, and can be written as
\[
\hi_0 \ni \ket{\psi^0} = \sum_{\alpha^0 \in G^{\times(n-1)}} \lambda_{\alpha^0} \ket{\alpha^0},
%\sum_\alpha\lambda_\alpha\bigotimes_{j=1}^{n-1}\ket{(g_\alpha)^0_j}= \sum_\alpha\lambda_\alpha \ket{(g_\alpha)^0_1}\ket{(g_\alpha)^0_2}\dots\ket{(g_\alpha)^0_{n-1}},
\]

where $\alpha^0$ an $(n-1)$-tuple of elements of $G$, i.e. $\alpha^0 = (g^0_1,g^0_2,\dots,g^0_{n-1})$ and $\ket{\alpha^0}$ is defined as a product state $\ket{\alpha^0}:=\ket{g^0_1}\ket{g^0_2}\dots\ket{g^0_{n-1}} \in \hi_0$ so that $\ket{g^0_j} \in \hi_j \cong L^2(G)$. If we pick a different system, say one with the label $i$, as a reference, the state relative to it is a vector
\[
\hi_i:= \bigotimes_{i \neq j=0}^{n-1} L^2(G)_j \ni \ket{\psi^i} = \sum_{\alpha^i \in G^{\times(n-1)}} \lambda_{\alpha^i} \ket{\alpha^i}
\]

The `coherent QRF change' is then specified as a linear map $U_{0 \to i}: \hi_0 \to \hi_i$. Notice here, that it can be singled out by specifying an action on a basis, which is in bijection with the $(n-1)$-tuples. It is then enough to have a transformation that will map $\alpha^0 \mapsto \alpha^i$. Such a map is proposed and justified on the ground of classical considerations that we now recall.

\subsection*{Classical intuition}

Consider a pair of classical 3D reference frames in the usual physical understanding, perhaps made by three sticks with marks on them, glued together in a perpendicular way. They should not be considered as 'living in' an auxiliary Euclidean (or perhaps affine linear) space - we see them simply as physical systems, rigid bodies. Upon idealization they can be assumed identical. Such a pair of 3D classical reference frames is then always in a relative orientation, given by an element of the Euclidean group: frame $B$ is in an orientation $g \in E$ with respect to $A$ if $g$ represents the congruence, or operation (combination of translation, rotation and reflection), that needs to be preformed on $A$ to align it with $B$. We then write $A \overset{g}{\to} B$. Then clearly we have $B \overset{g^{-1}}{\to} A$, and $A \overset{e}{\to} A$: to align $B$ to $A$ we need to perform the inverse congruence, and each frame is always aligned, i.e. in an orientation $e \in E$, with respect to itself. Now consider a collection of $n$ such classical reference frames $\{A_0,A_1,\dots,A_{n-1}\}$. We can describe all of their relative orientations by picking one of them, say $A_0$, and providing the relative orientations $A_0 \overset{g^0_j}{\to} A_j$. Indeed, when asked for the relative orientation $A_i \overset{g^i_j}{\to} A_j$ of $A_j$ with respect to $A_i$, we can first align $A_j$ with $A_0$ by $(g^0_j)^{-1}$, and then align it with $A_i$ via $g^0_i$. This operation is given by the composition: $g^i_j = g^0_i \circ (g^0_j)^{-1}$. Classically speaking, the situation from the perspective of $A_0$ is then described by $n$ elements of the Euclidean group 
\[
\alpha^0_\T = (g^0_0 = e, 
g^0_1,g^0_2,\dots,g^0_{n-1}).
\]
Such a configuration, when considered from the perspective of the frame $A_i$, will be described by the following sequence of group elements: 
\[
\alpha^i_\T = (g^i_0,g^i_1,\dots,g^i_{n-1}) = ((g^0_i)^{-1}, g^0_1(g^0_i)^{-1},\dots, g^i_0(g^0_i)^{-1}=e,\dots, g^0_{n-1}(g^0_i)^{-1}).
\]
Classically, the procedure of changing the perspective, or the reference frame, from $A_0$ to $A_i$ is then given by acting on the right with $(g^0_i)^{-1}$ on all the group elements in $\alpha^0_\T$, which can be formally written as
\[
\alpha^i_\T := \alpha^0_\T \cdot (g^0_i)^{-1}.
\]

\subsection*{Coherent QRF change}

We now go back to the quantum realm. An analog of the configuration $\alpha^0_\T$ is the pure product state $\ket{\alpha^0_\T} := \ket{e^0_0}\ket{g^0_1}\ket{g^0_2}\dots\ket{g^0_{n-1}} \in \hi_\T$. It is then postulated that such a state will transform to $\ket{\alpha^i_\T}:=\ket{\alpha^0_\T \cdot (g^0_i)^{-1}}$. Employing this analogy is enough to define the map $U_{0 \to i}$. Indeed, given a state $\ket{\alpha^0} = \ket{g^0_1}\ket{g^0_2}\dots\ket{g^0_{n-1}}$ we simply need to extend it to $\ket{\alpha^0_\T} = \ket{e}^0_0\ket{\alpha^0}$, apply the right regular action $U_R(g^0_i): \ket{h} \mapsto \ket{h(g^0_i)^{-1}}$ to each factor like we would be dealing with a classical configuration, and `forget' the resulting $\ket{e}^i_i$ state that is absent from $\ket{\psi^i} \in \hi_i$. We thus have
\begin{equation}\label{ACTomchannel1}
\ket{\alpha^i} = \ket{(g^0_i)^{-1}} \otimes U_R(g^0_i)^{\otimes (n-2)}\ket{g^0_1}\dots \ket{g^0_{i-1}}\ket{g^0_{i+1}}\dots\ket{g^0_{n-1}}.
\end{equation}

More generally, given a state $\ket{\psi^0} = \sum_{\alpha^0} \lambda_{\alpha^0} \ket{\alpha^0}$ we need to apply this procedure to each term in the sum separately -- this is the `coherence' assumption. The map $U_{0 \to i}:\hi_0 \to \hi_i$ is then given by
\begin{equation}\label{ACTomchannel2}
U_{0 \to i} = \sum_{g \in G} \dyad{g^{-1}}{g} \otimes U_R(g)^{\otimes (n-2)},
\end{equation}

where $U_R(g)^{\otimes (n-2)}$acts on $\hi_S :=\bigotimes_{i \neq j} L^2(G)_j$ and $\dyad{g^{-1}}{g}: L^2(G)_i \to L^2(G)_0$.

\end{document}